\documentclass[11pt]{article}

\newlength{\actualtopmargin}
\newlength{\actualsidemargin}
\usepackage[tbtags]{amsmath}
\usepackage{amsthm}
\usepackage{amssymb}
\usepackage{amsfonts}
\usepackage{hyperref}

\theoremstyle{plain}
  \newtheorem{theorem}{Theorem}
  \newtheorem{lemma}[theorem]{Lemma}
  \newtheorem{corollary}[theorem]{Corollary}
  \newtheorem{proposition}[theorem]{Proposition}
  \newtheorem{claim}{Claim}
\theoremstyle{definition}
  \newtheorem{definition}[theorem]{Definition}

\theoremstyle{remark}
  \newtheorem*{remark}{Remark}
\theoremstyle{plain}
  \newtheorem*{theorem*}{Theorem}
  \newtheorem*{lemma*}{Lemma}
  \newtheorem*{corollary*}{Corollary}
  \newtheorem*{proposition*}{Proposition}
  \newtheorem*{claim*}{Claim}

\newenvironment{step}
  {
    \begin{enumerate}

  }
  {\end{enumerate}}

\newenvironment{algorithm*}[1]
  {
    \begin{center}
      \hrulefill\\
      \textbf{#1}
  }
  {
    \vspace{-1\baselineskip}
    \hrulefill
    \end{center}
  }

\newenvironment{protocol*}[1]
  {
    \begin{center}
      \hrulefill\\
      \textbf{#1}
  }
  {
    \vspace{-1\baselineskip}
    \hrulefill
    \end{center}
  }

\newcommand{\bbC}{\mathbb{C}}

\newcommand{\bbN}{\mathbb{N}}

\newcommand{\bbZ}{\mathbb{Z}}

\newcommand{\bfD}{\mathbf{D}}

\newcommand{\calA}{\mathcal{A}}

\newcommand{\calH}{\mathcal{H}}

\newcommand{\calM}{\mathcal{M}}

\newcommand{\calP}{\mathcal{P}}

\newcommand{\calR}{\mathcal{R}}
\newcommand{\calS}{\mathcal{S}}
\newcommand{\calT}{\mathcal{T}}

\newcommand{\calV}{\mathcal{V}}
\newcommand{\calW}{\mathcal{W}}
\newcommand{\calX}{\mathcal{X}}
\newcommand{\calY}{\mathcal{Y}}

\newcommand{\sfA}{\mathsf{A}}
\newcommand{\sfB}{\mathsf{B}}

\newcommand{\sfM}{\mathsf{M}}

\newcommand{\sfP}{\mathsf{P}}
\newcommand{\sfQ}{\mathsf{Q}}
\newcommand{\sfR}{\mathsf{R}}
\newcommand{\sfS}{\mathsf{S}}

\newcommand{\sfV}{\mathsf{V}}

\newcommand{\classfont}{\mathrm}

\newcommand{\NP}{\classfont{NP}}

\newcommand{\BPP}{\classfont{BPP}}

\newcommand{\PP}{\classfont{PP}}

\newcommand{\AzeroPP}{\classfont{A}_0\classfont{PP}}
\newcommand{\SBQP}{\classfont{SBQP}}
\newcommand{\PSPACE}{\classfont{PSPACE}}

\newcommand{\QIP}{\classfont{QIP}}
\newcommand{\QMIP}{\classfont{QMIP}}

\newcommand{\MA}{\classfont{MA}}

\newcommand{\QMA}{\classfont{QMA}}
\newcommand{\BQNP}{\classfont{BQNP}}

\newcommand{\QCMA}{\classfont{QCMA}}
\newcommand{\MQA}{\classfont{MQA}}

\newcommand{\EPRQMA}[1]{\QMA^{#1\text{-}\EPR}}

\newcommand{\PH}{\classfont{PH}}

\newcommand{\defeq}{\stackrel{\mathrm{def}}{=}}

\newcommand{\const}{\mathrm{const}}

\newcommand{\SD}{\mathrm{SD}}

\newcommand{\bra}[1]{\langle #1 \vert}

\newcommand{\ket}[1]{\vert #1 \rangle}

\newcommand{\ketbra}[1]{\vert #1 \rangle \langle #1 \vert}

\newcommand{\braket}[2]{\langle #1 \vert #2 \rangle}

\newcommand{\conjugate}[1]{#1^{\dagger}}

\newcommand{\tr}{\mathrm{tr}}

\newcommand{\tensor}{\otimes}
\newcommand{\norm}[1]{\Vert #1 \Vert}

\newcommand{\bignorm}[1]{\bigl\Vert #1 \bigr\Vert}

\newcommand{\trnorm}[1]{\Vert #1 \Vert_{\tr}}

\newcommand{\bigtrnorm}[1]{\bigl\Vert #1 \bigr\Vert_{\tr}}

\newcommand{\biggtrnorm}[1]{\biggl\Vert #1 \biggr\Vert_{\tr}}

\newcommand{\abs}[1]{\vert #1 \vert}

\newcommand{\bigabs}[1]{\bigl\vert #1 \bigr\vert}
\newcommand{\Bigabs}[1]{\Bigl\vert #1 \Bigr\vert}
\newcommand{\biggabs}[1]{\biggl\vert #1 \biggr\vert}
\newcommand{\Biggabs}[1]{\Biggl\vert #1 \Biggr\vert}
\newcommand{\ceil}[1]{\lceil #1 \rceil}

\newcommand{\floor}[1]{\lfloor #1 \rfloor}

\newcommand{\function}[3]{{#1 \colon #2 \to #3}}
\newcommand{\set}[2]{{\{ #1 \colon #2 \}}}

\newcommand{\init}{\mathrm{init}}

\newcommand{\acc}{\mathrm{acc}}
\newcommand{\rej}{\mathrm{rej}}
\newcommand{\yes}{\mathrm{yes}}
\newcommand{\no}{\mathrm{no}}

\newcommand{\EPR}{\mathrm{EPR}}

\newcommand{\Complex}{\bbC}

\newcommand{\Natural}{\bbN}
\newcommand{\Integers}{\bbZ}
\newcommand{\Nonnegative}{\Integers^{+}}

\newcommand{\Binary}{{\{ 0, 1 \}}}
\newcommand{\Density}{\bfD}

\newcommand{\ignore}[1]{}

\begin{document}

\sloppy

% ---------------------------------------------------------------------------
%   Title page
% ---------------------------------------------------------------------------

\title{\Large
  \textbf{
    Stronger Methods of Making Quantum Interactive Proofs Perfectly Complete
  }\\
}

\author{
  Hirotada Kobayashi\footnotemark[1]\\
  \and
  Fran\c{c}ois Le Gall\footnotemark[2]\\
  \and
  Harumichi Nishimura\footnotemark[3]\\
}

\date{}

\maketitle
\thispagestyle{empty}
\pagestyle{plain}
\setcounter{page}{0}

\renewcommand{\thefootnote}{\fnsymbol{footnote}}

\vspace{-5mm}

\begin{center}
{\large
  \footnotemark[1]%
  Principles of Informatics Research Division\\
  National Institute of Informatics\\
  Tokyo, Japan\\
%%   2-1-2 Hitotsubashi, Chiyoda, Tokyo 101-8430, Japan\\
  [2.5mm]
  \footnotemark[2]%
  Department of Computer Science\\
  Graduate School of Information Science and Technology\\
  The University of Tokyo\\
  Tokyo, Japan\\
%%   7-3-1 Hongo, Bunkyo, Tokyo 113-0033, Japan\\
  [2.5mm]
  \footnotemark[3]%
  Department of Computer Science and Mathematical Informatics\\
  Graduate School of Information Science\\
  Nagoya University\\
  Nagoya, Aichi, Japan
%%   Furo-cho, Chikusa, Nagoya, Aichi 464-8601, Japan
}\\
%% [5mm]
%% {\large 4 May 2013} \\
[8mm]
\end{center}

\renewcommand{\thefootnote}{\arabic{footnote}}

% ---------------------------------------------------------------------------
%   Abstract
% ---------------------------------------------------------------------------

\begin{abstract}
This paper presents stronger methods of achieving perfect completeness in quantum interactive proofs.
First, it is proved that any problem in $\QMA$ has a two-message quantum interactive proof system
of perfect completeness with constant soundness error,
where the verifier has only to send a constant number of halves of EPR pairs.
This in particular implies that
the class~$\QMA$ is necessarily included by
the class~${\QIP_1(2)}$ of problems having two-message quantum interactive proofs
of perfect completeness,
which gives the first nontrivial upper bound for $\QMA$ in terms of quantum interactive proofs.
It is also proved that any problem having an $m$-message quantum interactive proof system
necessarily has an ${(m+1)}$-message quantum interactive proof system
of perfect completeness.
This improves the previous result due to Kitaev and Watrous,
where the resulting system of perfect completeness requires ${m+2}$~messages
if not using the parallelization result.
\end{abstract}

\clearpage

%%% Main Part

% ---------------------------------------------------------------------------
%   Introduction
% ---------------------------------------------------------------------------

\section{Introduction}
\subsection{Background and Motivation}

The classical complexity class~$\MA$ of problems having Merlin-Arthur (MA) proof systems,
first introduced by Babai \cite{Bab85STOC},
is a natural probabilistic generalization of the class~$\NP$.
Informally, in a Merlin-Arthur proof system,
Arthur, a probabilistic polynomial-time verifier,
first receives a message (a witness) from Merlin, an all-powerful but untrustworthy prover,
and then checks with high probability the validity of Merlin's claim
that the common input is a yes-instance of the problem.

Quantum Merlin-Arthur (QMA) proof systems are a generalization of the Merlin-Arthur proof systems to the quantum setting,
whose notion was already discussed at an early stage of quantum computing research 
in a technical report by Knill~\cite{Kni96TR}.
In this setting, Arthur now receives a quantum witness from Merlin
and performs polynomial-time quantum computation to check with high probability
whether the input is a yes-instance or not.
The resulting complexity class is called $\QMA$~\cite{Wat00FOCS}
(originally called $\BQNP$~\cite{Kit99AQIP,KitSheVya02Book}),
and has been central to the development of quantum complexity theory
in that it plays a role similar to that $\NP$ plays in classical computation.

The standard way of defining $\MA$ and $\QMA$ allows two-sided bounded error:
each yes-instance may be wrongly rejected with small probability (completeness error),
while each no-instance may also be wrongly accepted with small probability (soundness error).
If completeness error is zero,
that is, any yes-instance is never wrongly rejected,
the corresponding system is said to have \emph{perfect completeness}.
The versions of $\MA$ and $\QMA$ with perfect completeness are denoted by
$\MA_1$ and $\QMA_1$, respectively.

Classically,
it is known that any Merlin-Arthur proof system that may have two-sided bounded error
can always be modified into another Merlin-Arthur proof system with one-sided bounded error of perfect completeness,
i.e., ${\MA = \MA_1}$ holds~\cite{ZacFur87FSTTCS,GolZuc11LNCS}.
This is a very nice property in that
honest Merlin can always convince Arthur without error
by providing a suitable witness for a yes-instance.
A natural question to ask is whether the same property holds for quantum Merlin-Arthur proof systems as well,
i.e., whether ${\QMA=\QMA_1}$ or not.
This question still remains unsolved after many years of investigations.
Besides its theoretical interest, 
answering this question by the affirmative would lead to many consequences.
In particular, any computational problem complete for the class~$\QMA_1$,
for instance the \textsc{Quantum Satisfiability (QSAT)} problems~\cite{Bra06arXiv},
would immediately become complete for the class~$\QMA$ as well.
This would not only lead to a better understanding of $\QMA$
but also have potentials to significantly simplify and strengthen
a possible quantum version of the celebrated PCP theorem~\cite{AroSaf98JACM,AroLunMotSudSze98JACM}
that many researchers have been trying to establish~\cite{AhaAraLanVaz09STOC,AhaAraLanVaz11FOCS,AhaEld11FOCS},
partly because one-sided error verifications are much easier to treat,
and also because the \textsc{QSAT} problems are more direct quantum analogues of
the \textsc{SAT} problems than the \textsc{Local Hamiltonian} problems
(note that the classical PCP theorem can be viewed as proving the $\NP$-completeness of a special case of the \textsc{3SAT} problem
in which, for every no-instance,
at most a constant fraction of clauses are simultaneously satisfiable).

As a barrier to affirmatively answering the $\QMA$ versus $\QMA_1$ question,
%% As a barrier to proving ${\QMA = \QMA_1}$,
Aaronson~\cite{Aar09QIC} constructed a quantum oracle relative to which 
$\QMA_1$ is a proper subclass of $\QMA$, which means that a ``black-box'' proof %similar to the classical case 
%that quantum Merlin-Arthur protocols can be made one-sided error 
of ${\QMA=\QMA_1}$
% of it
cannot exist. 
Nevertheless, no classical oracle is known that separates $\QMA_1$ from $\QMA$, 
%It then might not be surprising that the 
and the following recent results in some sense step towards an affirmative answer to the question:
Nagaj, Wocjan, and Zhang~\cite{NagWocZha09QIC} showed that perfect completeness is achievable
for a special case of quantum Merlin-Arthur proof systems
in which some real number related to the maximum acceptance probability of a given system
%% (precisely, some real value related to it)
can be exactly expressed with a bit string of polynomial length.
More recently, Jordan, Kobayashi, Nagaj, and Nishimura~\cite{JorKobNagNis12QIC}
proved that the equality holds for quantum Merlin-Arthur proof systems \emph{of classical witness},
that is, ${\QCMA = \QCMA_1}$ (or ${\MQA = \MQA_1}$ in a recently-proposed terminology~\cite{Wat09ECSS,GhaSikUpa13QIC}) holds,
assuming that the circuit of a verifier is exactly implementable with a gate set
in which the Hadamard and any classical reversible transformations are performable without error.
%% for the corresonding classes of two-sided bounded error and one-sided bounded error of perfect completeness.
In particular, the latter result gives evidence that,
if we put some natural assumption on a gate set,
the quantum oracle barrier by Aaronson~\cite{Aar09QIC} may not be an insurmountable obstacle
when seeking the possibility of ${\QMA = \QMA_1}$,
as the arguments in Ref.~\cite{Aar09QIC} also lead to a quantum oracle
that separates $\QCMA_1$ from $\QCMA$.

Quantum Merlin-Arthur proof systems may be viewed as a special case of more general quantum interactive proof systems,
where the verifier and the prover may exchange messages using many rounds of communications. 
In their seminal paper, Kitaev and Watrous~\cite{KitWat00STOC} showed that
perfect completeness is achievable in quantum interactive proof systems.
More precisely, 
with two additional messages,
any quantum interactive proof system
that may involve two-sided bounded error
can be transformed into another quantum interactive proof system
that has one-sided bounded error of perfect completeness.
This in particular implies that ${\QMA \subseteq \QIP_1(3)}$,
where $\QIP_1(3)$ is the class of problems having three-message quantum interactive proof systems of perfect completeness.
Unfortunately, $\QIP_1(3)$ is already so powerful that it includes $\PSPACE$~\cite{Wat03TCS}
(actually, ${\QIP_1(3) = \QIP = \PSPACE}$~\cite{KitWat00STOC,JaiJiUpaWat11JACM},
where $\QIP$ denotes the class of problems having general quantum interactive proofs).
Accordingly, this only gives a weaker result for the upper bound of $\QMA$,
as $\QMA$ is known to be inside $\PP$~\cite{KitWat00STOC,Wat00FOCS,MarWat05CC}
(in fact, a slightly stronger bound~${\QMA \subseteq \AzeroPP = \SBQP}$ is known~\cite{Vya03ECCC,Kup09arXiv}).

% ---------------------------------------------------------------------------
%   Our Results and Their Meaning
% ---------------------------------------------------------------------------

\subsection{Our Results and Their Meaning}

This paper presents new general techniques to transform
quantum interactive proof systems into those of perfect completeness,
which increase the number of messages by just one.
Our first result states that any problem in $\QMA$ has a two-message quantum interactive proof of perfect completeness.
 
\begin{theorem}
${\QMA \subseteq \QIP_1(2)}$.
\label{Theorem: QMA is in QIP_1(2)}
\end{theorem}
Here $\QIP_1(2)$ is the class of problems having two-message quantum interactive proof systems of perfect completeness (with negligible soundness error).
This gives the first nontrivial upper bound of $\QMA$ in terms of quantum interactive proofs,
which has no relation known to the existing upper bound~${\AzeroPP = \SBQP}$.
%% which seems to be incomparable to the known upper bound~${\AzeroPP = \SBQP}$.
Note that the inclusion~${\QMA \subseteq \QIP(2)}$ is indeed trivial
for the two-sided error class~$\QIP(2)$ of two-message quantum interactive proofs,
%% Note that, indeed, the inclusion~${\QMA \subseteq \QIP(2)}$ is trivial
%% for the two-sided error class~$\QIP(2)$ of two-message quantum interactive proofs,
but the inclusion here is by the one-sided error class~$\QIP_1(2)$
and is nontrivial to prove.

In fact, we prove a much stronger result,
which arguably steps towards settling the $\QMA$ versus $\QMA_1$ question. 
Namely, we show that, to achieve perfect completeness with constant soundness error,
the verifier in the two-message quantum interactive proof system has only to send a constant number of halves of EPR pairs to the prover.
Or in other words, any problem in $\QMA$ has a quantum Merlin-Arthur proof system
of perfect completeness with constant soundness error,
in which Arthur and Merlin share a constant number of EPR pairs a priori.
More formally,
let ${\EPRQMA{k}(c,s)}$ denote the class of problems having quantum Merlin-Arthur proof systems
with completeness~$c$ and soundness~$s$,
where Arthur and Merlin initially share $k$~EPR pairs.
Then we have the following containment.

\begin{theorem}
For any constant~${s \in (0,1]}$,
there exists a constant~${k \in \Natural}$
such that
\[
\QMA \subseteq \EPRQMA{k}(1, s).
\]
%In particular,
%${
%  \QMA \subseteq \EPRQMA{2^{114}} \bigl( 1, 1 - \frac{1}{2^{110}} \bigr)
%}$.
\label{Theorem: QMA is in 2^k-EPR-QMA(1,s)}
\end{theorem}\vspace{-6mm}

%Notice that Theorem~\ref{Theorem: QMA is in QIP_1(2)} is an immediate corollary of Theorem~\ref{Theorem: QMA is in 2^k-EPR-QMA(1,s)},
%as one may view QMA proof systems with shared EPR pairs as a special case of
%two-message quantum interactive proofs,
%and the parallel repetition of two-message quantum interactive proofs
%works perfectly~\cite{KitWat00STOC}.

Theorem~\ref{Theorem: QMA is in QIP_1(2)} is an immediate consequence of Theorem~\ref{Theorem: QMA is in 2^k-EPR-QMA(1,s)},
as one may view quantum Merlin-Arthur proof systems with shared EPR pairs as a special case of two-message quantum interactive proofs where
the verifier first generates the EPR pairs and sends halves of them to the prover
(and the parallel repetition of two-message quantum interactive proofs works perfectly~\cite{KitWat00STOC}).
Theorem~\ref{Theorem: QMA is in 2^k-EPR-QMA(1,s)} nevertheless appears to be much stronger than Theorem~\ref{Theorem: QMA is in QIP_1(2)}
since it shows that perfect completeness is achievable
with just one additional message of a very restricted form
(a constant number of halves of EPR pairs).
To see this,
let $\EPRQMA{\const}$ be the class of problems having
quantum Merlin-Arthur proof systems with a constant number of prior shared EPR pairs
that may involve two-sided bounded error,
%% have constant gap between completeness and soundness,
and let $\EPRQMA{\const}_1$ be that
of perfect completeness.
%% and let ${\EPRQMA{\const}(1, \const)}$ be that
%% of perfect completeness and constant soundness error.
Then, indeed, the equality~${\EPRQMA{\const} = \QMA}$ immediately follows
from the result by Beigi, Shor, and Watrous~\cite{BeiShoWat11ToC},
as any quantum Merlin-Arthur proof system with a constant number of prior shared EPR pairs
is a special case of two-message quantum interactive proofs \emph{with short questions}
(i.e., two-message quantum interactive proofs
with the first message consisting of at most logarithmically many qubits).
Therefore, we obtain the following characterization of $\QMA$.
\begin{corollary}
${\EPRQMA{\const}_1 = \EPRQMA{\const} = \QMA}$.
\label{Corollary: QMA = const-EPR-QMA(1,const)}
\end{corollary}
% \begin{theorem}
% ${\EPRQMA{\const}_1 = \EPRQMA{\const} = \QMA}$.
% \label{Theorem: QMA = const-EPR-QMA(1,const)}
% \end{theorem}
This in particular implies that perfect completeness is achievable
for the model of quantum Merlin-Arthur proof systems with a constant number of prior shared EPR pairs,
a model that has computational power equivalent to $\QMA$.
Similar arguments further imply that perfect completeness is achievable
even with the models of quantum Merlin-Arthur proof systems with a logarithmic number of prior shared EPR pairs
and ``short-question'' two-message quantum interactive proof systems,
as both of these have computational power equivalent to $\QMA$.
%% This result simultaneously shows that perfect completeness is achievable
%% for the model of quantum Merlin-Arthur proof systems with a constant number of prior shared EPR pairs and that the computational power of this class is the same as $\QMA$.

The methodology developed in this paper essentially shows that,
in order to obtain the inclusion~${\QMA \subseteq \QMA_1}$
(and thus immediately the equality~${\QMA=\QMA_1}$),
it is sufficient to find a way of eliminating the need for the constant number of shared EPR pairs in our proof system.
In fact, as will be clear with our proof structure,
the constant number of shared EPR pairs are necessary
only for the purpose of forcing a dishonest prover to send a witness that is
close to some maximally entangled state of constant dimensions.
Hence, some suitable procedure that tests if a given state of constant dimensions
is sufficiently entangled or not may replace the shared EPR pairs
to affirmatively answer the $\QMA$ versus $\QMA_1$ question
(if two-sided error is allowed, such a test is possible with quantum state tomography).

For general quantum interactive proof systems,
we further present a method that
makes any quantum interactive proof system perfectly complete
by increasing the number of messages by just one.
This improves the previous result due to Kitaev and Watrous~\cite{KitWat00STOC},
whose construction increases the number of messages by two,
if not using their parallelization result.
More precisely, for the class~$\QIP(m)$ of problems having $m$-message quantum interactive proofs that may involve two-sided bounded error,
and the class~$\QIP_1(m)$ of problems having those of perfect completeness,
we show the following.
\begin{theorem}[informal statement]
For any ${m \geq 2}$,
%% ${\QIP(m) \subseteq \QIP_1(m+1)}$.
\[
\QIP(m) \subseteq \QIP_1(m+1).
\]
\label{Theorem: QIP(m) is in QIP_1(m+1), informal}
\end{theorem}
In fact, if the number of messages in the original system is odd,
our transformation does not increase it at all.
\begin{theorem}[informal statement]
For any odd ${m \geq 3}$,
%% ${\QIP(m) \subseteq \QIP_1(m)}$.
\[
\QIP(m) \subseteq \QIP_1(m).
\]
\label{Theorem: QIP(m) is in QIP_1(m) for odd m, informal}
\end{theorem}
While the inclusions of Theorems~\ref{Theorem: QIP(m) is in QIP_1(m+1), informal}~and~\ref{Theorem: QIP(m) is in QIP_1(m) for odd m, informal} can also be obtained by using the parallelization 
results in Refs.~\cite{KitWat00STOC,KemKobMatVid09CC}, our techniques give a new and arguably more direct way of obtaining 
these results.
Our construction actually works well even in the setting of quantum multi-prover interactive proof systems:
it transforms any quantum $k$-prover interactive proof system
into another quantum $k$-prover interactive proof system of perfect completeness
by increasing the number of turns by just one in general,
and without increasing it when the number of turns in the original system is odd.
This much improves the previous result in Ref.~\cite{KemKobMatVid09CC},
where the construction increases the number~$m$ of turns to $3m$ (i.e., by a factor of three),
again without using their parallelization result.
We refer to Theorems~\ref{Theorem: QIP(m) is in QIP_1(m+1)},~\ref{Theorem: QIP(m) is in QIP_1(m) for odd m},~\ref{Theorem: QMIP(k,m) is in QMIP_1(k,m+1)},~and~\ref{Theorem: QMIP(k,m) is in QMIP_1(k,m) for odd m} in Section~\ref{Section: QIP(m) is in QIP_1(m+1)}
for the precise statements of the results.

% ---------------------------------------------------------------------------
%   Organization of the paper
% ---------------------------------------------------------------------------

% \newpage
\subsection{Organization of This Paper}

Section~\ref{Section: Technical Overview for QMA subseteq const-EPR-QMA(1,const)}
gives a high-level explanation of how Theorem~\ref{Theorem: QMA is in 2^k-EPR-QMA(1,s)} (i.e, the inclusion~${\QMA \subseteq \EPRQMA{\const}_1}$) is proved.
Section~\ref{Section: Overview for QIP(m) subseteq QIP(m+1)} presents an overview of the proof of  Theorem~\ref{Theorem: QIP(m) is in QIP_1(m+1), informal}
(i.e., the inclusion 
${\QIP(m) \subseteq \QIP_1(m+1)}$).
% (i.e., Theorem~\ref{Theorem: QIP(m) is in QIP_1(m+1), informal}) is proved.
%% gives a high-level explanation of how the inclusion~${\QMA \subseteq \EPRQMA{\const}(1,\const)}$ is proved.
Section~\ref{Section: preliminaries} provides basic notions and definitions
that are used in this paper.
Section~\ref{Section: Reflection Procedure} rigorously describes and analyzes
the basic procedure called \textsc{Reflection Procedure},
which is the fundamental technical tool throughout this paper.
Section~\ref{Section: QMA is in QIP_1(2)} then gives a full proof of Theorem~\ref{Theorem: QMA is in 2^k-EPR-QMA(1,s)}.
%that establishes the inclusion~${\QMA \subseteq \EPRQMA{\const}_1}$.
Finally, Section~\ref{Section: QIP(m) is in QIP_1(m+1)} proves the results on
general quantum interactive proofs.

% ---------------------------------------------------------------------------
%   Technical Overview for QMA subseteq const-EPR-QMA(1,const)
% ---------------------------------------------------------------------------

\section{Proof Idea of Theorem~\ref{Theorem: QMA is in 2^k-EPR-QMA(1,s)}}
%% \section{Overview of the Proof of Theorem~\ref{Theorem: QMA is in 2^k-EPR-QMA(1,s)}}
\label{Section: Technical Overview for QMA subseteq const-EPR-QMA(1,const)}

The purpose of this section is to give a high-level description of our construction
that proves Theorem~\ref{Theorem: QMA is in 2^k-EPR-QMA(1,s)} (the inclusion~${\QMA \subseteq \EPRQMA{\const}_1})$.
%% that proves the inclusion~${\QMA \subseteq \EPRQMA{\const}(1,\const)}$.
We first describe the main idea in Subsection~\ref{subsec-mainidea} and a simple protocol for a very special case.
Then we explain in Subsection~\ref{subsec-towards} how to make this simple protocol robust against any cheating strategy, by introducing additional tests.
Finally, in Subsection~\ref{subsec-complete}, we present our complete protocol. 
%and to introduce notations that will be used in the formal exposition of the protocol. 
%We focus only on Theorem~\ref{Theorem: QMA is in 2^k-EPR-QMA(1,s)} in this section. 

% ---------------------------------------------------------------------------
%   Main Idea
% ---------------------------------------------------------------------------

\subsection{Underlying Ideas}
\label{subsec-mainidea}

For an input~$x$, let $V_x$ denote the verifier's quantum circuit in the original $\QMA$ proof system.
The operator~$V_x$ acts on two quantum registers, one register~$\sfA$ corresponding to the verifier's work space and another register~$\sfM$ corresponding to the space that stores the witness from the prover.
Let $p_x$ denote the maximum acceptance probability, over all possible witnesses, of the verification procedure.
From the definition of the class~$\QMA$ one can assume that,
for every yes-instance~$x$ it holds that ${p_x \geq 1/2}$,
and for every no-instance~$x$ it holds that ${p_x \leq 1/3}$.
As pointed out by Marriott~and~Watrous~\cite{MarWat05CC},
the maximum acceptance probability~$p_x$ of $V_x$ over all possible witnesses
is the maximum eigenvalue of the Hermitian operator
\[
M_x = \Pi_\init \conjugate{V_x} \Pi_\acc V_x \Pi_\init,
\]
where $\Pi_\init$ is the projection onto the subspace spanned by states in which
all the qubits in $\sfA$ are in state~$\ket{0}$,
and $\Pi_\acc$ is the projection onto the space spanned by the accepting states.

\paragraph{Reflection Procedure}
The basic idea of our protocol is to simulate a procedure that we call \textsc{Reflection Procedure},
presented in details in Section~\ref{Section: Reflection Procedure}.
Roughly speaking, this procedure is viewed as performing a part of amplitude amplification~\cite{Gro96STOC} on the original verification procedure,
and is quite similar to the so-called quantum rewinding technique~\cite{Wat09SIComp},
the underlying idea of which dates back to the strong amplification method for $\QMA$ due to Marriott~and~Watrous~\cite{MarWat05CC}.
Not surprisingly, our~\textsc{Reflection Procedure} can be analyzed in a way similar to the case of the strong amplification method for $\QMA$ due to Marriott~and~Watrous~\cite{MarWat05CC}.
%% This procedure is similar to the strong amplification method for $\QMA$ due to Marriot and Watrous~\cite{MarWat05CC}
%% and can be roughly described as performing amplitude amplification on the original verification algorithm.
We refer to 
Figure~\ref{Figure: Reflection Procedure-simplified} for a presentation of this procedure specialized to the case of QMA proof systems
(a more general description of the procedure will be given in Figure~\ref{Figure: Reflection Procedure} in Section~\ref{Section: Reflection Procedure}). 

\begin{figure}[t!]
\begin{algorithm*}{\textsc{Reflection Procedure}}
\begin{step}
\item
  Receive a quantum register~$\sfM$.
  Prepare $\ket{0}$ in each of the qubits in a quantum register~$\sfA$.
  Apply $V_x$ to the state in ${(\sfA, \sfM)}$.
%%  Reject if the state of the system does not belong to the subspace corresponding to the projection~$\Pi_\init$,
% % and otherwise apply $V_x$.
\item
  Perform a phase-flip (i.e., multiply $-1$ in phase)
  if the state in ${(\sfA, \sfM)}$ belongs to the subspace corresponding to the projection~$\Pi_\acc$.
%%   if the state of the system belongs to the subspace corresponding to the projection~$\Pi_\acc$.
\item
  Apply $\conjugate{V_x}$ to ${(\sfA, \sfM)}$.
%%   Apply $\conjugate{V_x}$.
\item
  Reject if the state in ${(\sfA, \sfM)}$ belongs to the subspace corresponding to $\Pi_\init$,
  and accept otherwise.
%%   Reject if the state belongs to the subspace corresponding to $\Pi_\init$,
%%   and accept otherwise.
\end{step}
\end{algorithm*}
\caption{The \textsc{Reflection Procedure}  (specialized to the case of QMA proof systems; see Figure \ref{Figure: Reflection Procedure} in Section~\ref{Section: Reflection Procedure} for the most general version of this procedure).}
\label{Figure: Reflection Procedure-simplified}
\end{figure}

The \textsc{Reflection Procedure} has access to the unitary transformation~$V_x$, receives a quantum state in register~$\sfM$, and has the following property: 
\begin{enumerate}
\item
If $M_x$ has an eigenvalue 1/2, then there exists a quantum state in $\sfM$ such that 
the procedure accepts with certainty.
% \textsc{Reflection Procedure} accepts with certainty.
\item
If $M_x$ has no eigenvalue in the interval $(\frac{1}{2}-\varepsilon,\frac{1}{2}+\varepsilon)$,
then for any quantum state in $\sfM$ given, 
the procedure rejects with probability at least ${4\varepsilon^2}$.
% \textsc{Reflection Procedure} rejects with probability at least ${4\varepsilon^2}$.
\end{enumerate}
This procedure would then enable us to transform the original QMA proof system
into another QMA proof system with perfect completeness
if we had exactly ${p_x=1/2}$ for any yes-instance~$x$.
This nice property on the completeness of course does not necessarily hold in general.

We mention that the \textsc{Reflection Procedure} is actually slightly superior to the original quantum rewinding technique
(for the purpose of achieving perfect completeness)
in that it requires just two applications of $V_x$ (more precisely, one application of $V_x$ and one application of $V_x^\dagger$), instead of three.
This property will be crucial for our analysis since the \textsc{Reflection Procedure} will ultimately be applied to a modified version of $V_x$ that cannot be implemented directly by the verifier without the help of the prover.

\paragraph{Simple Protocol when $\boldsymbol{p_x}$ is Known}  
In general, we only know that ${p_x \geq 1/2}$ for a yes-instance.
%% For simplicity we will simply write hereafter $p=p_x$. 
Assume that the verifier can apply the matrix 
\[W_q=
\begin{pmatrix}
\sqrt{1-q} & \sqrt{q}
\\
\sqrt{q} & - \sqrt{1-q}
\end{pmatrix}
\]
acting on one qubit, where $q$ is such that ${0 \leq q \leq 1}$ and ${p_xq = 1/2}$
(the value of $q$ depends of course on the input~$x$).
Then, by performing in parallel the original verification test (which succeeds with probability~$p_x$)
and an additional test that applies $W_q$ on a single qubit in the initial state~$\ket{0}$ and measures it,
we obtain a new verification procedure that accepts the input with probability exactly ${p_xq = 1/2}$
(where the new condition for acceptance is that the original test accepts \emph{and} the additional single qubit contains $1$). 
%In other words, instead of applying the original algorithm $V_x$, we apply $V_x\otimes W_q$.
%the Hermitian operator $\Pi_\init \conjugate{(M_xR_x)} \Pi_\acc (M_xR_x) \Pi_\init$ has maximum eigenvalue exactly $pq=1/2$. 
In particular, such a unitary transformation~$W_q$ always exists for any yes-instance~$x$,
and thus, this could achieve the perfect completeness if the verifier knew the probability~${p_x \geq 1/2}$.

The Hermitian operator corresponding to the case of applying in parallel these two tests can be represented by
\[
M'_x
=
(\Pi_\init \tensor \ketbra{0})
\conjugate{(V_x\otimes W_q)}
(\Pi_\acc \tensor \ketbra{1})
(V_x\otimes W_q)
(\Pi_\init \tensor \ketbra{0}),
\]
which has $1/2$ as an eigenvalue for a yes-instance~$x$.
Moreover, it can be easily shown that, on a negative instance, the eigenvalues of this Hermitian operator are bounded away from $1/2$.
Thus, the \textsc{Reflection Procedure} applied to the new verification test~${V_x \tensor W_q}$
transforms the original system into a perfect completeness system.
This protocol of course works only when the verifier can apply $W_q$.

%Note that, in the special case where the probability $p$ can be expressed classically using a polynomial number of bits, then we can ask the prover to give the value $p$ along with its witness. It is not hard to show that, on NO inputs, one can detect if a cheating prover sends the wrong value. 
%Nevertheless, this is not true in general.    

\paragraph{Reflection Simulation Test and Distillation Procedure}

The main problem with the protocol described above is that the verifier does not know in general the probability~$p_x$,
and is then not able to apply $W_q$.
Informally, our basic idea to overcome this difficulty consists in asking the prover to send, along with the witness~$\ket{w}$ of the original proof system, 
%one quantum state 
%$$
%\ket{\chi_p}=\sqrt{p}\ket{0}+\sqrt{1-p}\ket{1}
%$$
%to the verifier, and also 
the unitary transformation~$W_q$ to the verifier, where ${p_xq = 1/2}$.
Concretely, this is done by asking the prover to send two copies of the \emph{Choi-Jamio{\l}kowski state} associated with $W_q$,
denoted by $\ket{J(W_q)}$ and defined as follows:
\[
\ket{J(W_q)}=
%\frac{1}{\sqrt{2}}(\ket{0}W_q\ket{0}+\ket{0}W_q\ket{1})
(I \tensor W_q) \ket{\Phi^+}
=
\sqrt{1-q}\ket{\Phi^-}+\sqrt{q}\ket{\Psi^+},
%\frac{1}{\sqrt{2}}(\ket{0}(\sqrt{1-q}\ket{0}+\sqrt{q}\ket{1})+\ket{0}(\sqrt{q}\ket{0}-\sqrt{1-q}\ket{1})).
\]
where ${\ket{\Phi^+} = \frac{1}{\sqrt{2}} (\ket{00} + \ket{11})}$, ${\ket{\Phi^-} = \frac{1}{\sqrt{2}} (\ket{00} - \ket{11})}$ and ${\ket{\Psi^+} = \frac{1}{\sqrt{2}} (\ket{01} + \ket{10})}$.
By an analysis similar to the case of quantum teleportation,
one can see that the state~$\ket{J(W_q)}$ can be used to simulate one application of the unitary transformation~$W_q$ to any quantum state of a single qubit in a probabilistic manner;
the application succeeds with probability~$1/4$, and we know whether it succeeds or not. 

Let us denote by $\sfM$ the register that is expected to contain the witness $\ket{w}$,
and by $\sfS_1$, $\sfS_1'$, $\sfS_2$, and $\sfS_2'$ the four single-qubit registers
that altogether are expected to contain the two copies of the Choi-Jamio{\l}kowski state.
On a yes-instance, an (honest) prover will then send the state
\[
\ket{w}_{\sfM} \tensor \ket{J(W_q)}_{(\sfS_1, \sfS'_1)} \tensor \ket{J(W_q)}_{(\sfS_2, \sfS'_2)}.
\]
With this state given, the verifier can simulate the desired QMA system with underlying verification procedure~${V_x\! \tensor\! W_q}$
with success probability~${(1/4)^2=1/16}$
(note that ${\conjugate{W_q} = W_q}$,
and thus, one copy of $\ket{J(W_q)}$ is used to simulate the application of $W_q$,
and another copy of it is used to simulate the application of $\conjugate{W_q}$). 
In case where the simulation fails, the verifier systematically accepts
by giving up the simulation to keep perfect completeness.
%Note that two copies of $\ket{J(W_q)}$ are needed here since, in order to implement \textsc{Reflection Procedure}, 
%two applications of the matrix $V_x\otimes W_q$ are used. 
This is the core idea of the procedure \textsc{Reflection Simulation Test} described in Subsection~\ref{Subsection: Simulating the Reflection Procedure with Choi-Jamiolkowski States},
which is a key building block in our proof of Theorem~\ref{Theorem: QMA is in 2^k-EPR-QMA(1,s)}.
%We will show, in Subsection \ref{Subsection: Simulating Unitaries with Choi-Jamiolkowski States}, how , 
%%using an idea similar to ``gate teleportation.'' 

In fact, we incorporate one more technique called \textsc{Distillation Procedure}, 
which is again based on the analysis of Ref.~\cite{MarWat05CC},
and makes the analysis of our complete protocol significantly easier.
In general, one of the main difficulties when analyzing the soundness
with the simulation of the \textsc{Reflection Procedure} with the associated Hermitian operator~$M'_x$ above is
that one has to care about the entanglement between the witness part in $\sfM$
and the part for the Choi-Jamio{\l}kowski states in $\sfS_1$, $\sfS'_1$, $\sfS_2$, and $\sfS'_2$.
This could make the soundness analysis extremely hard,
and in fact, the authors do not even know if the soundness can be proved without using the \textsc{Distillation Procedure}.
The idea to settle this difficulty is that,
instead of directly simulating the \textsc{Reflection Procedure} above
%% instead of directly performing the \textsc{Reflection Procedure}
%% with the associated Hermitian operator~$M'_x$ above
on a received state
(that is expected to be a product state of a witness~$\ket{w}$ and two copies of the Choi-Jamio{\l}kowski state),
one first performs the \textsc{Distillation Procedure} twice in sequence
on the witness part (i.e., $\sfM$) of the received state
to produce a situation where one can perform a much simplified version of the \textsc{Reflection Procedure}
that does not even need to receive a witness.
This new \textsc{Reflection Procedure} has a very nice property that
it does not significantly change the behavior of the original \textsc{Reflection Procedure},
and its associated Hermitian operator acts over a space of just four dimensions
and has a much simpler form:
\[
(\ketbra{0} \tensor \ketbra{0})
\conjugate{(W_p \tensor W_q)}
(\ketbra{1} \tensor \ketbra{1})
(W_p \tensor W_q)
(\ketbra{0} \tensor \ketbra{0}),
\]
where ${p = p_x^2/(2p_x^2 - 2p_x + 1)}$ and ${q = 1/(2p)}$
(which is different from the value of $q$ in the previous case with $M'_x$).
%is  the acceptance probability of the verifier $V_x$ on witness $\ket{w}$.
%the eigenvalue of $M_x$ corresponding to an eigenstate $\ket{w}$ 
%(i.e., the acceptance probability of the verifier $V_x$ on witness $\ket{w}$).
More precisely, the two applications of the \textsc{Distillation Procedure} (described in Subsection~\ref{Subsection: Encoding Accepting Probability in Phase})
enable us to generate with high probability two identical copies of the single-qubit state
\[
\ket{\chi_p}=\sqrt{1-p}\ket{0}+\sqrt{p}\ket{1}
\]
from a given witness~$\ket{w}$
(and one can know whether the generation of the two copies succeeded or not).
The point is that, if the input were a no-instance,
and the original soundness were very small,
the generated state should be very close to ${\ket{0} \tensor \ket{0}}$,
and could be analyzed as if it were unentangled with other qubits.
Note that one can easily transform $\ket{\chi_p}$ into $\ket{J(W_p)}$,
and thus one essentially obtains the desired two copies of the Choi-Jamio{\l}kowski state corresponding to $W_p$
after the two applications of the \textsc{Distillation Procedure}.

% After this incorporation,  we obtain the procedure \textsc{Reflection Simulation Test} 
% %and described in details in Subsection \ref{Subsection: Simulating the Reflection Procedure with Choi-Jamiolkowski States} 
% that, when given the state  
% $\ket{\chi_p}^{\otimes 2}\otimes \ket{J(W_q)}^{\otimes 2}$ with $pq=1/2$ as input, where $\ket{\chi_p}^{\otimes 2}$ is obtained by \textsc{Distillation Procedure}, 
% simulates the reflection procedure of the previous paragraph 
%  (i.e., simulates \textsc{Reflection Procedure} applied to $W_p\otimes W_q$) and thus accepts with certainty on a yes-instance.
% % Here $\ket{\chi_p}$ is the which can be easily created by the verifier using the quantum circuit $V_x$ (we describe in Subsection~\ref{Subsection: Encoding Accepting Probability in Phase} a procedure called \textsc{Distillation Procedure} for this purpose). 

% ---------------------------------------------------------------------------
%   Towards the Actual Protocol
% ---------------------------------------------------------------------------

\subsection{Towards the Actual Protocol}
\label{subsec-towards}

The main problem of the strategy described in the previous subsection is of course that,
on a no-instance, a dishonest prover may not send the prescribed state.
Actually, for a dishonest prover who sends a state of the form~${\ket{w} \tensor \ket{J(W_q)}^{\tensor 2}}$,
then no matter which state~${\ket{w}}$ and no matter which value~$q$ the prover chooses,
the soundness can be analyzed with a quite straightforward argument.
The real issue lies in the case where a dishonest prover does not send a quantum state of the form~${\ket{w} \tensor \ket{J(W_q)}^{\tensor 2}}$,
and especially when the state in ${(\sfS_1, \sfS'_1, \sfS_2, \sfS'_2)}$
is not a product state of two identical copies of 
a Choi-Jamio{\l}kowski state.

To force a state in ${(\sfS_1, \sfS'_1, \sfS_2, \sfS'_2)}$
to be at least close to a mixture of two-fold products of an identical quantum state (which may be a mixed state),
we modify the protocol so that we can use the finite quantum de Finetti theorem~\cite{KonRen05JMP,ChrKonMitRen07CMP}.  
For this, the verifier now asks the prover to send not only two copies of $\ket{J(W_q)}$
but a larger number of copies of it: $\ket{J(W_q)}^{\tensor N}$ where $N$ is large but still a constant.
The expected witness sent by an honest prover is then 
\[
\ket{w}_{\sfM} \tensor \ket{J(W_q)}_{(\sfS_1, \sfS'_1)} \tensor \cdots \tensor \ket{J(W_q)}_{(\sfS_N, \sfS'_N)}.
%% \ket{w}_{\sfM} \tensor \ket{J(W_q)}_{(\sfS_1, \sfS'_1)} \tensor \ket{J(W_q)}_{(\sfS_2, \sfS'_2)} \tensor \cdots \tensor \ket{J(W_q)}_{(\sfS_N, \sfS'_N)}.
\]
The witness state in ${(\sfM, \sfS_1, \sfS_1',\ldots, \sfS_N, \sfS_N')}$ sent by a prover in a general case
may of course not be of the form above, if the prover is dishonest.
After the two applications of the \textsc{Distillation Procedure} with $\sfM$,
the verifier permutes the $N$ pairs of registers~${(\sfS_1, \sfS_1'), \ldots, (\sfS_N, \sfS_N')}$ uniformly at random.
This makes the state in ${(\sfS_1, \sfS_1'), \ldots, (\sfS_N, \sfS_N')}$
symmetric (i.e., invariant under any permutation of the $N$~pairs of registers~${(\sfS_1, \sfS_1'), \ldots, (\sfS_N, \sfS_N')}$),
and thus the quantum de Finetti theorem guarantees that the reduced state in ${(\sfS_1, \sfS'_1, \sfS_2, \sfS'_2)}$
of the resulting state after random permutation
must be close to some mixture of two-fold product states
\[
\sum_{j} \mu_j \xi_j \tensor \xi_j.
%% \sum_{j} p_j{\xi_j}_{(\sfS_1, \sfS'_1)} \tensor {\xi_j}_{(\sfS_2, \sfS'_2)}.
\]
Note that each state~$\xi_j$ may not necessarily be a pure state,
and is usually a mixed state.
The \textsc{Swap Test}, performed additionally to this random permutation,
will ensure that every $\xi_j$ must be actually close to some pure state.
This is nevertheless not enough:
we want to ensure that each $\xi_j$ is close to some Choi-Jamio{\l}kowski state.
To have this desirable property,
we now assume that each pair of registers~${(\sfS_j,\sfS'_j)}$
initially contains an EPR pair,
and that the verifier initially holds the registers~${\sfS_1, \ldots, \sfS_N}$
and receives only, additionally to $\sfM$, the registers~${\sfS'_1, \ldots, \sfS'_N}$ as witness.
This assumption is the only part where we need (a constant number of) shared EPR pairs,
and removing it is the last obstacle that prevents us from proving the result ${\QMA=\QMA_1}$.
To make use of this assumption,
we further device a test called the \textsc{Space Restriction Test}
that restricts the Hilbert space
corresponding to the registers~${(\sfS_1, \sfS'_1, \sfS_2, \sfS'_2)}$
in which the verifier expects to receive the copies of the Choi-Jamio{\l}kowski state.
The assumption of a constant number of prior-shared EPR pairs is
then tactically used with this \textsc{Space Restriction Test} to finally ensure
that each $\xi_j$ must be close to some legal Choi-Jamio{\l}kowski state. 
%% the entire witness sent from the prover must be close to some legal state of the prescribed form.
%% We then show that the assumption of the shared EPR pairs,
%% combined with some suitable measurements,
%% will finally force each $\xi_j$ to be close to some desirable Choi-Jamio{\l}kowski state. 
%implies that a malicious prover cannot create arbitrary states in each pair of register $(\sfS_i,\sfS'_i)$ and that our tests can detect the case where the $\xi_j$'s are not Choi-Jamio{\l}kowski states.

% ---------------------------------------------------------------------------
%   Complete Protocol
% ---------------------------------------------------------------------------

\subsection{Final Protocol}
\label{subsec-complete}

The final protocol of the verifier in a QMA system of perfect completeness
with a constant number of shared EPR pairs
is given in Figure~\ref{Figure: verifier's EPR-QMA protocol for achieving perfect completeness-simplified}.
Actually, Figure~\ref{Figure: verifier's EPR-QMA protocol for achieving perfect completeness-simplified}
presents a slightly simplified exposition of our final protocol;
the complete description will appear in Section~\ref{Section: QMA is in QIP_1(2)}
(see Figure~\ref{Figure: verifier's EPR-QMA protocol for achieving perfect completeness}
in the proof of Theorem~\ref{Theorem: QMA is in 2^k-EPR-QMA(1,s)}).

\begin{figure}[t!]
\begin{algorithm*}{Verifier's QMA Protocol for Achieving Perfect Completeness with $\boldsymbol{N}$~Prior-Shared EPR Pairs (Simplified)}
\begin{step}
\item
  Store the particles of the shared $N$~EPR pairs in ${(\sfS_1, \ldots, \sfS_N)}$. 
  Receive a quantum witness
  in registers ${(\sfM, \sfS'_1, \ldots, \sfS'_N)}$.
%%   Prepare $\ket{0}$ in each of two single-qubit registers~$\sfR_1$ and $\sfR_2$,
%%   and $\ket{\bmzero}$ in a multi-qubit register~$\sfA$, which corresponds to the private space of the original verifier.
\item
  Execute the \textsc{Distillation Procedure} twice in sequence, both using a state in $\sfM$.
  Accept if any of the two executions fails,
  and continue otherwise, with storing the two generated single-qubit states in $\sfR_1$ and $\sfR_2$.
%%   Execute the \textsc{Distillation Procedure} with ${(\sfR_1, \sfA, \sfM)}$ as input, and then again 
%%   using ${(\sfR_2, \sfA, \sfM)}$ as input. Accept if any of the two procedure fails, and continue otherwise.
\item
  Permute the $N$~pairs of registers~${(\sfS_1, \sfS'_1), \ldots, (\sfS_N,\sfS'_N)}$ uniformly at random.
\item
  Perform the \textsc{Space Restriction Test}.
  That is,
  test if the state in ${(\sfS_j, \sfS'_j)}$ is in the space spanned by $\{ \ket{\Phi^-}, \ket{\Psi^+}\}$, for each ${j \in \{1,2\}}$.
  Reject if not so, and continue otherwise.\\
\item
  Perform the \textsc{Swap Test} between ${(\sfS_1, \sfS'_1)}$ and ${(\sfS_2, \sfS'_2)}$.
  Reject if it fails, and continue otherwise.\\
\item
  Perform the \textsc{Reflection Simulation Test}
  with ${(\sfR_1, \sfR_2, \sfS_1, \sfS'_1, \sfS_2, \sfS'_2)}$ as input.
  Accept if this returns ``accept'', and reject otherwise.
\end{step}
\end{algorithm*}
\caption{Slightly simplified description of the verifier's QMA protocol for achieving perfect completeness with $\boldsymbol{N}$ pre-shared EPR pairs. The complete description appears as Figure~\ref{Figure: verifier's EPR-QMA protocol for achieving perfect completeness} in Section~\ref{Section: QMA is in QIP_1(2)}.}
\label{Figure: verifier's EPR-QMA protocol for achieving perfect completeness-simplified}
\end{figure}

Let us briefly describe the protocol step by step, focusing on what happens when the prover is honest.
At the end of Step~1,
i.e., just after receiving a witness from the prover,
the state in ${(\sfM, \sfS_1, \sfS_1', \ldots, \sfS_N, \sfS_N')}$ is given by
\[
\ket{w}_{\sfM} \tensor \ket{J(W_q)}_{(\sfS_1, \sfS'_1)} \tensor \cdots \tensor \ket{J(W_q)}_{(\sfS_N, \sfS'_N)}.
%% \ket{w}_{\sfM} \tensor \ket{J(W_q)}_{(\sfS_1, \sfS'_1)} \tensor \ket{J(W_q)}_{(\sfS_2, \sfS'_2)} \tensor \cdots \tensor \ket{J(W_q)}_{(\sfS_N, \sfS'_N)}.
\]
When none of the two executions of the \textsc{Distillation Procedure} fails
in Step~2,
the state in ${(\sfR_1,\! \sfR_2,\! \sfS_1,\! \sfS'_1, \ldots,\! \sfS_N,\! \sfS'_N)}$ becomes
%% the state in ${(\sfR_1, \sfR_2, \sfS_1, \sfS'_1, \ldots, \sfS_N, \sfS'_N)}$ becomes
\[
\ket{\chi_p}_{\sfR_1} \tensor \ket{\chi_p}_{\sfR_2}
\tensor
\ket{J(W_q)}_{(\sfS_1, \sfS'_1)}  \tensor \cdots \tensor \ket{J(W_q)}_{(\sfS_N, \sfS'_N)},
\]
at the end of this step.
Step~3 just permutes the $N$ pairs of registers ${(\sfS_1, \sfS'_1), \ldots, (\sfS_N,\sfS'_N)}$ uniformly at random, 
which does not change the state at all.
The \textsc{Space Restriction Test} in Step~4 forces each of the two-qubit states in ${(\sfS_1,\sfS'_1)}$ and ${(\sfS_2,\sfS'_2)}$
to be in the subspace spanned by $\ket{\Phi^-}$ and $\ket{\Psi^+}$
(as the state must be in this subspace if it is a product of the desirable Choi-Jamio{\l}kowski states),
which does not change the state either.
Then the \textsc{Swap Test} in Step~5 never fails,
since the registers~${(\sfS_1,\sfS'_1)}$ and ${(\sfS_2,\sfS'_2)}$ contain
the identical pure state.
Finally, Step~6 performs the \textsc{Reflection Simulation Test},
which must result in acceptance with certainty,
as the value~$q$ was chosen appropriately
so that the associated Hermitian operator with this \textsc{Reflection Simulation Test}
has an eigenvalue exactly $1/2$.

\paragraph{Rough Sketch of Soundness Analysis}
Here we give a very rough sketch of the soundness analysis for a no-instance case.
The rigorous analysis can be found in Section~\ref{Section: QMA is in QIP_1(2)}.

Without loss of generality, it is assumed that
the original QMA system has soundness exponentially close to $0$.
Then,
if none of the two executions of the \textsc{Distillation Procedure} fails,
whatever witness has been received in Step~1,
the state generated in ${(\sfR_1,\sfR_2)}$ after Step~2
must be exponentially close to
\[
\ket{\chi_0}_{\sfR_1} \tensor \ket{\chi_0}_{\sfR_2} =\ket{0}_{\sfR_1}\tensor \ket{0}_{\sfR_2}
\]
(and the probability that the \textsc{Distillation Procedure} fails is actually exponentially small in this case).
This implies that the state in ${(\sfR_1, \sfR_2)}$ is almost unentangled with the state in ${(\sfS_1, \sfS'_1, \ldots, \sfS_N, \sfS'_N)}$.

As the random permutation in Step~3 makes the state in ${(\sfS_1, \sfS'_1, \ldots, \sfS_N, \sfS'_N)}$ symmetric,
from the quantum de Finetti theorem,
the reduced state in ${(\sfR_1, \sfR_2, \sfS_1, \sfS'_1, \sfS_2, \sfS'_2)}$ after Step~3 
must be close to the state of the form
\[
(\ketbra{0})^{\tensor 2} \tensor \biggl( \sum_j \mu_j \xi_j^{\tensor 2} \biggr).
\]
%% \[
%% (\ketbra{0})^{\tensor 2} \tensor \biggl( \sum_j p_j \xi_j \tensor \xi_j \biggr).
%% \]
A key property is that the reduced state in ${(\sfS_1, \sfS_2)}$
%% $\tr_{\sfS'_1\sfS'_2}\left(\sum_j p_j \xi_j\tensor\xi_j \right)$
is exponentially close to the totally mixed state ${(I/2)^{\tensor 2}}$,
which is guaranteed by the facts that each state in $\sfS_j$ for ${j \in \{1, \ldots, N \}}$
was originally a half of the shared EPR pair,
that the two executions of the \textsc{Distillation Procedure}
disturbed the state by an amount at most exponentially small,
and that the state~${(I/2)^{\tensor N}}$ in ${(\sfS_1, \ldots, \sfS_N)}$
is invariant under random permutation.

Now one can show that (stated here informally)
if the probability of rejection is very small in the \textsc{Space Restriction Test} in Step~4
(otherwise the dishonest prover is caught with some reasonable probability in this Step~4),
the state in ${(\sfR_1, \sfR_2, \sfS_1, \sfS'_1, \sfS_2, \sfS'_2)}$ at the end of Step~4 is sufficiently close to a state of the form
\[
(\ketbra{0})^{\tensor 2} \tensor \biggl( \sum_j \mu'_j {\xi'_j}^{\tensor 2} \biggr),
\]
%% \[
%% (\ketbra{0})^{\tensor 2} \tensor \biggl( \sum_j p'_j \xi'_j \tensor \xi'_j \biggr),
%% \]
where each $\xi'_j$ is a mixed state over the Hilbert space spanned by $\ket{\Phi^-}$ and $\ket{\Psi^+}$,
while the \textsc{Swap Test} in Step~5 requires that each $\xi'_j$ must be close to some pure state
(otherwise the dishonest prover is caught with some reasonable probability in this Step~5).

Together with the fact mentioned above that
the reduced state in ${(\sfS_1, \sfS_2)}$
%% $\tr_{\sfS'_1\sfS'_2}\left(\sum_j p_j \xi_j\tensor\xi_j \right)$
was close to the totally mixed state ${(I/2)^{\tensor 2}}$ when entering Step~4,
these two properties finally ensure that
the state in ${(\sfR_1, \sfR_2, \sfS_1, \sfS'_1, \sfS_2, \sfS'_2)}$ at the end of Step~5 must be sufficiently close to a state of the form
\[
(\ketbra{0})^{\tensor 2} \tensor \biggl[ \sum_j \mu''_j \bigl( \ketbra{J(W^\pm_{a_j})} \bigr)^{\tensor 2} \biggr],
\]
where each
$W^\pm_{a_j}$ is equal to either $W_{a_j}$ or $ZW_{a_j}Z$,
with
${
  Z =
  \left(
    \begin{smallmatrix}
      1 & 0\\
      0 & -1
    \end{smallmatrix}
  \right)
}$.
Notice that this is a mixture of desired states and their slightly different variants.

For each state of the form
\[
{\ket{0}^{\tensor 2} \tensor \ket{J(W^\pm_{a_j})}^{\tensor 2} = \ket{\chi_0}^{\tensor 2} \tensor \ket{J(W^\pm_{a_j})}^{\tensor 2}},
\]
however, we can easily show that the \textsc{Reflection Simulation Test} in Step~6 rejects with sufficiently large probability (shown to be exactly $1/16$)
irrelevant to the value~$a_j$,
and thus, the verifier can reject with probability close to $1/16$
even when the verification procedure reaches Step~6 with very high probability.

% ---------------------------------------------------------------------------
%   Technical Overview for QIP(m) subseteq QIP_1(m+1)
% ---------------------------------------------------------------------------

\section{Proof Idea of Theorem \ref{Theorem: QIP(m) is in QIP_1(m+1), informal}}
% \section{Overview of the Proof of Theorem 4}
%\section{Overview for $\boldsymbol{\QIP(m) \subseteq \QIP_1(m+1)}$}
\label{Section: Overview for QIP(m) subseteq QIP(m+1)}

This section gives an overview of the proof of Theorem~\ref{Theorem: QIP(m) is in QIP_1(m+1), informal}
(more precisely, of the formal statement of this result, Theorem~\ref{Theorem: QIP(m) is in QIP_1(m+1)}),
which proves the inclusion~${\QIP(m) \subseteq \QIP_1(m+1)}$,
for each ${m \geq 2}$.
For simplicity, here we assume that the number~$m$ of messages is odd
(the case with even number of messages can be proved with essentially the same argument),
and completeness and soundness are $2/3$ and $1/3$, respectively,
in the original quantum interactive proof system.

% ---------------------------------------------------------------------------
%   Underlying Idea
% ---------------------------------------------------------------------------

% \subsection{Underlying Idea}
% \label{Subsection: Undrelying Idea for QIP(m) subseteq QIP_1(m+1)}

The basic idea is again to simulate the \textsc{Reflection Procedure}
associated with the original $m$-message quantum interactive proof system.

Fix an input~$x$
and the transformations of the prover~$P$ on $x$
in the original $m$-message quantum interactive proof system.
This time, we consider that the register~$\sfM$ in the \textsc{Reflection Procedure} described in Figure~\ref{Figure: Reflection Procedure-simplified}
contains all the qubits the prover~$P$ can access in the original system
(i.e., all the private qubits of the prover and all the message qubits that are used for communications).
We further consider that the register~$\sfA$ contains all the private qubits of the verifier in the original system.
Now, if we replace $V_x$ in Figure~\ref{Figure: Reflection Procedure-simplified} by the unitary transformation~$U$ derived from the original quantum interactive proof system
when the verifier communicates with $P$ on input~$x$,
the \textsc{Reflection Procedure} described in Figure~\ref{Figure: Reflection Procedure-simplified}
can be viewed as first applying $U$ by performing a forward simulation of the communications with $P$,
then applying a phase-flip with respect to the accepting states,
and further applying $\conjugate{U}$ by performing a backward simulation of the communications with $P$
to confirm if the entire state \emph{does not} go back to a legal initial state.

Hence, if there is a strategy for a prover that can convince the verifier with probability exactly $1/2$ in the original system,
then this specific \textsc{Reflection Procedure} with such a prover
must result in acceptance with certainty, from the property of the \textsc{Reflection Procedure}.
Fortunately, if the number~$m$ of messages is at least two,
it is not hard for an all powerful prover to arbitrarily decrease the accepting probability,
and thus, this essentially achieves the perfect completeness when the input is a yes-instance.
On the other hand, for any no-instance,
no prover can convince the verifier with probability more than $1/3$.
This implies that the above specific \textsc{Reflection Procedure}
must result in rejection with some constant probability (actually with probability at least $1/9$),
again from the property of the \textsc{Reflection Procedure}.
Therefore, this basically establishes a quantum interactive proof system of perfect completeness, as desired.

There are two problems in this construction.
One is that a dishonest prover may not be so cooperative that a backward simulation
forms $\conjugate{U}$ as required
(i.e., a prover may behave during the backward simulation differently
from the inverse of what he/she behaved during the forward simulation).
The other is that the number of messages increases from $m$ to ${2m-1}$,
and thus,
it is less communication-efficient
than the existing construction of achieving perfect completeness in quantum interactive proofs due to Kitaev and Watrous~\cite{KitWat00STOC}.

\paragraph{Modified Reflection Procedure}

Both of the two problems mentioned above originate from the fact that
the \textsc{Reflection Procedure}
involves one application of $U$ and one application of $\conjugate{U}$.
Now we modify the procedure so that it involves one application of $\conjugate{U}$ only (and no application of $U$ is required),
which simultaneously settles both of the two problems.

To do this, at the beginning,
one expects to receive a state just after Step~1 of the \textsc{Reflection Procedure},
and then performs on this state two tests,
called \textsc{Reflection Test} and \textsc{Invertibility Test}, respectively,
with equal probability
without revealing which test the prover is undergoing.
In the \textsc{Reflection Test},
one simply performs Steps~2--4 of the \textsc{Reflection Procedure}
(i.e., one first applies the appropriate phase-flip
and then applies $\conjugate{U}$)
to finish the simulation of it.
In the \textsc{Invertibility Test},
one apply just $\conjugate{U}$ without performing the phase-flip
and checks if the entire state \emph{does} go back to a legal initial state
of the original \textsc{Reflection Procedure}.
We call the resulting procedure the \textsc{Modified Reflection Procedure},
a precise description of which will be given in Subsection~\ref{Subsection: Modified Reflection Procedure}.
The idea of making use of the \textsc{Invertibility Test}
originally appeared in Ref.~\cite{KemKobMatVid09CC} when achieving perfect completeness in quantum multi-prover interactive proofs,
but the test was used only after the forward simulation of the protocol
in their original construction,
and was not for the purpose of reducing the number of messages.
% From another viewpoint, the modification here may be considered as applying
% the ``halving technique'' in Ref.~\cite{KemKobMatVid09CC}
% to the original \textsc{Reflection Procedure},
% the technique originally used to reduce the number of turns by (almost) half in quantum multi-prover interactive proofs.

As is clear from the construction above,
the \textsc{Modified Reflection Procedure} requires only one application of $\conjugate{U}$ as desired.
Thus, the quantum interactive proof system that simulates this \textsc{Modified Reflection Procedure}
involves only $m$ messages as required
(for an even $m$, it involves ${m+1}$ messages,
as the original system starts with a turn for a verifier,
while the verifier in the constructed system needs to receive a witness before
his/her first turn).
Moreover, for any yes-instance,
the honest prover clearly has only to cooperate with the verifier to perform the backward simulation of the original \textsc{Reflection Procedure}
and can convince the verifier with certainty.
On the other hand,
for any no-instance,
the original \textsc{Reflection Procedure} would have rejected
with high probability,
if the proper $\conjugate{U}$ had been performed.
Thus, if the backward simulation in the \textsc{Modified Reflection Procedure}
were properly performed,
the \textsc{Reflection Test} of it could reject with high probability
as it properly simulates the original \textsc{Reflection Procedure}.
In contrast, if the backward simulation were not proper in the \textsc{Modified Reflection Procedure},
then the \textsc{Invertibility Test} of it would result in rejection with high probability,
as it essentially forces the prover to perform a proper backward simulation
of the original \textsc{Reflection Procedure}.
Indeed, as will be proved in Subsection~\ref{Subsection: Modified Reflection Procedure},
if one starts with a \textsc{Reflection Procedure} that rejects with probability at least $\varepsilon$ for every possible witness,
the resulting \textsc{Modified Reflection Procedure}
rejects with probability at least $\varepsilon/4$
no matter which witness is received
(the proof of Proposition~\ref{Proposition: soundness of Modified Reflection Procedure} essentially proves this).
Hence, the soundness can be shown as well in the \textsc{Modified Reflection Procedure}.

% ---------------------------------------------------------------------------
%   Preliminaries
% ---------------------------------------------------------------------------

\section{Preliminaries}
\label{Section: preliminaries}

Throughout this paper,
let $\Natural$ and $\Nonnegative$ denote
the sets of positive and nonnegative integers, respectively,
and let ${\Sigma = \Binary}$ denote the binary alphabet set.
A function~$\function{f}{\Nonnegative}{\Natural}$ is \emph{polynomially bounded}
if there exists a polynomial-time deterministic Turing machine
that outputs ${1^{f(n)}}$ on input~$1^n$.
A function~$\function{f}{\Nonnegative}{[0,1]}$ is \emph{negligible}
if, for every polynomially bounded function~$\function{g}{\Nonnegative}{\Natural}$,
it holds that ${f(n) < 1/g(n)}$ for all but finitely many values of~$n$.

\paragraph{Quantum Fundamentals}

We assume the reader is familiar with the quantum formalism,
including pure and mixed quantum states, density operators,
measurements, trace norm, fidelity, as well as the quantum circuit model
(see Refs.~\cite{NieChu00Book,KitSheVya02Book}, for instance).
Here we summarize some notations and properties that are used in this paper.

For each ${k \in \Natural}$, let ${\Complex(\Sigma^k)}$ denote the $2^k$-dimensional complex Hilbert space whose standard basis vectors are indexed by the elements in $\Sigma^k$.
%% In what follows, we simply say Hilbert spaces to mention complex Hilbert spaces.
In this paper, all Hilbert spaces are complex and have dimension a power of two.
For a Hilbert space~$\calH$,
let $I_\calH$ denote the identity operator over $\calH$,
and let ${\Density(\calH)}$ be the set of density operators over $\calH$.
For a quantum register~$\sfR$,
let $\ket{0}_\sfR$ denote the state
in which all the qubits in $\sfR$ are in state~$\ket{0}$.
As usual,
denote the two single-qubit states in ${\Complex(\Sigma)}$
that form the \emph{Hadamard basis} by
\[
\ket{+} = \frac{1}{\sqrt{2}} (\ket{0} + \ket{1}),
\quad
\ket{-} = \frac{1}{\sqrt{2}} (\ket{0} - \ket{1}),
\]
% ${\ket{+} = \frac{1}{\sqrt{2}} (\ket{0} + \ket{1})}$
% and
% ${\ket{-} = \frac{1}{\sqrt{2}} (\ket{0} - \ket{1})}$,
%% ${\ket{+} = \bigl( \ket{0} + \ket{1} \bigr) / \sqrt{2}}$
%% and
%% ${\ket{-} = \bigl( \ket{0} - \ket{1} \bigr) / \sqrt{2}}$,
and the four two-qubit states in ${\Complex(\Sigma^2)}$
that form the \emph{Bell basis} by
\begin{alignat*}{2}
\ket{\Phi^+}
&
=
\frac{1}{\sqrt{2}} (\ket{00} + \ket{11}),
&
\quad
\ket{\Phi^-}
&
=
\frac{1}{\sqrt{2}} (\ket{00} - \ket{11}),\\
\ket{\Psi^+}
&
=
\frac{1}{\sqrt{2}} (\ket{01} + \ket{10}),
&
\quad
\ket{\Psi^-}
&
=
\frac{1}{\sqrt{2}} (\ket{01} - \ket{10}),
\end{alignat*}
%% ${\ket{\Phi^+} = \bigl( \ket{00} + \ket{11} \bigr) / \sqrt{2}}$,
%% ${\ket{\Phi^-} = \bigl( \ket{00} - \ket{11} \bigr) / \sqrt{2}}$,
%% ${\ket{\Psi^+} = \bigl( \ket{01} + \ket{10} \bigr) / \sqrt{2}}$,
%% and
%% ${\ket{\Psi^-} = \bigl( \ket{01} - \ket{10} \bigr) / \sqrt{2}}$,
respectively.
Let
\[
  H
  =
  \frac{1}{\sqrt{2}}
  \begin{pmatrix}
    1 & 1\\
    1 & -1
  \end{pmatrix},
  \quad
  X
  =
  \begin{pmatrix}
    0 & 1\\
    1 & 0
  \end{pmatrix},
  \quad
  Z
  =
  \begin{pmatrix}
    1 & 0\\
    0 & -1
  \end{pmatrix}
\]
denote the Hadamard and Pauli operators.
% Let
% ${
%   H =
%   \frac{1}{\sqrt{2}}
%   \left(
%     \begin{smallmatrix}
%       1 & 1\\
%       1 & -1
%     \end{smallmatrix}
%   \right)
% }$
% denote the Hadamard operator,
% and let
% ${
%   X =
%   \left(
%     \begin{smallmatrix}
%       0 & 1\\
%       1 & 0
%     \end{smallmatrix}
%   \right)
% }$
% and
% ${
%   Z =
%   \left(
%     \begin{smallmatrix}
%       1 & 0\\
%       0 & -1
%     \end{smallmatrix}
%   \right)
% }$
% denote the Pauli operators.
For convenience, we may identify a unitary operator with the unitary transformation it induces.
In particular, for a unitary operator~$U$,
the induced unitary transformation is also denoted by $U$.

For a linear operator~$A$, the \emph{trace norm} of $A$ is defined by
\[
\trnorm{A} = \tr \sqrt{\conjugate{A}A}.
\]
%% ${\trnorm{A} = \tr \sqrt{\conjugate{A}A}}$.
For two quantum states~$\rho$ and $\sigma$,
the \emph{trace distance} between them is defined by
\[
D(\rho, \sigma) = \frac{1}{2} \trnorm{\rho - \sigma},
\]
%% ${D(\rho, \sigma) = \trnorm{\rho - \sigma}/2}$,
and the \emph{fidelity} between them is defined by
\[
F(\rho, \sigma) = \tr \sqrt{\sqrt{\rho} \sigma \sqrt{\rho}}.
\]
A special case of the trace distance is the \emph{statistical difference}
between two probability distributions~$\mu$ and $\nu$,
which is defined by
\[
\SD(\mu, \nu) = D(\mu, \nu)
\]
by viewing probability distributions as special cases of quantum states
with diagonal density operators.
We will use the following important properties of the trace distance
and fidelity.

\begin{lemma}
Let $\mu_\rho$ and $\mu_\sigma$
be the probability distributions derived from two quantum states~$\rho$ and $\sigma$, respectively, by performing an arbitrary identical measurement.
Then,
\[
  \SD(\mu_\rho, \mu_\sigma) \leq D(\rho, \sigma).
\]
\label{Lemma: trace distance and probability}
\end{lemma}

%% \begin{lemma}[\cite{FucGra99IEEEIT}]
%% For any quantum states~$\rho$ and $\sigma$,
%% ${
%%   1 - F(\rho, \sigma) \leq D(\rho, \sigma) \leq \sqrt{1 - (F(\rho, \sigma))^2}
%% }$.
%% \label{Lemma: trace distance and fidelity}
%% \end{lemma}

\begin{lemma}[\cite{SpeRud02PRA,NaySho03PRA}]
For any quantum states $\rho$, $\sigma$, and $\xi$,
%% ${F(\rho, \sigma)^2 + F(\sigma, \xi)^2 \leq 1 + F(\rho, \xi)}$.
\[
F(\rho, \sigma)^2 + F(\sigma, \xi)^2
\leq
1 + F(\rho, \xi).
\]
\label{Lemma: F(a,b)^2 + F(b,c)^2 < 1 + F(a,c)}
\end{lemma}

For any unitary transformation~$U$ acting over the two-dimensional Hilbert space~${\calH = \Complex(\Sigma)}$ (i.e., the single-qubit space),
the \emph{Choi-Jamio{\l}kowski state} of $U$ is the two-qubit state in ${\calH \tensor \calH = \Complex(\Sigma^2)}$ defined by
\[
\ket{J(U)} = (I \tensor U) \ket{\Phi^+}.
\]
In fact, the Choi-Jamio{\l}kowski state can be defined for any admissible 
(and not limited to unitary) 
transformation and any finite-dimensional Hilbert space,
using the Choi-Jamio{\l}kowski representation~\cite{Jam72RMP,Cho75LAA},
but which is unnecessary in this paper.

% ---------------------------------------------------------------------------
%   Finite Quantum de Finetti Theorem
% ---------------------------------------------------------------------------

\paragraph{The Finite Quantum de Finetti Theorem}

For ${N \in \Natural}$ and quantum registers~${\sfQ_1, \ldots, \sfQ_N}$,
each consisting of $k$~qubits,
an $N$-partite quantum state~$\rho$ in ${(\sfQ_1, \ldots, \sfQ_N)}$
is said to be \emph{symmetric} if $\rho$ is invariant under any permutation over the registers~${\sfQ_1, \ldots, \sfQ_N}$.
%% Further, $\rho$ is said to be a \emph{symmetric product state},
%% if it is of the form $\xi^{\tensor N}$ for some $k$-qubit state~$\xi$.

The \emph{finite quantum de Finetti theorem}~\cite{KonRen05JMP,ChrKonMitRen07CMP}
provides a very useful property that
the reduced $m$-partite state of any $N$-partite symmetric state
when tracing out the last ${N-m}$~subsystems
must be close to a mixture of $m$-fold product states.
This paper uses the following bound proved in Ref.~\cite{ChrKonMitRen07CMP}.

%\begin{theorem}[The finite quantum de Finetti theorem]
\begin{theorem}[Finite\hspace{-0.35mm} quantum\hspace{-0.35mm} de\hspace{-0.35mm} Finetti\hspace{-0.35mm} theorem]
For ${N, k \in \Natural}$,
let ${\sfQ_1, \ldots, \sfQ_N}$ be quantum registers
each consisting of $k$~qubits,
and let $\rho$ be an $N$-partite symmetric state in ${(\sfQ_1, \ldots, \sfQ_N)}$.
For any ${m \in \Natural}$ satisfying ${m < N}$
and the $m$-partite reduced state~$\rho^{(m)}$ of $\rho$ in ${(\sfQ_1, \ldots, \sfQ_m)}$,
%% obtained by tracing out the subsystems corresponding to ${\sfQ_{m+1}, \ldots, \sfQ_N}$,
there exist ${C \in \Natural}$, a set~${ \{\xi_j\}_{j \in \{1, \ldots, C \}}}$ of $k$-qubit states,
and an associated probability distribution~${\{\mu_j\}_{j \in \{1, \ldots, C \}}}$
such that
\[
D \biggl( \rho^{(m)}, \sum_{j=1}^C \mu_j \xi_j^{\tensor m} \biggr)
\leq
\frac{2^{2k+1} m}{N}.
\]
\label{Theorem: quantum de Finetti}
\end{theorem}

% ---------------------------------------------------------------------------
%   Polynomial-Time Uniformly Generated Families of Quantum Circuits
% ---------------------------------------------------------------------------

\paragraph{Polynomial-Time Uniformly Generated Families of Quantum Circuits}
%% \subsection{Polynomial-Time Uniformly Generated Families of Quantum Circuits}
%% \label{Subsection: Uniform QC}

Following conventions,
we define quantum Merlin-Arthur proof systems
in terms of quantum circuits.
In particular, we use the following notion of
polynomial-time uniformly generated families of quantum circuits.

A family~${\{ Q_x \}}$ of quantum circuits is
\emph{polynomial-time uniformly generated}
if there exists a deterministic procedure
that, on every input~$x$, outputs a description of $Q_x$
and runs in time polynomial in $\abs{x}$.
It is assumed that the circuits in such a family are composed of gates
in some reasonable, universal, finite set of quantum gates.
Furthermore, it is assumed that the number of gates in any circuit
is not more than the length of the description of that circuit.
Therefore $Q_x$ must have size polynomial in $\abs{x}$.
For convenience,
we may identify a circuit~$Q_x$ with the unitary operator it induces.

Throughout this paper, we assume a gate set
with which the Hadamard and any classical reversible transformations
can be exactly implemented.
Note that this assumption is satisfied by many standard gate sets
such as the Shor basis~\cite{Sho96FOCS}
consisting of the Hadamard, controlled-$i$-phase-shift, and Toffoli gates,
and the gate set consisting of the Hadamard, Toffoli, and NOT gates~\cite{Shi02QIC,Aha03arXiv}.
Moreover, as the Hadamard transformation in some sense can be viewed as
a quantum analogue of the classical operation of flipping a fair coin,
our assumption would be the most natural quantum correspondence
%% our assumption of the perfect implementability of
%% the Hadamard and any classical reversible transformations
%% would be the most natural quantum correspondence
to the tacit classical assumption in randomized complexity theory
that fair coins and perfect logical gates are available.
Hence we believe that our condition is very reasonable and not restrictive.
Note that, with a gate set satisfying this assumption,
any transformation corresponding to a Clifford group operator
is exactly implementable.
In particular,
the controlled-phase-flip transformation $Z$ can be exactly realized
by using an ancilla qubit prepared in state
${\ket{-} = \frac{1}{\sqrt{2}} (\ket{0} - \ket{1})}$
(by applying a NOT and an Hadamard in sequence to $\ket{0}$)
and performing a CNOT with this ancilla as the target.

Since non-unitary and unitary quantum circuits
are equivalent in computational power~\cite{AhaKitNis98STOC},
it is sufficient to treat only unitary quantum circuits,
which justifies the above definition.
Nevertheless, for readability,
most procedures in this paper will be described
using intermediate projective measurements
and unitary operations conditioned on the outcome of the measurements.
All of these intermediate measurements can be deferred
to the end of the procedure by a standard technique
so that the procedure becomes implementable with a unitary circuit.

\paragraph{Quantum Interactive Proof Systems}
%% \subsection{Quantum Merlin-Arthur Proof Systems}
%% \label{Subsection: QMA}

Now we review the model of quantum interactive proof systems.

A quantum interactive proof system is a communication model between two players
called a \emph{quantum verifier}~$V$ and a \emph{quantum prover}~$P$,
both of whom receive a common input~${x \in \Sigma^\ast}$.
Fix the input~$x$.
Let $\sfV$ and $\sfP$ be quantum registers corresponding to the private spaces of $V$ and $P$, respectively,
and let $\sfM$ be a quantum register corresponding to the message space
that is used to exchange messages between $V$ and $P$.
One of the qubits in $\sfV$, which is private to $V$, is designated as the \emph{output qubit}.
At the beginning, all the qubits in $\sfV$ and $\sfM$ are initialized to state~$\ket{0}$,
while the quantum state in $\sfP$ can be arbitrarily prepared by $P$.
%% Without loss of generality, the initial state in $\sfP$ can be assumed to be pure.
Then $V$ and $P$ together run a protocol that consists of alternating turns of the verifier and of the prover.
The first turn is for the verifier if the total number of turns is even,
and it is for the prover otherwise,
whereas the last turn is always for the prover.
%% The protocol consists of exchanging messages between the verifier and of the prover,
%% where the first message is from the verifier if the total number~$m$ of messages is even,
%% and it is from the prover if $m$ is odd.
At each turn of the verifier,
$V$ applies some unitary transformation implementable with a polynomial-size quantum circuit
to the state in ${(\sfV, \sfM)}$,
and then sends the register~$\sfM$ to $P$.
At each turn of the prover,
$P$ applies some unitary transformation to the state in ${(\sfP,\sfM)}$,
and then sends $\sfM$ to $V$.
After the last turn,
the verifier~$V$ further applies some unitary transformation implementable with a polynomial-size quantum circuit
to the state in ${(\sfV, \sfM)}$,
and then measures the output qubit in the standard basis.
$V$ accepts if this measurement results in $\ket{1}$ and rejects otherwise.

Formally, for any function~$\function{m}{\Nonnegative}{\Natural}$ that is polynomially bounded,
%% Formally, for any polynomially bounded function $\function{m}{\Nonnegative}{\Natural}$,
an \emph{$m$-message polynomial-time quantum verifier}
is a polynomial-time computable mapping~$\function{V}{\Sigma^\ast}{\Sigma^\ast}$.
For each input~${x \in \Sigma^\ast}$,
${V(x)}$ is interpreted as describing a series
${
\{ V_{x,j} \}_{j \in \{1, \ldots, \ceil{(m(\abs{x})+1)/2} \}}
}$ 
of quantum circuits
acting over the same number of qubits
as well as a partition of the qubits on which these circuits act into registers~$\sfV$ and $\sfM$,
where ${\{ V_{x,j} \}}$ is a polynomial-time uniformly generated family of quantum circuits
explained before
%% explained in Section~\ref{Section: preliminaries}
(in particular, every circuit~$V_{x,j}$ is composed of gates
in some reasonable, universal, finite set of quantum gates).
For any polynomially bounded function~$\function{m}{\Nonnegative}{\Natural}$,
an \emph{$m$-message quantum prover}
is a mapping~$P$
that simply maps an input binary string~${x \in \Sigma^\ast}$
to a series
${
\{ P_{x,j} \}_{j \in \{1, \ldots, \floor{(m(\abs{x})+1)/2} \}}
}$ 
of unitary transformations
as well as a partition of the qubits on which these unitary transformations act into registers~$\sfM$ and $\sfP$.
It is always assumed that $V$ and $P$ are \emph{compatible}
(i.e., the register~$\sfM$ is common for $V$ and $P$)
when they are associated with the same quantum interactive proof system.

Given an input~$x$, an $m$-message polynomial-time quantum verifier~$V$, and an $m$-message quantum prover~$P$,
let $Q_x$ be the unitary transformation induced from $V$ and $P$,
acting over the space corresponding to ${(\sfV, \sfM, \sfP)}$:
\[
  Q_x
  =
  (V_{x, (m(\abs{x}) + 1)/2} \tensor I_\calP)
  (I_\calV \tensor P_{x, (m(\abs{x}) + 1)/2})
  \cdots
  (V_{x,1} \tensor I_\calP)
  (I_\calV \tensor P_{x,1})
\]
if ${m(\abs{x})}$ is odd,
while
\[
  Q_x
  =
  (V_{x, (m(\abs{x})/2) + 1} \tensor I_\calP)
  (I_\calV \tensor P_{x, m(\abs{x})/2})
  (V_{x, m(\abs{x})/2} \tensor I_\calP)
  \cdots
  (I_\calV \tensor P_{x,1})
  (V_{x,1} \tensor I_\calP)
\]
if ${m(\abs{x})}$ is even,
% ${
%   Q_x
%   =
%   (V_{x, (m(\abs{x}) + 1)/2} \tensor I_\calP)
%   (I_\calV \tensor P_{x, (m(\abs{x}) + 1)/2})
%   \cdots
%   (V_{x,1} \tensor I_\calP)
%   (I_\calV \tensor P_{x,1})
% }$
% if ${m(\abs{x})}$ is odd,
% while
% ${
%   Q_x
%   =
%   (V_{x, (m(\abs{x})/2) + 1} \tensor I_\calP)
%   (I_\calV \tensor P_{x, m(\abs{x})/2})
%   (V_{x, m(\abs{x})/2} \tensor I_\calP)
%   \cdots
%   (I_\calV \tensor P_{x,1})
%   (V_{x,1} \tensor I_\calP)
% }$
% if ${m(\abs{x})}$ is even,
where $\calV$ and $\calP$ are the Hilbert spaces corresponding to $\sfV$ and $\sfP$, respectively.
When communicating with the prover~$P$ who prepares the initial state~${\rho \in \Density(\calP)}$,
the verifier~$V$ accepts the input~$x$
if the measurement of the designated output qubit in $\sfV$ in the standard basis results in $\ket{1}$
at the end of the protocol after having applied the unitary transformation~$Q_x$
to the initial state~${\ketbra{0}_{(\sfV, \sfM)} \tensor \rho}$ in ${(\sfV, \sfM, \sfP)}$.
%% The pair~${(V, P)}$ accepts the input~$x$ with initial state~${\rho \in \Density(\calP)}$
%% if the protocol starts with the initial state~${\ketbra{0}_{(\sfV, \sfM)} \tensor \rho}$ in ${(\sfV, \sfM, \sfP)}$
%% and the measurement of the designated output qubit in $\sfV$ in the standard basis results in $\ket{1}$
%% at the end of the protocol after having applied the unitary transformation~$Q_x$
%% to the initial state.

Formally, the class~${\QIP(m,c,s)}$ of problems
having $m$-message quantum interactive proof systems
with completeness~$c$ and soundness~$s$
is defined as follows.
For generality, throughout this paper,
we use promise problems~\cite{EveSelYac84ICtrl} rather than languages
when defining complexity classes.

\begin{definition}
Given a polynomially bounded function~$\function{m}{\Nonnegative}{\Natural}$
and functions~$\function{c,s}{\Nonnegative}{[0,1]}$ satisfying ${c > s}$,
%% and polynomial-time computable functions~$\function{c,s}{\Nonnegative}{[0,1]}$ satisfying ${c > s}$,
a promise problem~${A = (A_\yes, A_\no)}$ is in ${\QIP(m,c,s)}$
iff there exists an $m$-message polynomial-time quantum verifier~$V$
%% who is a polynomial-time uniformly generated family~$\{V_{x,j}\}_{x\in\Binary^*}$
%% of quantum circuits acting over ${v(\abs{x}) + m(\abs{x})}$ qubits,
such that, for every input~$x$:
\begin{description}
\item[\textnormal{(Completeness)}]
  if ${x \in A_\yes}$,
  there exist an $m$-message quantum prover~$P$
  and the initial state~$\rho_x$ of $P$
  that make $V$ accept $x$
%%   such that the pair~${(V, P)}$ accepts $x$ with initial state~$\rho_x$
  with probability at least ${c(\abs{x})}$,
\item[\textnormal{(Soundness)}]
  if ${x \in A_\no}$,
  for any $m$-message quantum prover~$P'$
  and any initial state~$\rho'_x$ of $P'$ prepared,
  $V$ accepts $x$
%%   the pair~${(V, P')}$ accepts $x$ with initial state~$\rho'_x$
  with probability at most ${s(\abs{x})}$.
\end{description}
\label{Definition: QIP(m,c,s)}
\end{definition}

The class~${\QIP(m)}$ of problems having $m$-message quantum interactive proof systems
is defined as follows.

\begin{definition}
Given a polynomially bounded function~$\function{m}{\Nonnegative}{\Natural}$,
a promise problem~${A = (A_\yes, A_\no)}$ is in ${\QIP(m)}$
iff $A$ is in ${\QIP(m, 1- \varepsilon, \varepsilon)}$
for some negligible function~$\function{\varepsilon}{\Nonnegative}{[0,1]}$.
\label{Definition: QIP(m)}
\end{definition}

Similarly, the class~${\QIP_1(m)}$ of problems having $m$-message quantum interactive proof systems of perfect completeness
is defined as follows.

\begin{definition}
Given a polynomially bounded function~$\function{m}{\Nonnegative}{\Natural}$,
a promise problem~${A = (A_\yes, A_\no)}$ is in ${\QIP_1(m)}$
iff $A$ is in ${\QIP(m, 1, \varepsilon)}$
for some negligible function~$\function{\varepsilon}{\Nonnegative}{[0,1]}$.
\label{Definition: QIP_1(m)}
\end{definition}

Finally, as quantum Merlin-Arthur proof systems are nothing
but one-message quantum interactive proof systems,
the classes~$\QMA$~and~$\QMA_1$ of problems having quantum Merlin-Arthur proof systems
and those of perfect completeness
are simply defined as follows, respectively.

\begin{definition}
A promise problem~${A = (A_\yes, A_\no)}$ is in $\QMA$
iff $A$ is in ${\QIP(1, 1- \varepsilon, \varepsilon)}$
for some negligible function~$\function{\varepsilon}{\Nonnegative}{[0,1]}$.
\label{Definition: QMA}
\end{definition}

\begin{definition}
A promise problem~${A = (A_\yes, A_\no)}$ is in $\QMA_1$
iff $A$ is in ${\QIP(1, 1, \varepsilon)}$
for some negligible function~$\function{\varepsilon}{\Nonnegative}{[0,1]}$.
\label{Definition: QMA_1}
\end{definition}

% Note that ${\QIP(m) = \QIP(m, 2/3, 1/3)}$ and ${\QIP_1(m) = \QMA(m, 1, 1/2)}$ hold,
% since the gap between completeness and soundness can be amplified exponentially
% by simply repeating the same proof system.

% ---------------------------------------------------------------------------
%   Quantum Merlin-Arthur Proof Systems with Shared EPR Pairs
% ---------------------------------------------------------------------------

\paragraph{Quantum Merlin-Arthur Proof Systems with Shared EPR Pairs}

We further introduce another variant of quantum Merlin-Arthur proof systems
in which Arthur and Merlin initially share some copies of the EPR pair~$\ket{\Phi^+}$.
If Arthur and Merlin are allowed to share $k$~EPR pairs initially,
the resulting systems are called
\emph{quantum Merlin-Arthur proof systems with $k$~shared EPR pairs},
or \emph{$k$-EPR QMA proof systems} in short.
Notice that this model is actually equivalent to a special case of
two-message quantum interactive proof systems in which
the first transformation of a verifier is just to create $k$~copies of the EPR pairs
(and $k$~halves of these EPR pairs are sent to a prover as the first message).

Formally, the class~${\EPRQMA{k}(c,s)}$ of problems
having quantum Merlin-Arthur proof systems with $k$~shared EPR pairs
with completeness~$c$ and soundness~$s$
is defined as follows.

\begin{definition}
Given a polynomially bounded function~$\function{k}{\Nonnegative}{\Natural}$
and functions~$\function{c,s}{\Nonnegative}{[0,1]}$ satisfying ${c > s}$,
a promise problem~${A = (A_\yes, A_\no)}$ is in ${\EPRQMA{k}(c,s)}$
iff $A$ has a two-message quantum interactive proof system
with completeness~$c$ and soundness~$s$
in which, for every input~$x$,
the first transformation of the associated quantum verifier
is just to create ${k(\abs{x})}$ copies of EPR pairs
and the first message from the verifier consists only of the ${k(\abs{x})}$~halves of these EPR pairs.
\label{Definition: k-EPR-QMA(c,s)}
\end{definition}

We further define the class~$\EPRQMA{\const}$ of problems
having quantum Merlin-Arthur proof systems with a constant number of shared EPR pairs with constant gap between completeness and soundness
and the class~$\EPRQMA{\const}_1$ of problems
%% and the class~${\EPRQMA{\const}(1,\const)}$ of problems
having those of perfect completeness with constant soundness error
as follows.

\begin{definition}
A promise problem~${A = (A_\yes, A_\no)}$ is in $\EPRQMA{\const}$
iff $A$ is in ${\EPRQMA{k}(2/3, 1/3)}$ for some constant~${k \in \Natural}$.
%% ${\EPRQMA{\const} \defeq \bigcup_{k \in \Natural, 0 \leq s < c \leq 1} \EPRQMA{k}(c,s)}$.
\label{Definition: const-EPR-QMA}
\end{definition}

\begin{definition}
A promise problem~${A = (A_\yes, A_\no)}$ is in $\EPRQMA{\const}_1$
iff $A$ is in ${\EPRQMA{k}(1, 1/2)}$ for some constant~${k \in \Natural}$.
%% ${\EPRQMA{\const}(1,\const) \defeq \bigcup_{k \in \Natural, s \in [0,1)} \EPRQMA{k}(1,s)}$.
\label{Definition: const-EPR-QMA_1}
\end{definition}

\begin{remark}
Definitions~\ref{Definition: const-EPR-QMA}~and~\ref{Definition: const-EPR-QMA_1}
are equivalent to the seemingly most conservative definitions~${\EPRQMA{\const} \defeq \bigcup_{k \in \Natural, 0 \leq s < c \leq 1} \EPRQMA{k}(c,s)}$
and ${\EPRQMA{\const}_1 \defeq \bigcup_{k \in \Natural, s \in [0,1)} \EPRQMA{k}(1,s)}$
% are equivalent to the seemingly most conservative definitions
% \begin{align*}
% \EPRQMA{\const}
% &
% \defeq
% \bigcup_{k \in \Natural, 0 \leq s < c \leq 1}
% \hspace{-1.5mm}
% \EPRQMA{k}(c,s)
% \\
% \EPRQMA{\const}_1
% &
% \defeq
% \hspace{1.45mm}
% \bigcup_{k \in \Natural, s \in [0,1)}
% \EPRQMA{k}(1,s)
% \end{align*}
of these classes,
for repeating the associated system with each of these classes constant times
can achieve arbitrarily large constant gap between completeness and soundness
(in the two-sided error case, one first achieves sufficiently large completeness
via a parallel repetition followed by a threshold value computation,
and then achieves desirably small soundness via another parallel repetition
of the obtained large-completeness system, without decreasing the completeness too much).
\end{remark}

% ---------------------------------------------------------------------------
%   Reflection Procedure
% ---------------------------------------------------------------------------

\section{Reflection Procedure}
\label{Section: Reflection Procedure}

We start with presenting a very simple base procedure,
which we call the \textsc{Reflection Procedure},
that forms a very base of our protocols to be constructed
-- basically, our protocols aim to simulate this base procedure
with several suitable modifications.

Let $\calH$ be some Hilbert space,
and consider two decompositions of $\calH$
into ${\calX_0 \oplus \calX_1}$ and ${\calY_0 \oplus \calY_1}$
for subspaces~$\calX_0$, $\calX_1$, $\calY_0$, and $\calY_1$ of $\calH$.
Let $\Delta_j$ be the projection over $\calH$ onto the subspace~$\calX_j$
and let $\Pi_j$ be that onto $\calY_j$,
for each ${j \in \Binary}$.

Let $U$ be some unitary transformation acting over $\calH$,
and let $M$ be the Hermitian operator over $\calH$ defined by
\[
M = \Delta_0 \conjugate{U} \Pi_0 U \Delta_0.
\]
Suppose that $M$ has an eigenvalue~${\lambda > 0}$
and consider the eigenstate (i.e., the normalized eigenvector)~$\ket{\phi_0}$
corresponding to $\lambda$.
Then, ${M \ket{\phi_0} = \lambda \ket{\phi_0}}$,
and thus,
\[
\Delta_0 \ket{\phi_0} = \frac{1}{\lambda} \Delta_0 M \ket{\phi_0} = \frac{1}{\lambda} M \ket{\phi_0} = \ket{\phi_0}.
\]
Define the four states~$\ket{\psi_0}$, $\ket{\psi_1}$, $\ket{\xi_0}$, and $\ket{\xi_1}$ in $\calH$ as follows:
\[
\ket{\psi_0} = \frac{\Pi_0 U \ket{\phi_0}}{\norm{\Pi_0 U \ket{\phi_0}}},
\quad
\ket{\psi_1} = \frac{\Pi_1 U \ket{\phi_0}}{\norm{\Pi_1 U \ket{\phi_0}}},
\quad
\ket{\xi_0} = \frac{\Delta_0 \conjugate{U} \ket{\psi_0}}{\norm{\Delta_0 \conjugate{U} \ket{\psi_0}}},
\quad
\ket{\xi_1} = \frac{\Delta_1 \conjugate{U} \ket{\psi_0}}{\norm{\Delta_1 \conjugate{U} \ket{\psi_0}}}.
\]
Then,
${
\norm{\Pi_0 U \ket{\phi_0}}
=
\norm{\Pi_0 U \Delta_0 \ket{\phi_0}}
%% =
%% \sqrt{\bra{\phi_0} \Delta_0 \conjugate{U} \Pi_0 U \Delta_0 \ket{\phi_0}}
=
\sqrt{\bra{\phi_0} M \ket{\phi_0}}
=
\sqrt{\lambda}
}$,
and thus,
${
\norm{\Pi_1 U \ket{\phi_0}}
=
\sqrt{1 - \lambda}
}$.
It follows that
\[
\norm{\Delta_0 \conjugate{U} \ket{\psi_0}}
=
\frac{1}{\sqrt{\lambda}} \norm{\Delta_0 \conjugate{U} \Pi_0 U \ket{\phi_0}}
=
\frac{1}{\sqrt{\lambda}} \norm{\Delta_0 \conjugate{U} \Pi_0 U \Delta_0 \ket{\phi_0}}
=
\frac{1}{\sqrt{\lambda}} \norm{M \ket{\phi_0}}
=
\sqrt{\lambda},
\]
and thus,
${
\norm{\Delta_1 \conjugate{U} \ket{\psi_0}}
=
\sqrt{1 - \lambda}
}$.
Hence,
\[
\ket{\xi_0}
=
\frac{1}{\sqrt{\lambda}} \Delta_0 \conjugate{U} \ket{\psi_0}
=
\frac{1}{\lambda} \Delta_0 \conjugate{U} \Pi_0 U \ket{\phi_0}
=
\frac{1}{\lambda} \Delta_0 \conjugate{U} \Pi_0 U \Delta_0 \ket{\phi_0}
=
\frac{1}{\lambda} M \ket{\phi_0}
=
\ket{\phi_0}.
\]
This implies that
\[
\conjugate{U} \ket{\psi_0}
=
\sqrt{\lambda} \ket{\xi_0} + \sqrt{1 - \lambda} \ket{\xi_1},
\quad
\conjugate{U} \ket{\psi_1}
=
\sqrt{1 - \lambda} \ket{\xi_0} - \sqrt{\lambda} \ket{\xi_1},
\]
which was the crucial property analyzed by Marriott~and~Watrous~\cite{MarWat05CC}
to develop their space-efficient QMA amplification technique.

It follows that
\[
\conjugate{U} (- \Pi_0 + \Pi_1) U \ket{\phi_0}
=
\conjugate{U} \bigl( - \sqrt{\lambda} \ket{\psi_0} + \sqrt{1 - \lambda} \ket{\psi_1} \bigr)
=
(1-2\lambda) \ket{\xi_0} - 2 \sqrt{\lambda (1 - \lambda)} \ket{\xi_1},
\]
and thus, when $M$ has an eigenvalue~$1/2$,
the corresponding eigenstate (which is necessarily in $\calX_0$)
must be transformed into a state in $\calX_1$ after the following process:
one first applies $U$ to $\ket{\phi_0}$,
next flips the phase of states in $\calY_0$
(i.e., applies the unitary transformation~${- \Pi_0 + \Pi_1}$),
and then applies $\conjugate{U}$.
This property can be used to test if $M$ has an eigenvalue $1/2$,
which is summarized in Figure~\ref{Figure: Reflection Procedure}.

\begin{figure}[t!]
\begin{algorithm*}{\textsc{Reflection Procedure}}
\begin{step}
\item
  Receive a quantum register~$\sfQ$.
  Reject if the state in $\sfQ$ does not belong to the subspace corresponding to the projection~$\Delta_0$,
  and otherwise apply $U$ to $\sfQ$.
%%   Measure the state in $\sfQ$ with respect to the projective measurement~$\{ \Delta_0, \Delta_1 \}$.
%%   Reject if the state is projected onto the space corresponding to $\Delta_1$,
%%   and otherwise apply $U$ to $\sfQ$.
\item
  Perform a phase-flip (i.e., multiply $-1$ in phase)
  if the state in $\sfQ$ belongs to the subspace corresponding to the projection~$\Pi_0$.
\item
  Apply $\conjugate{U}$ to $\sfQ$.
\item
  Reject if the state in $\sfQ$ belongs to the subspace corresponding to $\Delta_0$,
  and accept otherwise.
%%   Measure the state in $\sfQ$ with respect to the projective measurement~$\{ \Delta_0, \Delta_1 \}$.
%%   Reject if the state is projected onto the space corresponding to $\Delta_0$,
%%   and accept otherwise.
\end{step}
\end{algorithm*}
\caption{The \textsc{Reflection Procedure}.}
\label{Figure: Reflection Procedure}
\end{figure}

\begin{proposition}
Suppose that the Hermitian operator~${M = \Delta_0 \conjugate{U} \Pi_0 U \Delta_0}$
has an eigenvalue~$1/2$.
Then there exists a quantum state given in Step~1 of the \textsc{Reflection Procedure}
such that the procedure results in acceptance with certainty.
%% If the Hermitian operator~${M = \Delta_0 \conjugate{U} \Pi_0 U \Delta_0}$
%% has an eigenvalue~$1/2$,
%% there exists a quantum state given in Step~1 of the \textsc{Reflection Procedure}
%% such that the procedure results in acceptance with certainty.
\label{Proposition: completeness of Reflection Procedure}
\end{proposition}

\begin{proof}
Consider the case where the eigenstate of $M$ with its corresponding eigenvalue~$1/2$ is received in $\sfQ$ in Step~1.
Then the claim is immediate from the argument above.
\end{proof}

\begin{proposition}
For any ${\varepsilon \in (0, \frac{1}{2}]}$,
% For any ${\varepsilon > 0}$,
suppose that none of the eigenvalues of the Hermitian operator~${M = \Delta_0 \conjugate{U} \Pi_0 U \Delta_0}$
is in the interval~${\bigl( \frac{1}{2} - \varepsilon, \frac{1}{2} + \varepsilon \bigr)}$.
Then, for any quantum state given in Step~1 of the \textsc{Reflection Procedure},
the procedure results in rejection with probability at least ${4 \varepsilon^2}$.
\label{Proposition: soundness of Reflection Procedure}
\end{proposition}

\begin{proof}
Let $\ket{\psi}$ be any state received in $\sfQ$ in Step~1.
Without loss of generality,
one can assume that $\ket{\psi}$ is in $\calX_0$
(as otherwise either rejected in Step~1 or projected onto a state in $\calX_0$).

For the Hilbert space~$\calH$,
there always exists an orthonormal basis such that
all the basis states of it are eigenstates of $M$,
and thus,
the state~$\ket{\psi}$ can be necessarily written as
${
\ket{\psi}
=
\sum_{j=1}^d \alpha_j \ket{\phi_j}
}$
for ${d = \dim \calX_0 \leq \dim \calH}$,
where each $\ket{\phi_j}$ is an eigenstate of $M$ in $\calX_0$
and ${\sum_{j=1}^d \abs{\alpha_j}^2 = 1}$.
%%  a linear combination of the orthonormal eigenstates of $M$ in $\calX_0$.

From the analysis above,
every eigenstate~$\ket{\phi_j}$ of $M$ in $\calX_0$
with corresponding eigenvalue~${\lambda_j > 0}$
must satisfy that
\[
\Delta_0 \conjugate{U} (- \Pi_0 + \Pi_1) U \ket{\phi_j}
=
(1 - 2 \lambda_j) \ket{\phi_j}.
\]
On the other hand,
for every eigenstate~$\ket{\phi_j}$ of $M$ in $\calX_0$ with corresponding eigenvalue~${\lambda_j = 0}$,
it holds that
${
\norm{\Pi_0 U \ket{\phi_j}}
=
\norm{\Pi_0 U \Delta_0 \ket{\phi_j}}
=
\sqrt{\bra{\phi_j} M \ket{\phi_j}}
=
0
}$.
This implies 
${\Pi_1 U \ket{\phi_j} = U \ket{\phi_j}}$,
and thus,
\[
\Delta_0 \conjugate{U} (- \Pi_0 + \Pi_1) U \ket{\phi_j}
=
\Delta_0 \ket{\phi_j}
=
\ket{\phi_j}
=
(1 - 2 \lambda_j) \ket{\phi_j}.
\]
Therefore,
\[
\Delta_0 \conjugate{U} (- \Pi_0 + \Pi_1) U \ket{\psi}
=
\sum_{j=1}^d \alpha_j (1 - 2 \lambda_j) \ket{\phi_j},
\]
and thus,
the probability of rejection is at least
${
\sum_{j=1}^d \abs{\alpha_j}^2 (1 - 2 \lambda_j)^2
\geq
4 \varepsilon^2 \sum_{j=1}^d \abs{\alpha_j}^2
=
4 \varepsilon^2
}$,
as claimed.
\end{proof}

% ---------------------------------------------------------------------------
%   QMA is in QIP_1(2)
% ---------------------------------------------------------------------------

\section{$\boldsymbol{\QMA \subseteq \EPRQMA{\const}_1 \subseteq \QIP_1(2)}$}
%% \section{$\boldsymbol{\QMA \subseteq \EPRQMA{\const}(1,\const) \subseteq \QIP_1(2)}$}
\label{Section: QMA is in QIP_1(2)}

The goal of this section is to prove Theorem~\ref{Theorem: QMA is in 2^k-EPR-QMA(1,s)}.
In Subsection \ref{Subsection: Building Blocks} we first describe building blocks, before 
presenting the proof in Subsection \ref{Subsection: Proof of Theorem QMA is in 2^k-EPR-QMA(1,s)}.

% ---------------------------------------------------------------------------
%   Building Blocks
% ---------------------------------------------------------------------------

\subsection{Building Blocks}
\label{Subsection: Building Blocks}

% ---------------------------------------------------------------------------
%   Encoding Accepting Probability in Phase
% ---------------------------------------------------------------------------

\subsubsection{Encoding Accepting Probability in Phase}
\label{Subsection: Encoding Accepting Probability in Phase}

Let $V$ be the verifier of a certain QMA system.
%% Let ${A = (\Ayes, \Ano)}$ be in $\QMA$
%% and let $V$ be the verifier of the corresponding QMA system.
%% Without loss of generality,
%% one can assume that both completeness and soundness errors
%% are exponentially small in this QMA system.
Consider the quantum circuit~$V_x$ of $V$ when the input is $x$,
which acts over a pair of two registers~$\sfA$ of ${v(\abs{x})}$~qubits and $\sfM$ of ${m(\abs{x})}$~qubits,
for some polynomially bounded functions~$\function{v,m}{\Nonnegative}{\Natural}$.
The circuit~$V_x$ expects to receive a quantum witness of ${m(\abs{x})}$~qubits in register~$\sfM$,
and uses the ${v(\abs{x})}$~qubits in $\sfA$ as its work qubits.
The Hilbert spaces associated with $\sfA$ and $\sfM$
are denoted by $\calA$ and $\calM$, respectively.

For an input~$x$, let $p_x$ be the maximum acceptance probability of the verifier~$V$ in this QMA system.
Then, as pointed out by Marriott~and~Watrous~\cite{MarWat05CC},
$p_x$ corresponds to the maximum eigenvalue of the Hermitian operator
\[
M_x = \Pi_\init \conjugate{V_x} \Pi_\acc V_x \Pi_\init,
\]
where $\Pi_\init$ is the projection onto the subspace spanned by states in which
all the qubits in $\sfA$ are in state~$\ket{0}$,
and $\Pi_\acc$ is that onto the subspace spanned by accepting states of this QMA system.
Let $\ket{w_x}$ be the eigenstate (i.e., eigenvector) of $M_x$
corresponding to the eigenvalue~$p_x$.
A crucial analysis of Ref.~\cite{MarWat05CC}
(which essentially follows from the arguments in Section~\ref{Section: Reflection Procedure})
is that
\begin{alignat*}{2}
&
\Pi_\init \conjugate{V_x} \Pi_\acc V_x (\ket{0}_\sfA \tensor \ket{w_x}_\sfM)
&
&
=
p_x \ket{0}_\sfA \tensor \ket{w_x}_\sfM,
\\
&
\Pi_\init \conjugate{V_x} \Pi_\rej V_x (\ket{0}_\sfA \tensor \ket{w_x}_\sfM)
&
&
=
(1 - p_x) \ket{0}_\sfA \tensor \ket{w_x}_\sfM,
\end{alignat*}
where ${\Pi_\rej = I_{\calA \tensor \calM} - \Pi_\acc}$
%% where ${\Pi_\rej = I_{(\sfA,\sfM)} - \Pi_\acc}$
is the projection onto the subspace spanned by rejecting states of this QMA system.

Let ${p = p_x^2/(2p_x^2 - 2p_x + 1)}$.
Using the property explained above,
if one copy of $\ket{w_x}$ is given,
one can generate with high probability the state
\[
\ket{\chi_p}
%% =
%% W_p \ket{0}
=
\frac{1}{\sqrt{2p_x^2 - 2p_x + 1}}
\bigl[
  (1-p_x) \ket{0} + p_x \ket{1}
\bigr]
\]
as follows.
One uses a single-qubit register~$\sfR$ in addition to $\sfA$ and $\sfM$,
where one sets $\ket{w_x}$ in $\sfM$, and initializes all the qubits in $\sfA$ and $\sfR$
to state~$\ket{0}$.
First, one performs a forward simulation of the original system over $\sfA$ and $\sfM$
(i.e., applies $V_x$ to ${(\sfA, \sfM)}$),
and flips the qubit in $\sfR$ if the content of ${(\sfA, \sfM)}$
corresponds to an accepting state of the original system
(i.e., applies the unitary transformation
${X \tensor \Pi_\acc + I \tensor \Pi_\rej}$ to ${(\sfR, \sfA, \sfM)}$).
One then performs a backward simulation of the original system over $\sfA$ and $\sfM$
(i.e., applies $\conjugate{V_x}$ to ${(\sfA, \sfM)}$).
Now one measures all the qubits in $\sfA$ in the computational basis.
If no $\ket{1}$ is measured
(i.e., if the state is projected with respect to $\Pi_\init$,
which happens with probability ${2p_x^2 - 2p_x + 1}$),
the unnormalized state in the system must be
\[
\ket{0}_\sfR \tensor (1-p_x) \ket{0}_\sfA \tensor \ket{w_x}_\sfM
+
\ket{1}_\sfR \tensor p_x \ket{0}_\sfA \tensor \ket{w_x}_\sfM
=
\bigl[ (1-p_x) \ket{0} + p_x \ket{1} \bigr]_\sfR \tensor \ket{0}_\sfA \tensor \ket{w_x}_\sfM,
\]
and thus, the desired state is successfully generated in $\sfR$.
We call this procedure the \textsc{Distillation Procedure},
which is summarized in Figure~\ref{Figure: Distillation Procedure}.

\begin{figure}[t!]
\begin{algorithm*}{\textsc{Distillation Procedure}}
\begin{description}
\item[Input:]
  a single-qubit register~$\sfR$,
  a ${v(\abs{x})}$-qubit register~$\sfA$,
  and an ${m(\abs{x})}$-qubit register~$\sfM$.
\item[Output:]
  a single-qubit register~$\sfR$ or a symbol~$\bot$.
\end{description}
\begin{step}
\item
  Apply $V_x$ to ${(\sfA, \sfM)}$.
\item
  Flip the qubit in $\sfR$
  if the content of ${(\sfA, \sfM)}$ corresponds to an accepting state of the original system.
\item
  Apply $\conjugate{V_x}$ to ${(\sfA, \sfM)}$.
\item
  Measures all the qubits in $\sfA$ in the computational basis.
  If any of these measurements result in $\ket{1}$,
  output $\bot$,
  otherwise output $\sfR$.
\end{step}
\end{algorithm*}
\caption{The \textsc{Distillation Procedure}.}
\label{Figure: Distillation Procedure}
\end{figure}

% ---------------------------------------------------------------------------
%   Multiplicatively Adjusting Accepting Probabilities
% ---------------------------------------------------------------------------

\subsubsection{Multiplicatively Adjusting Accepting Probabilities}
\label{Subsection: Multiplicatively Adjusting Accepting Probabilities}

For a real number~${a \in [0,1]}$,
let $W_a$ be the unitary transformation defined by
\[
W_a
=
\begin{pmatrix}
\sqrt{1-a} & \sqrt{a}
\\
\sqrt{a} & - \sqrt{1-a}
\end{pmatrix}.
\]

Given a unitary transformation~$W_p$ for some real number ${p \in \bigl[\frac{1}{2}, 1 \bigr]}$,
we construct another unitary transformation~$U$
and an appropriate projection operator~$\Pi_0$
acting over two qubits
so that the probability~$\norm{\Pi_0 U \ket{00}}^2$ exactly equals $1/2$.

Suppose that one can apply another unitary transformation~$W_q$, for some real number ${q \in [0,1]}$,
and define the unitary transformation~$U$ and projection operator~$\Pi_0$ by
\[
U = W_p \tensor W_q,
\quad
\Pi_0 = \ketbra{11}.
\]
Then, clearly, ${\norm{\Pi_0 U \ket{00}}^2 = pq}$,
and thus, this probability equals $1/2$ if and only if ${pq = 1/2}$.
This in particular implies that
there exists a real number ${q \in [0,1]}$ that achieves the adjusted accepting probability exactly $1/2$ when ${p \geq 1/2}$,
but no ${q \in [0,1]}$ can make it exactly equal to $1/2$ when ${p < 1/2}$.

% ---------------------------------------------------------------------------
%   Simulating Unitaries with Choi-Jamiolkowski States
% ---------------------------------------------------------------------------

\subsubsection{Simulating Unitaries with Choi-Jamio{\l}kowski States}
\label{Subsection: Simulating Unitaries with Choi-Jamiolkowski States}

In this subsection,
we consider the case where the aforementioned unitary transformation~$W_a$ itself is not available,
but only the copies of its Choi-Jamio{\l}kowski state~${\ket{J(W_a)} = (I \tensor W_a) \ket{\Phi^+}}$ are available. 

Note that one copy of the Choi-Jamio{\l}kowski state~$\ket{J(W_a)}$
can be used to simulate one application of $W_a$
(the simulation succeeds with probability~$1/4$).
More precisely, the simulation of $W_a$ is done as follows.
Suppose one wants to apply $W_a$ to the qubit in some single-qubit register~$\sfR_1$,
while the state~$\ket{J(W_a)}$ is available in ${(\sfR_2, \sfR'_2)}$,
for some single-qubit registers~$\sfR_2$ and $\sfR'_2$.
Then one measures the state in ${(\sfR_1, \sfR_2)}$ in the Bell basis.
If this results in $\ket{\Phi^+}$, the application of $W_a$ succeeds,
and the desired state is available in the register~$\sfR'_2$
(which can be verified via an argument similar to the analysis of seminal quantum teleportation).

Actually, when one wants to apply $W_a$ to the specific state~$\ket{0}$,
there is a more efficient way than the simulation just explained above.
A key observation is that,
for any real number~${a \in [0,1]}$,
the unitary transformation~$W_a$ in the last subsection can be written as
\[
W_a
=
\begin{pmatrix}
\sqrt{1-a} & \sqrt{a}
\\
\sqrt{a} & - \sqrt{1-a}
\end{pmatrix}
=
\sqrt{1-a} Z + \sqrt{a} X,
\]
and thus,
the state~$\ket{\chi_a}$ is given by
%% the state~${U_a \ket{0}}$ is given by
\[
\ket{\chi_a} = W_a \ket{0} = \sqrt{1-a} \ket{0} + \sqrt{a} \ket{1},
\]
while the Choi-Jamio{\l}kowski state of $W_a$ is given by
\[
\ket{J(W_a)}
=
\sqrt{1-a} \ket{J(Z)} + \sqrt{a} \ket{J(X)}
%% =
%% (I \tensor U_a) \ket{\Phi^+}
%% =
%% \sqrt{a} (I \tensor Z) \ket{\Phi^+} + \sqrt{1-a} (I \tensor X) \ket{\Phi^+}
=
\sqrt{1-a} \ket{\Phi^-} + \sqrt{a} \ket{\Psi^+}.
\]
Hence, given one copy of the Choi-Jamio{\l}kowski state~$\ket{J(W_a)}$,
one can easily generate the state~${\ket{\chi_a} = W_a \ket{0}}$ in the first qubit
by applying the following unitary transformation~$T$ to $\ket{J(W_a)}$:
\[
T \colon
\ket{\Phi^-} \mapsto \ket{00},
\quad
\ket{\Psi^-} \mapsto \ket{01},
\quad
\ket{\Psi^+} \mapsto \ket{10},
\quad
\ket{\Phi^+} \mapsto \ket{11}
\]
(note that this $T$ can be realized
by first applying the CNOT transformation using the first qubit as the control,
then applying the Hadamard transformation~$H$ and the NOT transformation~$X$
in this order to the first qubit,
and finally applying CNOT again using the first qubit as the control).

% ---------------------------------------------------------------------------
%   Simulating the Reflection Procedure with Choi-Jamiolkowski States
% ---------------------------------------------------------------------------

\subsubsection{Simulating the Reflection Procedure with Choi-Jamio{\l}kowski States}
\label{Subsection: Simulating the Reflection Procedure with Choi-Jamiolkowski States}

Now we consider simulating the \textsc{Reflection Procedure}
with given two copies of ${\ket{\chi_p} = W_p \ket{0}}$
and two copies of a Choi-Jamio{\l}kowski state~$\ket{J(W_q)}$,
where $p$ and $q$ are real numbers in ${[0,1]}$.
The procedure basically follows the \textsc{Reflection Procedure}
with taking the register~$\sfQ$ to be a two-qubit register,
%% with taking the register~$\sfQ$ to be ${(\sfR, \sfS)}$ for single-qubit registers $\sfR$ and $\sfS$,
the initial state~$\ket{\phi_0}$ to be $\ket{00}$,
the projection~$\Delta_0$ to be $\ketbra{00}$,
and the underlying unitary~$U$ and projection~$\Pi_0$
to be ${W_p \tensor W_q}$ and $\ketbra{11}$, as defined in Subsection~\ref{Subsection: Multiplicatively Adjusting Accepting Probabilities}.
Thus, to precisely perform the \textsc{Reflection Procedure} in Figure~\ref{Figure: Reflection Procedure} in this setting,
we need to apply each of ${W_p = \conjugate{W_p}}$ and ${W_q = \conjugate{W_q}}$
twice.
Fortunately, each of the first applications of $W_p$ and $W_q$ is to the $\ket{0}$ state,
and thus,
one may simply replace these applications by just using a given copy of $\ket{\chi_p}$ and generating $\ket{\chi_q}$ from a copy of $\ket{J(W_q)}$, respectively.
The second applications of these unitaries
can be probabilistically simulated by
using the Choi-Jamio{\l}kowski states~$\ket{J(W_p)}$ and $\ket{J(W_q)}$,
where one creates $\ket{J(W_p)}$ from a copy of $\ket{\chi_p}$.
This leads to the procedure called \textsc{Reflection Simulation Test}
described in Figure~\ref{Figure: Reflection Simulation Test}.

\begin{figure}[t!]
\begin{algorithm*}{\textsc{Reflection Simulation Test}}
\begin{description}
\item[Input:]
  single-qubit registers~$\sfR_1$, $\sfR_2$, $\sfS_1$, $\sfS'_1$, $\sfS_2$, and $\sfS'_2$.
\item[Output:]
  ``accept'' or ``reject''.
\end{description}
\begin{step}
\item
  Receive six single-qubit registers~$\sfR_1$, $\sfR_2$, $\sfS_1$, $\sfS'_1$, $\sfS_2$, and $\sfS'_2$.\\
  Apply the unitary transformation~$T$ to the state in ${(\sfS_1, \sfS'_1)}$.\\
  Prepare $\ket{0}$ in a single-qubit register $\sfR'_2$.
\item
  Perform a phase-flip (i.e., multiply $-1$ in phase)
  if ${(\sfR_1, \sfS_1)}$ contains $11$.
%%   if the state in ${(\sfR_1, \sfS_1)}$ belongs to the subspace corresponding to the projection~$\Pi_0$.
\item
  Try to simulate Step~3 of the \textsc{Reflection Procedure}
  by performing the following:\\
  Apply $\conjugate{T}$ to the state in ${(\sfR_2, \sfR'_2)}$. % where $\sfR'_2$ is an ancilla qubit initialized to $\ket{0}$.
  Measure the states in ${(\sfR_1, \sfR_2)}$ and ${(\sfS_1, \sfS_2)}$ in the Bell basis.
  Continue if both of these two measurements result in $\ket{\Phi^+}$,
  and accept otherwise (accept with giving up due to failure of the simulation).
\item
  Reject if ${(\sfR'_2, \sfS'_2)}$ contains $00$,
  and accept otherwise.
\end{step}
\end{algorithm*}
\caption{The \textsc{Reflection Simulation Test}, which tries to simulate the \textsc{Reflection Procedure} using Choi-Jamio{\l}kowski states.}
\label{Figure: Reflection Simulation Test}
\end{figure}

Now we analyze the properties of this simulation.

\begin{proposition}
The \textsc{Reflection Simulation Test} accepts with certainty
if the state in the input register~${(\sfR_1, \sfR_2, \sfS_1, \sfS'_1, \sfS_2, \sfS'_2)}$ is 
${
\ket{\chi_p}^{\tensor 2} \tensor \ket{J(W_q)}^{\tensor 2}
}$
for some real numbers~${p, q \in [0,1]}$ satisfying ${pq = 1/2}$.
\label{Proposition: completeness of Reflection Simulation Test}
\end{proposition}
 
\begin{proof}
The claim is almost obvious.
With $\ket{\chi_p}$ in $\sfR_1$ and $\ket{J(W_q)}$ in ${(\sfS_1, \sfS'_1)}$ for such $p$ and $q$,
Step~1 in the \textsc{Reflection Simulation Test}
creates the state
\[
U \ket{00}
=
\bigl( \sqrt{1-p} \ket{0} + \sqrt{p} \ket{1} \bigr)_{\sfR_1}
\tensor
\bigl( \sqrt{1-q} \ket{0} + \sqrt{q} \ket{1} \bigr)_{\sfS_1}
\]
in ${(\sfR_1, \sfS_1)}$,
since the application of $T$ generates the state~$\ket{\chi_q}$ in $\sfS_1$.
As the application of $\conjugate{T}$ in Step~3 generates
the Choi-Jamio{\l}kowski state~$\ket{J(W_p)}$ in ${(\sfR_2, \sfR'_2)}$,
one succeeds in Step~3 with probability ${(1/4)^2 = 1/16}$ in applying both of ${\conjugate{W_p} = W_p}$ and ${\conjugate{W_q} = W_q}$,
which successfully simulates $\conjugate{U}$
with generating the desired state in ${(\sfR'_2, \sfS'_2)}$.
Hence, the simulation of the \textsc{Reflection Procedure} succeeds with probability~$1/16$,
in which case the test necessarily results in acceptance
as in the analysis in Section~\ref{Section: Reflection Procedure},
since
${
(\ketbra{00} \conjugate{U} \Pi_0 U \ketbra{00}) \ket{00}
= \norm{\Pi_0 U \ket{00}}^2 \ket{00}
= \frac{1}{2} \ket{00}
}$.
On the other hand, if any of measurements in Step~3 fails in measuring $\ket{\Phi^+}$, 
the test just stops and accepts with giving up.
Therefore, the test must result in acceptance with certainty.
\end{proof}

\begin{proposition}
For any real number~${q \in [0,1]}$,
the \textsc{Reflection Simulation Test} results in rejection with probability~$1/16$
if the state in the input register~${(\sfR_1, \sfR_2, \sfS_1, \sfS'_1, \sfS_2, \sfS'_2)}$ is either
${
\ket{0}^{\tensor 2} \tensor \ket{J(W^+_q)}^{\tensor 2}
}$
or
${
\ket{0}^{\tensor 2} \tensor \ket{J(W^-_q)}^{\tensor 2}
}$,
where ${W^+_q = W_q}$ and
% ${W_q^- = Z W_q Z = \sqrt{1-q} Z - \sqrt{q} X}$.
\[
W^-_q
=
Z W_q Z
=
\begin{pmatrix}
  \sqrt{1-q} & - \sqrt{q}\\
  - \sqrt{q} & - \sqrt{1-q}
\end{pmatrix}
=
\sqrt{1-q} Z - \sqrt{q} X.
\]
\label{Proposition: soundness of Reflection Simulation Test}
\end{proposition}

\begin{proof}
We prove the case where the state in ${(\sfR_1, \sfR_2, \sfS_1, \sfS'_1, \sfS_2, \sfS'_2)}$ is
${
\ket{0}^{\tensor 2} \tensor \ket{J(W^+_q)}^{\tensor 2}
}$.
The other case is proved similarly,
by noticing that
${
  T \ket{J(W^-_q)} = (W^-_q \ket{0}) \tensor \ket{0}
}$
and
${\conjugate{{W^-_q}} = W^-_q}$
hold for any ${q \in [0,1]}$.

With $\ket{0}$ in $\sfR_1$ and ${\ket{J(W^+_q)} = \ket{J(W_q)}}$ in ${(\sfS_1, \sfS'_1)}$,
Step~1 in the \textsc{Reflection Simulation Test}
creates the state
\[
\ket{0}_{\sfR_1} \tensor \ket{\chi_q}_{\sfS_1}
\]
in ${(\sfR_1, \sfS_1)}$.
For this state given, Step~2 in the \textsc{Reflection Simulation Test}
does not change the state in ${(\sfR_1, \sfS_1)}$ at all.
As ${\ket{0} = \ket{\chi_0}}$,
the application of $\conjugate{T}$ in Step~3 generates
the Choi-Jamio{\l}kowski state~$\ket{J(W_0)}$ in ${(\sfR_2, \sfR'_2)}$,
and thus,
one succeeds in Step~3 with probability ${(1/4)^2 = 1/16}$ in applying both of ${\conjugate{W_0} = W_0}$ and ${\conjugate{W_q} = W_q}$.
If such an event occurs,
the state in ${(\sfR'_2, \sfS'_2)}$ becomes ${\ket{0}_{\sfR'_2} \tensor \ket{0}_{\sfS'_2}}$,
and thus, the test results in rejection with certainty.

Taking it into account that the test just stops and accepts with giving up
when any of measurements in Step~3 fails in measuring $\ket{\Phi^+}$,
the test results in rejection with probability~$1/16$ in total.
\end{proof}

% ---------------------------------------------------------------------------
%   Proof of Theorem QMA is in QIP_1(2)
% ---------------------------------------------------------------------------

\subsection{Proof of Theorem~\ref{Theorem: QMA is in 2^k-EPR-QMA(1,s)}}
\label{Subsection: Proof of Theorem QMA is in 2^k-EPR-QMA(1,s)}

Now we are ready to prove Theorem~\ref{Theorem: QMA is in 2^k-EPR-QMA(1,s)}.

\begin{proof}[Proof of Theorem~\ref{Theorem: QMA is in 2^k-EPR-QMA(1,s)}]
Let ${A = (A_\yes, A_\no)}$ be in $\QMA$
and let $V$ be the verifier of the corresponding QMA system.
Without loss of generality,
one can assume that both completeness and soundness errors
are exponentially small in this QMA system.

For an input~$x$,
the quantum circuit~$V_x$ of the verifier~$V$
acts over a pair of two registers~$\sfA$ of ${v(\abs{x})}$~qubits and $\sfM$ of ${m(\abs{x})}$~qubits,
for some polynomially bounded functions~$\function{v,m}{\Nonnegative}{\Natural}$.
% for some functions~$\function{v,m}{\Nonnegative}{\Natural}$ that are polynomially bounded.
This can be interpreted as
$V_x$ expecting to receive a quantum witness~$\ket{w}$ of ${m(\abs{x})}$~qubits in register~$\sfM$,
and using the ${v(\abs{x})}$~qubits in $\sfA$ as its work qubits.
By Refs.~\cite{Shi02QIC,Aha03arXiv},
one can further assume that the quantum circuit~$V_x$ for any input~$x$
consists of only the Hadamard, Toffoli, and NOT gates.
As pointed out by Marriott~and~Watrous~\cite{MarWat05CC},
the maximum acceptance probability~$p_x$ of $V$ with input~$x$
corresponds to the maximum eigenvalue of the Hermitian operator
\[
M_x = \Pi_\init \conjugate{V_x} \Pi_\acc V_x \Pi_\init,
\]
where $\Pi_\init$ is the projection onto the subspace spanned by states in which
all the qubits in $\sfA$ are in state~$\ket{0}$,
and $\Pi_\acc$ is the projection onto the space spanned by the accepting states of $V$.
From this verifier~$V$,
we shall construct a protocol for the verifier~$W$ of another QMA system
in which $W$ shares $N$~EPR pairs a priori with a prover communicating with,
where $N$ is a constant that is a power of two.
%% where $N$ is a constant that is a power of two and at least $2^{18}$.

Our basic strategy is to try to perform the \textsc{Reflection Simulation Test}
using $V_x$. 
Fix an input~$x$,
and let 
${p = \frac{p_x^2}{2p_x^2 - 2 p_x + 1}}$.
Let ${\sfS_1, \ldots, \sfS_N}$ be single-qubit registers
which store the particles of the shared EPR pairs.
In addition to $\sfM$,
$W$ receives $N$~single-qubit registers~${\sfS'_1, \ldots, \sfS'_N}$.
$W$ expects to receive in $\sfM$ the state~$\ket{w_x}$
that is the eigenstate (i.e., eigenvector) of $M_x$ corresponding to the eigenvalue~$p_x$,
and to receive states in ${\sfS'_1, \ldots, \sfS'_N}$
such that the state in ${(\sfS_j, \sfS'_j)}$ forms $\ket{J(W_q)}$ for each ${j \in \{1, \ldots, N\}}$,
for $q$ satisfying 
%% \[
%% pq = \frac{p_x^2}{2p_x^2 - 2 p_x + 1} q = \frac{1}{2}.
%% \]
${pq = \frac{p_x^2}{2p_x^2 - 2 p_x + 1} q = 1/2}$.
In addition to $\sfA$,
$W$ prepares three single-qubit registers~$\sfB$, $\sfR_1$, and $\sfR_2$.
All the qubits in $\sfA$, $\sfB$, $\sfR_1$, and $\sfR_2$
are initialized to the $\ket{0}$ state.

First, $W$ performs the \textsc{Distillation Procedure} twice in sequence,
first with ${(\sfR_1, \sfA, \sfM)}$ as input,
and second with ${(\sfR_2, \sfA, \sfM)}$ as input.
If any of these two runs of the \textsc{Distillation Procedure} outputs a symbol~$\bot$,
the simulation fails, and thus accept with giving up.
If not failed, then $W$ chooses two indices~$r_1$ and $r_2$ from the set~$\{1, \ldots, N\}$ uniformly at random.
If ${r_2 = 1}$, $W$ accepts with giving up.
Otherwise $W$ swaps the registers~${(\sfS_1, \sfS'_1)}$ and ${(\sfS_{r_1}, \sfS'_{r_1})}$ if ${r_1 \geq 2}$,
and further swaps ${(\sfS_2, \sfS'_2)}$ and ${(\sfS_{r_2}, \sfS'_{r_2})}$ if ${r_2 \geq 3}$.
Afterwards, $W$ never touches the registers~${(\sfS_j, \sfS'_j)}$ for ${j \geq 3}$,
and thus this process essentially has the same effect
as performing a random permutation over the registers~${(\sfS_1, \sfS'_1), \ldots, (\sfS_N, \sfS'_N)}$.
$W$ then performs the \textsc{Space Restriction Test}
by checking if the state in ${(\sfS_j, \sfS'_j)}$ is in the space spanned by $\{ \ket{\Phi^-}, \ket{\Psi^+}\}$, for each ${j \in \{1,2\}}$,
and further performs the \textsc{Swap Test} between ${(\sfS_1, \sfS'_1)}$ and ${(\sfS_2, \sfS'_2)}$
(using the register~$\sfB$ as the control).
Finally, $W$ performs the \textsc{Reflection Simulation Test}
with ${(\sfR_1, \sfR_2, \sfS_1, \sfS'_1, \sfS_2, \sfS'_2)}$ as input.
The protocol is summarized in Figure~\ref{Figure: verifier's EPR-QMA protocol for achieving perfect completeness}.
Notice that this protocol is exactly implementable when the Hadamard and any classical reversible transformations can be performed exactly.

\begin{figure}[t!]
\begin{algorithm*}{Verifier's QMA Protocol for Achieving Perfect Completeness with $\boldsymbol{N}$~Prior-Shared EPR Pairs}
\begin{step}
\item
  Store the particles of the shared $N$~EPR pairs in ${(\sfS_1, \ldots, \sfS_N)}$. 
  Receive an ${\bigl( m(\abs{x}) + N \bigr)}$-qubit quantum witness
  in ${(\sfM, \sfS'_1, \ldots, \sfS'_N)}$,
  where the first ${m(\abs{x})}$ qubits of the witness are in $\sfM$,
  and the ${\bigl( m(\abs{x}) + j \bigr)}$-th qubit of the witness is in $\sfS_j$, for ${j \in \{1, \ldots, N \}}$.\\
  Prepare $\ket{0}$ in each of the three single-qubit registers~$\sfB$, $\sfR_1$ and $\sfR_2$,
  and $\ket{0}^{\tensor v(\abs{x})}$ in a ${v(\abs{x})}$-qubit register~$\sfA$, which corresponds to the private space of the original verifier.
\item
  Execute the \textsc{Distillation Procedure} with ${(\sfR_1, \sfA, \sfM)}$ as input.
  Accept if this outputs $\bot$, and continue otherwise.
  Execute the \textsc{Distillation Procedure} again, this time using ${(\sfR_2, \sfA, \sfM)}$ as input.
  Accept if this outputs $\bot$, and continue otherwise.
\item
  Choose two integers~$r_1$ and $r_2$ from ${\{1, \ldots, N \}}$ uniformly at random.
  Accept if ${r_2 = 1}$ (accept with giving up due to failure of simulation),
  and continue otherwise.
  Swap the registers~${(\sfS_1, \sfS'_1)}$ and ${(\sfS_{r_1}, \sfS'_{r_1})}$
  if ${r_1 \geq 2}$,
  and further swap the registers~${(\sfS_2, \sfS'_2)}$ and ${(\sfS_{r_2}, \sfS'_{r_2})}$
  if ${r_2 \geq 3}$.
\item
  Perform the \textsc{Space Restriction Test}
  to check if the state in ${(\sfS_j, \sfS'_j)}$ is in the space spanned by $\{ \ket{\Phi^-}, \ket{\Psi^+}\}$, for each ${j \in \{1,2\}}$.
  Reject if not so, and continue otherwise.\\
  That is, perform the following for each ${j \in \{1,2\}}$:
  Apply the unitary transformation~$T$ defined by
  \[
  T \colon
  \ket{\Phi^-} \! \mapsto \! \ket{00},
  \ket{\Psi^-} \! \mapsto \! \ket{01},
  \ket{\Psi^+} \! \mapsto \! \ket{10},
  \ket{\Phi^+} \! \mapsto \! \ket{11}
  \]
  to the state in ${(\sfS_j, \sfS'_j)}$.
  Reject if $\sfS'_j$ contains $1$, and apply $\conjugate{T}$ to the state in ${(\sfS_j, \sfS'_j)}$ to continue otherwise.
\item
  Perform the \textsc{Swap Test} between ${(\sfS_1, \sfS'_1)}$ and ${(\sfS_2, \sfS'_2)}$.
  Reject if it fails, and continue otherwise.\\
  That is,
  apply $H$ to $\sfB$,
  swap ${(\sfS_1, \sfS'_1)}$ and ${(\sfS_2, \sfS'_2)}$ if $\sfB$ contains $1$,
  apply $H$ to $\sfB$ again,
  and reject if $\sfB$ contains $1$, and continue otherwise.
\item
  Perform the \textsc{Reflection Simulation Test}
  with ${(\sfR_1, \sfR_2, \sfS_1, \sfS'_1, \sfS_2, \sfS'_2)}$ as input.
  Accept if this returns ``accept'', and reject otherwise.
\end{step}
\end{algorithm*}
\caption{Verifier's QMA protocol for achieving perfect completeness with $\boldsymbol{N}$ pre-shared EPR pairs.}
\label{Figure: verifier's EPR-QMA protocol for achieving perfect completeness}
\end{figure}

For the completeness, suppose that $x$ is in $A_\yes$.
Let ${p = \frac{p_x^2}{2p_x^2 - 2 p_x + 1}}$.
The honest Merlin sets his shares of the $N$~EPR pairs in single-qubit registers~${\sfS'_1, \ldots, \sfS'_N}$,
and applies $W_q$ to each qubit in ${(\sfS'_1, \ldots, \sfS'_N)}$
to create the state~$\ket{J(W_q)}$ in ${(\sfS_j, \sfS'_j)}$, for ${j \in \{1, \ldots, N\}}$,
where $q$ satisfies ${pq = 1/2}$
(such a $q$ always exists when ${p_x \geq 1/2}$,
which is ensured by the completeness condition of the original QMA system).
He also prepares $\ket{w_x}$ in $\sfM$,
and sends the ${\bigl( m(\abs{x}) + N \bigr)}$-qubit state in ${(\sfM, \sfS'_1, \ldots, \sfS'_N)}$ as a witness.
Then, conditioned on the first application of the \textsc{Distillation Procedure} not outputting $\bot$,
the state~${\ket{\chi_p} = W_p \ket{0}}$ is generated in $\sfR_1$,
and ${\ket{0}^{\tensor v(\abs{x})} \tensor \ket{w_x}}$ is left in ${(\sfA, \sfM)}$,
and thus,
the state~$\ket{\chi_p}$ is generated also in $\sfR_2$
when the second application of the \textsc{Distillation Procedure} does not output $\bot$.
Conditioned on the chosen $r_2$ not being $1$ in Step~3,
the protocol continues and the state remains the same after this step.
When continued,
the \textsc{Space Restriction Test} in Step~4 clearly never rejects and does not change the state at all,
as the state in ${(\sfS_j, \sfS'_j)}$ is
${\ket{J(W_q)} = \sqrt{1-q} \ket{\Phi^-} + \sqrt{q} \ket{\Psi^+}}$
for each ${j \in \{1, 2\}}$.
Furthermore,
the \textsc{Swap Test} never fails in Step~5 and it does not change the state at all
(and thus, the protocol never results in rejection in this step).
Therefore, the state in ${(\sfR_1, \sfR_2, \sfS_1, \sfS'_1, \sfS_2, \sfS'_2)}$
is ${\ket{\chi_p}^{\tensor 2} \tensor \ket{J(W_q)}^{\tensor 2}}$,
when entering Step~6.
Hence, from Proposition~\ref{Proposition: completeness of Reflection Simulation Test},
the \textsc{Reflection Simulation Test} results in acceptance with certainty,
when the protocol reaches Step~6.
As rejections can happen only in Steps~4,~5,~and~6, this proves the perfect completeness.

Now for the soundness, suppose that $x$ is in $A_\no$.
Let $\calR_j$, $\calS_j$, and $\calS'_j$ denote the Hilbert spaces
associated with the quantum registers~$\sfR_j$, $\sfS_j$, and $\sfS'_j$, for each $j$, respectively.

As the soundness error of the original QMA system is exponentially small,
whatever state the register $\sfM$ contains,
the probability that the first application of the \textsc{Distillation Procedure} outputs $\bot$ is exponentially small.
Moreover, conditioned on this not outputting $\bot$,
the state generated in $\sfR_1$ is exponentially close to $\ket{0}$ (in trace distance).
Similarly, whatever state left in $\sfM$ after the first application of the \textsc{Distillation Procedure},
the probability that the second application of the \textsc{Distillation Procedure} outputs $\bot$ is exponentially small,
and the state generated in $\sfR_2$ is exponentially close to $\ket{0}$.
Hence, the state in ${(\sfR_1, \sfR_2, \sfS_1, \sfS'_1, \ldots, \sfS_N, \sfS'_N)}$
when entering Step~2 must be exponentially close to
${ (\ketbra{0})^{\tensor 2} \tensor \rho }$
for some $2N$-qubit state~$\rho$
such that the reduced state ${\tr_{\calS'_1 \tensor \cdots \tensor \calS'_N} \rho}$
is equal to the $N$-qubit totally mixed state~${(I/2)^{\tensor N}}$.
%% such that the reduced state of $\rho$ when tracing out the space corresponding to
%% the registers~${\sfS'_1, \ldots, \sfS'_N}$
%% is equal to the $N$-qubit totally mixed state~${(I/2)^{\tensor N}}$.

As Step~3 essentially has the same effect
as performing a random permutation over the registers~${(\sfS_1, \sfS'_1), \ldots, (\sfS_N, \sfS'_N)}$
for the purpose of computing the reduced state in ${(\sfS_1, \sfS'_1, \sfS_2, \sfS'_2)}$,
from the finite quantum de Finetti theorem (Theorem~\ref{Theorem: quantum de Finetti}),
the state in ${(\sfR_1, \sfR_2, \sfS_1, \sfS'_1, \sfS_2, \sfS'_2)}$ after Step~3
should have trace distance at most $\frac{2^6}{N}$
to the state
\[
\sigma
= (\ketbra{0})^{\tensor 2} \tensor \Bigl( \sum_j \mu_j \xi_j^{\tensor 2} \Bigr)
\]
for some two-qubit states~$\xi_j$,
where ${\sum_j \mu_j = 1}$,
if the state in ${(\sfR_1, \sfR_2, \sfS_1, \sfS'_1, \ldots, \sfS_N, \sfS'_N)}$
were ${ (\ketbra{0})^{\tensor 2} \tensor \rho }$ when entering Step~3
and if ${r_2 \neq 1}$
(here we are taking the randomness over the choices of $r_1$ and $r_2$ into account).
By letting ${\tau = \sum_j \mu_j \xi_j^{\tensor 2}}$,
this in particular implies that for the reduced state~${\tr_{\calS'_1 \tensor \calS'_2} \tau}$ and the two-qubit totally mixed state~${(I/2)^{\tensor 2}}$,
\[
D \biggl( \tr_{\calS'_1 \tensor \calS'_2} \tau, \Bigl( \frac{I}{2} \Bigr)^{\tensor 2} \biggr) \leq \frac{2^6}{N}
\]
holds,
since ${\tr_{\calS'_1 \tensor \cdots \tensor \calS'_N} \rho = (I/2)^{\tensor N}}$.
%% This in particular implies that the reduced state of ${\tau_1 = \sum_j p_j \xi_j^{\tensor 2}}$
%% when tracing out the space corresponding to the registers~$\sfS'_1$ and $\sfS'_2$
%% has trace distance at most $\frac{2^6}{N}$ to the two-qubit totally mixed state~${(I/2)^{\tensor 2}}$,
%% since the reduced state of $\rho$ when tracing out the space corresponding to
%% the registers~${\sfS'_1, \ldots, \sfS'_N}$
%% was equal to the $N$-qubit totally mixed state~${(I/2)^{\tensor N}}$.
Taking it into account that
the protocol enters Step~3 with probability exponentially close to $1$
with the state in ${(\sfR_1, \sfR_2, \sfS_1, \sfS'_1, \ldots, \sfS_N, \sfS'_N)}$
being exponentially close to ${ (\ketbra{0})^{\tensor 2} \tensor \rho }$ in trace distance,
we conclude that the protocol enters Step~4 with probability exponentially close to ${1 - \frac{1}{N}}$
with the state in ${(\sfR_1, \sfR_2, \sfS_1, \sfS'_1, \sfS_2, \sfS'_2)}$
having trace distance at most ${\frac{2^6}{N} + \varepsilon}$ to $\sigma$
for some exponentially small $\varepsilon$.

Now from Proposition~\ref{Proposition: soundness with sigma_1}
which will be found below and proved in the end of this section,
%% below,
the protocol should result in rejection with probability at least
${
  \min \bigl\{
         \frac{2^7}{N},
         \frac{1}{16} - 15 \bigl( \frac{2^6}{N} \bigr)^{\frac{1}{8}}
       \bigr\}
}$
if the state in ${(\sfR_1, \sfR_2, \sfS_1, \sfS'_1, \sfS_2, \sfS'_2)}$ were $\sigma$
when entering Step~4.
Hence, using Lemma~\ref{Lemma: trace distance and probability},
the protocol results in rejection with probability at least
${
  \min \bigl\{
         \frac{2^6}{N} - \varepsilon,
         \frac{1}{16} - \frac{2^6}{N} - \varepsilon - 15 \bigl( \frac{2^6}{N} \bigr)^{\frac{1}{8}}
       \bigr\}
}$,
when entering Step~4.
As the protocol enters Step~4 with probability exponentially close to ${1 - \frac{1}{N}}$,
by taking ${N = 2^{70}}$,
the protocol results in rejection with probability at least
\[
\Bigl( 1 - \frac{1}{2^{69}} \Bigr)
\cdot
\min \Bigl\{
       \frac{1}{2^{65}},
       \frac{1}{16} - \frac{1}{2^{63}} - \frac{15}{2^8}
     \Bigr\}
\geq
\frac{1}{2^{66}}.
\]

This proves the inclusion 
\[
\QMA \subseteq \EPRQMA{2^{70}} \Bigl( 1, 1 - \frac{1}{2^{66}} \Bigr).
\]
Now for any constant~${s \in (0,1)}$,
one can achieve soundness~$s$ simply by repeating this proof system $t$~times in parallel
for some appropriate constant~$t$,
as the system is a special case of two-message quantum interactive proof systems,
for which parallel repetition works perfectly~\cite{KitWat00STOC}.
This completes the proof.
\end{proof}

Finally, we prove the following proposition.

\begin{proposition}
When entering Step~4 of the protocol described in Figure~\ref{Figure: verifier's EPR-QMA protocol for achieving perfect completeness},
suppose that the state in ${(\sfR_1, \sfR_2, \sfS_1, \sfS'_1, \sfS_2, \sfS'_2)}$ were of the form
${(\ketbra{0})^{\tensor 2} \tensor \tau}$
where
${\tau = \sum_j \mu_j \xi_j^{\tensor 2}}$
for some two-qubit states~$\xi_j$
and real numbers~${\mu_j \in [0,1]}$ satisfying ${\sum_j \mu_j = 1}$,
such that
the reduced state of $\tau$ in ${(\sfS_1, \sfS_2)}$
has trace distance at most $\delta$ to the two-qubit totally mixed state~${(I/2)^{\tensor 2}}$
for some positive $\delta$ satisfying
${
  \frac{1}{16} - 15 \delta^{\frac{1}{8}} > 0
}$.
Then the protocol should result in rejection with probability at least
${
  \min \bigl\{
         2 \delta,
         \frac{1}{16} - 15 \delta^{\frac{1}{8}}
       \bigr\}
}$.
\label{Proposition: soundness with sigma_1}
\end{proposition}

To prove Proposition~\ref{Proposition: soundness with sigma_1},
we first show two propositions that are special cases of Proposition~\ref{Proposition: soundness with sigma_1}.

\begin{proposition}
Let $\calW$ be the two-dimensional space spanned by $\ket{\Phi^-}$ and $\ket{\Psi^+}$.
%% Let $\calW$ be the space of dimension two spanned by $\ket{\Phi^-}$ and $\ket{\Psi^+}$.
When entering Step~4 of the protocol described in Figure~\ref{Figure: verifier's EPR-QMA protocol for achieving perfect completeness},
suppose that the state in ${(\sfR_1, \sfR_2, \sfS_1, \sfS'_1, \sfS_2, \sfS'_2)}$ were of the form
${(\ketbra{0})^{\tensor 2} \tensor \tau}$
where
${\tau = \sum_j \mu_j (\ketbra{\psi_j})^{\tensor 2}}$
for some two-qubit states~${\ket{\psi_j} \in \calW}$
and real numbers~${\mu_j \in [0,1]}$ satisfying ${\sum_j \mu_j = 1}$,
such that
the reduced state of $\tau$ in ${(\sfS_1, \sfS_2)}$
has trace distance at most $\delta$ to the two-qubit totally mixed state~${(I/2)^{\tensor 2}}$
for some positive $\delta$ satisfying
${\frac{1}{16} - \frac{\pi}{2} \delta^{\frac{1}{2}} > 0}$.
Then the protocol should result in rejection with probability at least
${\frac{1}{16} - \frac{\pi}{2} \delta^{\frac{1}{2}}}$.
\label{Proposition: soundness with sigma_3}
\end{proposition}

%% We first show the following lemma.
The following lemma is essential for the proof of Proposition~\ref{Proposition: soundness with sigma_3}.

\begin{lemma}
For each ${j \in \{1,2\}}$,
let $\calS_j$ and $\calS'_j$ be two-dimensional complex Hilbert spaces~${\Complex(\Sigma)}$,
and let $\calW_j$ be the two-dimensional subspace of ${\calS_j \tensor \calS'_j}$
spanned by $\ket{\Phi^-}$ and $\ket{\Psi^+}$.
Let $\rho$ be any four-qubit state in 
${
\Density(\calW_1 \tensor \calW_2) \subseteq \Density(\calS_1 \tensor \calS'_1 \tensor \calS_2 \tensor \calS'_2)
}$
that is a mixture of two-fold product pure states~${\ket{\zeta_j}^{\tensor 2}}$ in ${\calW_1 \tensor \calW_2}$
and such that
${
  D(\tr_{\calS'_1 \tensor \calS'_2} \rho, (I/2)^{\tensor 2}) \leq \delta
}$.
Then there exists a four-qubit state~$\sigma$
that is a mixture of two-fold products~$\ket{J(W^\pm_{a_j})}^{\tensor 2}$ of a Choi-Jamio{\l}kowski state, for real numbers~${a_j \in [0,1]}$,
such that
${
D(\rho, \sigma)
\leq
\frac{\pi}{2} \delta^{\frac{1}{2}}
}$,
where each $W^\pm_{a_j}$ is equal to either ${W^+_{a_j} = W_{a_j}}$ or ${W^-_{a_j} = Z W_{a_j} Z}$.
\label{Lemma: closeness to the mixture of 2-fold products of CJ states}
\end{lemma}

\begin{proof}
As $\rho$ is a mixture of two-fold product pure states in ${\calW_1 \tensor \calW_2}$,
it must be written as
\[
\rho = \sum_j \mu_j (\ketbra{\zeta_j})^{\tensor 2},
\]
where ${\ket{\zeta_j}^{\tensor 2}\! \in \calW_1 \tensor \calW_2}$, 
${\mu_j\! \in \![0,1]}$ for each $j$,
and  ${\sum_j \mu_j \!=\! 1}$.
Without loss of generality, one may assume that
\[
\ket{\zeta_j} = \alpha_j \ket{\Phi^-} + \beta_j e^{i \theta_j} \ket{\Psi^+}
\]
for each $j$,
where $\alpha_j$ and $\beta_j$ are real numbers in ${[0,1]}$ satisfying ${\alpha_j^2 + \beta_j^2 = 1}$,
and $\theta_j$ is a real number in ${[0, 2\pi)}$.
For each $j$, let ${a_j = \beta_j^2}$,
and define the two-qubit pure state~$\ket{\eta_j}$ as
\[
\ket{\eta_j}
=
\alpha_j \ket{\Phi^-} + \beta_j \ket{\Psi^+}
=
\sqrt{1 - a_j} \ket{\Phi^-} + \sqrt{a_j} \ket{\Psi^+}
=
\ket{J(W^+_{a_j})}
%% \ket{J(W_{a_j})}
\]
if ${j \in J_+}$,
and
\[
\ket{\eta_j}
=
\alpha_j \ket{\Phi^-} - \beta_j \ket{\Psi^+}
=
\sqrt{1 - a_j} \ket{\Phi^-} - \sqrt{a_j} \ket{\Psi^+}
=
\ket{J(W^-_{a_j})}
\]
if ${j \in J_-}$,
where
${J_+ = \set{j}{\theta_j \in [0, \pi/2] \cup [3 \pi/2, 2\pi)}}$
and
${J_- = \set{j}{\theta_j \in (\pi/2, 3 \pi/2)}}$.

Now take the four-qubit state~$\sigma$ as
\[
\sigma = \sum_j \mu_j (\ketbra{\eta_j})^{\tensor 2}.
\]
We shall show that this $\sigma$ has the desired property. 
For this purpose, we prove two claims.

\begin{claim}\label{lemma1205}
${
D \bigl( \tr_{\calS'_1 \tensor \calS'_2} \rho, (I/2)^{\tensor 2} \bigr) \geq 2 \sum_j \mu_j \alpha_j^2\beta_j^2\sin^2\theta_j
}$.
\end{claim}

\begin{proof}
Noticing that
\[
\begin{split}
\ket{\zeta_j}
&
= 
\frac{1}{\sqrt{2}}
\bigl[
  \alpha_j(\ket{00} - \ket{11}) + \beta_j e^{i \theta_j} (\ket{01} + \ket{10})
\bigr]
\\
&
=
\frac{1}{\sqrt{2}}
\bigl[
  (\alpha_j \ket{0} + \beta_j e^{i \theta_j} \ket{1}) \tensor \ket{0}
  +
  e^{i \theta_j} (\beta_j \ket{0} - \alpha_j e^{-i \theta_j} \ket{1}) \tensor \ket{1} 
\bigr],
\end{split}
\]
the reduced state~${\tr_{\calS'_1 \tensor \calS'_2} \rho}$
is the mixture of the following four states
\begin{align*}
&
(\alpha_j \ket{0} + \beta_j e^{i \theta_j} \ket{1}) \tensor (\alpha_j \ket{0} + \beta_j e^{i \theta_j} \ket{1}),
\\
&
(\alpha_j \ket{0} + \beta_j e^{i \theta_j} \ket{1}) \tensor (\beta_j \ket{0} - \alpha_j e^{-i \theta_j} \ket{1}),
\\
&
(\beta_j \ket{0} - \alpha_j e^{-i \theta_j} \ket{1}) \tensor (\alpha_j \ket{0} + \beta_j e^{i \theta_j} \ket{1}),
\\
&
(\beta_j \ket{0} - \alpha_j e^{-i \theta_j} \ket{1}) \tensor (\beta_j \ket{0} - \alpha_j e^{-i \theta_j} \ket{1})
\end{align*}
with equal probability~$1/4$ for each,
which can be expressed as a density matrix by 
\[
\frac{1}{4}
\begin{pmatrix}
1 & -2i\alpha_j\beta_js_j & -2i\alpha_j\beta_js_j & -4\alpha_j^2\beta_j^2s_j^2\\
2i\alpha_j\beta_js_j & 1 & 4\alpha_j^2\beta_j^2s_j^2 & -2i\alpha_j\beta_js_j\\ 
2i\alpha_j\beta_js_j & 4\alpha_j^2\beta_j^2s_j^2 & 1 & -2i\alpha_j\beta_js_j\\
 -4\alpha_j^2\beta_j^2s_j^2 & 2i\alpha_j\beta_js_j & 2i\alpha_j\beta_js_j & 1
\end{pmatrix},
\]
where $s_j$ is the shorthand of ${\sin\theta_j}$.
Let us denote the difference between ${\tr_{\calS'_1 \tensor \calS'_2} \rho}$ and ${(I/2)^{\tensor 2}}$ by $A$
(i.e., 
${
A = \tr_{\calS'_1 \tensor \calS'_2} \rho - (I/2)^{\tensor 2}
}$
).
In order to find the eigenvalues of $2A$, we solve the characteristic equation~${\abs{2A - \lambda I} = 0}$.
Straightforward calculations show that the four solutions of the equation~${\abs{2A - \lambda I} = 0}$ are given by
${-2 \sum_j \mu_j \alpha_j^2 \beta_j^2 s_j^2}$ (two-fold)
and
${
2 \sum_j \mu_j \alpha_j^2 \beta_j^2 s_j^2 \pm 2 \bigabs{\sum_j \mu_j \alpha_j \beta_j s_j}
}$.
This implies that
\[
\begin{split}
&
D \biggl( \tr_{\calS'_1 \tensor \calS'_2} \rho, \biggl( \frac{I}{2} \biggr)^{\tensor 2} \biggr)
=
\frac{1}{2} \tr \sqrt{\conjugate{A} A}
\\
&
\hspace{5mm}
=
\frac{1}{2}
\Biggl[
  2 \sum_j \mu_j \alpha_j^2 \beta_j^2 s_j^2
  +
  \biggl(
    \sum_j \mu_j \alpha_j^2 \beta_j^2 s_j^2
    +
    \biggabs{\sum_j \mu_j \alpha_j \beta_j s_j}
  \biggr)
  +
  \Biggabs{
    \sum_j \mu_j \alpha_j^2 \beta_j^2 s_j^2 
    -
    \biggabs{\sum_j \mu_j \alpha_j \beta_j s_j}
  }
\Biggr]
\\
&
\hspace{5mm}
=
\sum_j \mu_j \alpha_j^2 \beta_j^2 s_j^2
+
\max \biggl\{
       \sum_j \mu_j \alpha_j^2 \beta_j^2 s_j^2,
       \biggabs{\sum_j \mu_j \alpha_j \beta_j s_j}
     \biggr\},
\end{split}
\]
which is at least ${2 \sum_j \mu_j \alpha_j^2 \beta_j^2 s_j^2}$.
This completes the proof of the claim.
\end{proof}

\begin{claim}\label{lemma1205-2}
Let ${\{\mu_j\}}$ be a probability distribution, and ${\{c_j\}}$ be a set of real numbers. 
%%such that ${\abs{c_j} \leq 1}$.
If ${\sum_j \mu_j c_j^2 \leq \varepsilon}$,
it holds that ${\sum_j \mu_j \abs{c_j} \leq \varepsilon^{\frac{1}{2}}}$.
\end{claim}

\begin{proof}
By the Cauchy-Schwarz inequality, we have
\[
\sum_j \mu_j \abs{c_j}
=
\sum_j \sqrt{\mu_j} \cdot \sqrt{\mu_j} \abs{c_j}
\leq
\biggl(\sum_j \mu_j \biggr)^{\frac{1}{2}} \biggl(\sum_j \mu_j \abs{c_j}^2 \biggr)^{\frac{1}{2}}
\leq
\varepsilon^{\frac{1}{2}},
\]
as claimed.
\end{proof}

Now we bound ${D(\rho, \sigma)}$. 
Notice that 
\[
D(\rho,\sigma)
\leq
\sum_j \mu_j D \bigl( (\ketbra{\zeta_j})^{\tensor 2}, (\ketbra{\eta_j})^{\tensor 2} \bigr)
=
\sum_{j \in J_+} \mu_j \sqrt{1 - \bigabs{\braket{\zeta_j}{\eta_j}}^4}
+
\sum_{j \in J_-} \mu_j \sqrt{1 - \bigabs{\braket{\zeta_j}{\eta_j}}^4}.
%% \sum_{j \in J_+} \mu_j \sqrt{1 - \Bigl(\bigabs{\braket{\zeta_j}{\eta_j}}^2\Bigr)^2}
%% +
%% \sum_{j \in J_-} \mu_j \sqrt{1 - \Bigl(\bigabs{\braket{\zeta_j}{\eta_j}}^2\Bigr)^2}.
\]
If ${j \in J_+}$, it holds that 
\[
\bigabs{\braket{\zeta_j}{\eta_j}}^4
%% \Bigl(\bigabs{\braket{\zeta_j}{\eta_j}}^2\Bigr)^2
=
\bigabs{\alpha_j^2 +\beta_j^2 e^{-i \theta_j}}^4
%% \Bigl(
%%   \bigabs{\alpha_j^2 +\beta_j^2 e^{-i \theta_j}}^2
%% \Bigr)^2
=
\Bigl[
  \bigl( \alpha_j^2 + \beta_j^2 \cos \theta_j \bigr)^2
  +
  \bigl( \beta_j^2 \sin \theta_j \bigr)^2
\Bigr]^2
=
\Bigl( 1 - 4 \alpha_j^2 \beta_j^2 \sin^2 \frac{\theta_j}{2} 
\Bigr)^2,
\]
and thus,
\[
\sqrt{1 - \bigabs{\braket{\zeta_j}{\eta_j}}^4}
%% \sqrt{1 - \Bigl(\bigabs{\braket{\zeta_j}{\eta_j}}^2\Bigr)^2}
=
2 \sqrt{2}
\Bigabs{\alpha_j \beta_j \sin \frac{\theta_j}{2}}
\sqrt{1 - 2\alpha_j^2 \beta_j^2 \sin^2 \frac{\theta_j}{2}}
\leq
2 \sqrt{2} \Bigabs{\alpha_j \beta_j \sin \frac{\theta_j}{2}}
\leq
\frac{\pi}{\sqrt{2}} \bigabs{\alpha_j \beta_j \sin \theta_j},
\]
where the last inequality comes from the fact that 
for any ${\theta \in [0, \pi/2] \cup [3 \pi/2, 2\pi)}$,
${
  \bigabs{\sin \frac{\theta}{2}}
  \leq
  \bigabs{\frac{\theta}{2}}
  \leq
  \frac{\pi}{4} \abs{\sin \theta}
}$.

On the other hand, if ${j \in J_-}$, we have
\[
\bigabs{\braket{\zeta_j}{\eta_j}}^4
=
\bigabs{\alpha_j^2 -\beta_j^2 e^{-i \theta_j}}^4
=
\bigabs{\alpha_j^2 +\beta_j^2 e^{-i \theta'_j}}^4,
% \Bigl(\bigabs{\braket{\zeta_j}{\eta_j}}^2\Bigr)^2
% =
% \Bigl(
%   \bigabs{\alpha_j^2 -\beta_j^2 e^{-i \theta_j}}^2
% \Bigr)^2
% =
% \Bigl(
%   \bigabs{\alpha_j^2 +\beta_j^2 e^{-i \theta'_j}}^2
% \Bigr)^2,
\]
where ${\theta'_j = \theta_j + \pi \pmod{2\pi}}$. 
Noticing that ${\theta'_j \in [0, \pi/2] \cup [3 \pi/2, 2\pi)}$,
it holds that
\[
\sqrt{1 - \bigabs{\braket{\zeta_j}{\eta_j}}^4}
%% \sqrt{1 - \Bigl(\bigabs{\braket{\zeta_j}{\eta_j}}^2\Bigr)^2}
\leq
\frac{\pi}{\sqrt{2}}\bigabs{\alpha_j \beta_j \sin \theta'_j}.
\]

Therefore,
\[
D(\rho, \sigma) \leq \frac{\pi}{\sqrt{2}} \sum_j \mu_j \abs{c_j},
\]
where
\[
c_j
=
\begin{cases}
  \alpha_j \beta_j \sin \theta_j & \mbox{if ${j \in J_+}$},\\
  \alpha_j \beta_j \sin \theta'_j & \mbox{if ${j\in J_-}$}.
\end{cases}
\]
By Claim~\ref{lemma1205}
and the fact that ${\sin^2 \theta'_j = \sin^2 \theta_j}$ for each ${j \in J_-}$,
the assumption~${D \bigl( \tr_{\calS'_1 \tensor \calS'_2} \rho, (I/2)^{\tensor 2} \bigr) \leq \delta}$
implies that
%\begin{equation}\label{eq:1205-1}
\[
\sum_j \mu_j c_j^2
=
\sum_j \mu_j \alpha_j^2 \beta_j^2 \sin^2 \theta_j
\leq
\frac{\delta}{2}.
\]
%\end{equation}
By Claim~\ref{lemma1205-2}, this implies that
${
  D(\rho, \sigma)
  \leq
  \frac{\pi}{\sqrt{2}} (\frac{\delta}{2})^{\frac{1}{2}}
  =
  \frac{\pi}{2} \delta^{\frac{1}{2}}
}$,
which completes the proof of Lemma~\ref{Lemma: closeness to the mixture of 2-fold products of CJ states}.
\end{proof}

\begin{proof}[Proof of Proposition~\ref{Proposition: soundness with sigma_3}]
Let ${\sigma = (\ketbra{0})^{\tensor 2} \tensor \tau}$.
From Lemma~\ref{Lemma: closeness to the mixture of 2-fold products of CJ states},
there exists a quantum state~$\tau'$
that is a mixture of two-fold products~$\ket{J(W^\pm_{a_j})}^{\tensor 2}$ of a Choi-Jamio{\l}kowski state,
for real numbers~${a_j \in [0,1]}$,
such that, for ${\sigma' = (\ketbra{0})^{\tensor 2} \tensor \tau'}$,
${D(\sigma, \sigma') \leq \frac{\pi}{2} \delta^{\frac{1}{2}}}$.
%% \[
%% D(\sigma, \sigma') \leq \frac{\pi}{2} \delta^{\frac{1}{2}}.
%% \]
Here, as in Lemma~\ref{Lemma: closeness to the mixture of 2-fold products of CJ states},
each $W^\pm_{a_j}$ is equal to either ${W^+_{a_j} = W_{a_j}}$ or ${W^-_{a_j} = Z W_{a_j} Z}$.
From Proposition~\ref{Proposition: soundness of Reflection Simulation Test},
the \textsc{Reflection Simulation Test} should result in rejection with probability~$\frac{1}{16}$
if the quantum state in ${(\sfR_1, \sfR_2, \sfS_1, \sfS'_1, \sfS_2, \sfS'_2)}$ were $\sigma'$.
By Lemma~\ref{Lemma: trace distance and probability},
this implies that the \textsc{Reflection Simulation Test} should result in rejection
with probability at least ${\frac{1}{16} - \frac{\pi}{2} \delta^{\frac{1}{2}}}$
if the state in ${(\sfR_1, \sfR_2, \sfS_1, \sfS'_1, \sfS_2, \sfS'_2)}$ were $\sigma$.
Note that $\sigma$ is never rejected in Step~4
and passes the Swap-Test in Step~5 with certainty,
and the state is not changed at all in these two steps.
Hence, if the state in ${(\sfR_1, \sfR_2, \sfS_1, \sfS'_1, \sfS_2, \sfS'_2)}$
were $\sigma$ when entering Step~4,
the protocol should result in rejection with probability at least
${\frac{1}{16} - \frac{\pi}{2} \delta^{\frac{1}{2}}}$,
as claimed.
\end{proof}

We next show the following proposition,
which is more general than Proposition~\ref{Proposition: soundness with sigma_3},
but still is a special case of Proposition~\ref{Proposition: soundness with sigma_1}.

\begin{proposition}
Let $\calW$ be the two-dimensional space spanned by $\ket{\Phi^-}$ and $\ket{\Psi^+}$.
%% Let $\calW$ be the space of dimension two spanned by $\ket{\Phi^-}$ and $\ket{\Psi^+}$.
When entering Step~4 of the protocol described in Figure~\ref{Figure: verifier's EPR-QMA protocol for achieving perfect completeness},
suppose that the state in ${(\sfR_1, \sfR_2, \sfS_1, \sfS'_1, \sfS_2, \sfS'_2)}$ were of the form
${(\ketbra{0})^{\tensor 2} \tensor \tau}$
where
${\tau = \sum_j \mu_j {\xi_j}^{\tensor 2}}$
for some two-qubit states~${\xi_j \in \Density(\calW)}$
and real numbers~${\mu_j \in [0,1]}$ satisfying ${\sum_j \mu_j = 1}$,
such that
the reduced state of $\tau$ in ${(\sfS_1, \sfS_2)}$
has trace distance at most $\delta$ to the two-qubit totally mixed state~${(I/2)^{\tensor 2}}$
for some positive $\delta$ satisfying ${\frac{1}{16} - 10 \delta^{\frac{1}{4}} > 0}$.
Then the protocol should result in rejection with probability at least
${
  \min \bigl\{
         2 \delta,
         \frac{1}{16} - 10 \delta^{\frac{1}{4}}
       \bigr\}
}$.
\label{Proposition: soundness with sigma_2}
\end{proposition}

\begin{proof}
Let ${\sigma = (\ketbra{0})^{\tensor 2} \tensor \tau}$.
Note that $\sigma$ is never rejected in Step~4,
and the state is not changed at all in this step.

Fix a constant ${\gamma_1 \in (0,1)}$,
%% Fix a positive constant ${\gamma_1 \in (0,1)}$,
and let $S$ be the set of indices~$j$ defined by
\[
S = \set{j}{\tr {\xi_j}^2 \geq 1 - \gamma_1}.
\]
Notice that the inequality~${\tr {\xi_j}^2 \geq 1 - \gamma_1}$ implies that
the maximum eigenvalue of the Hermitian matrix~$\xi_j$ is at least ${1 - \gamma_1}$,
and thus, for each ${j \in S}$,
there exist a two-qubit pure state~${\ket{\psi_j} \in \calW}$,
a two-qubit state~${\nu_j \in \Density(\calW)}$,
and a real number~${\lambda_j \in [1 - \gamma_1, 1]}$
such that
\[
\xi_j = \lambda_j \ketbra{\psi_j} + (1 - \lambda_j) \nu_j.
\]
%% ${\xi_j = \lambda_j \ketbra{\psi_j} + (1 - \lambda_j) \nu_j}$.
This implies that
\[
\bigtrnorm{\xi_j -  \ketbra{\psi_j}}
=
\bigtrnorm{\lambda_j \ketbra{\psi_j} + (1 - \lambda_j) \nu_j - \ketbra{\psi_j}}
=
(1 - \lambda_j) \bigtrnorm{\nu_j - \ketbra{\psi_j}},
\]
which further implies that
\[
D(\xi_j, \ketbra{\psi_j})
\leq
(1 - \lambda_j) D(\nu_j, \ketbra{\psi_j})
\leq
1 - \lambda_j
\leq
\gamma_1.
\]

Fix another constant ${\gamma_2 \in (0,1)}$.
%% Fix another positive constant ${\gamma_2 \in (0,1)}$.

If ${\sum_{j \in S} \mu_j < 1 - \gamma_2}$,
the \textsc{Swap Test} in Step~5 results in rejection with probability greater than ${\frac{1}{2} \gamma_1 \gamma_2}$.

On the other hand, if ${\sum_{j \in S} \mu_j \geq 1 - \gamma_2}$,
the state~$\sigma$ has trace distance at most ${2 \gamma_1 + \gamma_2}$
to the state~${\sigma' = (\ketbra{0})^{\tensor 2} \tensor \tau'}$,
where
\[
\tau' = \frac{1}{\sum_{j \in S} \mu_j} \sum_{j \in S} \mu_j (\ketbra{\psi_j})^{\tensor 2}
\]
% \[
% \tau' = \frac{1}{\sum_{j \in S} \mu_j} \tau''
% \]
% with
% \[
% {\tau''}=\sum_{j \in S} \mu_j (\ketbra{\psi_j})^{\tensor 2},
% \]
and the reduced state of $\tau'$ in ${(\sfS_1, \sfS_2)}$
has trace distance at most ${\delta + 2 \gamma_1 + \gamma_2}$ to ${(I/2)^{\tensor 2}}$.

Indeed,
\[
\begin{split}
\trnorm{\tau - \tau'}
&
=
\biggtrnorm{
  \sum_j \mu_j \xi_j^{\tensor 2}
  -
  \tau'
%   \frac{1}{\sum_{j \in S} \mu_j} \sum_{j \in S} \mu_j (\ketbra{\psi_j})^{\tensor 2}
}
\\
&
\leq
\biggtrnorm{
  \sum_j \mu_j \xi_j^{\tensor 2}
  -
  \biggl(
    \sum_{j \in S} \mu_j (\ketbra{\psi_j})^{\tensor 2}
    +
    \sum_{j \not\in S} \mu_j \xi_j^{\tensor 2}
  \biggr)
}
\\
&
\hspace{5mm}
+
\biggtrnorm{
  \biggl(
    \sum_{j \in S} \mu_j (\ketbra{\psi_j})^{\tensor 2}
    +
    \sum_{j \not\in S} \mu_j \xi_j^{\tensor 2}
  \biggr)
  -
  \tau'
%   \frac{1}{\sum_{j \in S} \mu_j} \sum_{j \in S} \mu_j (\ketbra{\psi_j})^{\tensor 2}
}
\\
&
\leq
\sum_{j \in S} \mu_j \bigtrnorm{\xi_j^{\tensor 2} - (\ketbra{\psi_j})^{\tensor 2}}
+
\biggtrnorm{
  \sum_{j \not\in S} \mu_j \xi_j^{\tensor 2}
  -
  \biggl( \frac{1}{\sum_{j \in S} \mu_j} - 1 \biggr) \sum_{j \in S} \mu_j (\ketbra{\psi_j})^{\tensor 2}
}
\\
&
\leq
\sum_{j \in S}
  \mu_j
  \Bigl(
    \bigtrnorm{\xi_j^{\tensor 2} - \ketbra{\psi_j} \tensor \xi_j}
    +
    \bigtrnorm{\ketbra{\psi_j} \tensor \xi_j - (\ketbra{\psi_j})^{\tensor 2}}
  \Bigr)
\\
&
\hspace{5mm}
+
\biggl( 1 - \sum_{j \in S} \mu_j \biggr)
\biggtrnorm{
  \frac{1}{\sum_{j \not\in S} \mu_j} \sum_{j \not\in S} \mu_j \xi_j^{\tensor 2}
  -
  \tau'
%   \frac{1}{\sum_{j \in S} \mu_j} \sum_{j \in S} \mu_j (\ketbra{\psi_j})^{\tensor 2}
}
\\
&
\leq
2 \sum_{j \in S} \mu_j \bigtrnorm{\xi_j - \ketbra{\psi_j}}
%% \\
%% &
%% \hspace{5mm}
+
\biggl( 1 - \sum_{j \in S} \mu_j \biggr)
\biggtrnorm{
  \frac{1}{\sum_{j \not\in S} \mu_j} \sum_{j \not\in S} \mu_j \xi_j^{\tensor 2}
  -
  \tau'
%   \frac{1}{\sum_{j \in S} \mu_j} \sum_{j \in S} \mu_j (\ketbra{\psi_j})^{\tensor 2}
},
\end{split}
\]
and thus,
\[
\begin{split}
D(\sigma, \sigma')
&
=
D(\tau, \tau')
\\
&
\leq
2 \sum_{j \in S} \mu_j D(\xi_j, \ketbra{\psi_j})
+
\biggl( 1 - \sum_{j \in S} \mu_j \biggr)
D \biggl(
    \frac{1}{\sum_{j \not\in S} \mu_j} \sum_{j \not\in S} \mu_j \xi_j^{\tensor 2},
    \tau'
%     \frac{1}{\sum_{j \in S} \mu_j} \sum_{j \in S} \mu_j (\ketbra{\psi_j})^{\tensor 2}
  \biggr)
\\
&
\leq
2 \gamma_1 + \gamma_2.
\end{split}
\]
As the reduced state of $\tau$ in ${(\sfS_1, \sfS_2)}$
has trace distance at most $\delta$ to ${(I/2)^{\tensor 2}}$,
it follows that
the reduced state of $\tau'$ in ${(\sfS_1, \sfS_2)}$
has trace distance at most ${\delta + 2 \gamma_1 + \gamma_2}$ to ${(I/2)^{\tensor 2}}$.
Now from Proposition~\ref{Proposition: soundness with sigma_3},
the protocol should result in rejection with probability at least
${\frac{1}{16} - \frac{\pi}{2} (\delta + 2 \gamma_1 + \gamma_2)^{\frac{1}{2}}}$
if the state in ${(\sfR_1, \sfR_2, \sfS_1, \sfS'_1, \sfS_2, \sfS'_2)}$ were $\sigma'$
when entering Step~4.
Hence, from Lemma~\ref{Lemma: trace distance and probability},
the protocol should result in rejection with probability at least
${\frac{1}{16} - 2 \gamma_1 - \gamma_2 - \frac{\pi}{2} (\delta + 2 \gamma_1 + \gamma_2)^{\frac{1}{2}}}$
if the state in the registers ${(\sfR_1, \sfR_2, \sfS_1, \sfS'_1, \sfS_2, \sfS'_2)}$ were $\sigma$
when entering Step~4.

Overall, the protocol should result in rejection with probability at least
\[
\min \Bigl\{
       \frac{1}{2} \gamma_1 \gamma_2,
       \frac{1}{16} - 2 \gamma_1 - \gamma_2 - \frac{\pi}{2} (\delta + 2 \gamma_1 + \gamma_2)^{\frac{1}{2}}
     \Bigr\}
\]
if the state in ${(\sfR_1, \sfR_2, \sfS_1, \sfS'_1, \sfS_2, \sfS'_2)}$ were $\sigma$
when entering Step~4.
Taking
% ${\gamma_1 = (2 \delta)^{\frac{1}{2}}}$
% and
% ${\gamma_2 = (8 \delta)^{\frac{1}{2}}}$,
${\gamma_1 = \sqrt{2} \delta^{\frac{1}{2}}}$
and
${\gamma_2 = 2 \sqrt{2} \delta^{\frac{1}{2}}}$,
this is at least
\[
\min \Bigl\{
       2 \delta,
%%        \frac{1}{16} - (32 \delta)^{\frac{1}{2}} - \frac{\pi}{2} \bigl[ \delta + (32 \delta)^{\frac{1}{2}} \bigr]^{\frac{1}{2}}
       \frac{1}{16} - 4 \sqrt{2} \delta^{\frac{1}{2}} - \frac{\pi}{2} \bigl( \delta + 4 \sqrt{2} \delta^{\frac{1}{2}} \bigr)^{\frac{1}{2}}
     \Bigr\}
\geq
\min \Bigl\{
       2 \delta,
%%        \frac{1}{16} - (32 \delta)^{\frac{1}{2}} - \frac{\pi}{2} (45 \delta)^{\frac{1}{4}}
       \frac{1}{16} - 4 \sqrt{2} \delta^{\frac{1}{2}} - \frac{\pi}{2} \sqrt{7} \delta^{\frac{1}{4}}
     \Bigr\}
\geq
\min \Bigl\{
       2 \delta,
       \frac{1}{16} - 10 \delta^{\frac{1}{4}}
     \Bigr\},
\]
which completes the proof.
\end{proof}

Now we are ready to prove Proposition~\ref{Proposition: soundness with sigma_1}.

\begin{proof}[Proof of Proposition~\ref{Proposition: soundness with sigma_1}]
Let ${\sigma = (\ketbra{0})^{\tensor 2} \tensor \tau}$.
%% Consider the quantum state ${\sigma = (\ketbra{0})^{\tensor 2} \tensor \tau}$.
Let $\calW$ be the two-dimensional space spanned by $\ket{\Phi^-}$ and $\ket{\Psi^+}$,
and let
${\Pi_\calW = \ketbra{\Phi^-} + \ketbra{\Psi^+}}$
%% \[
%% \Pi_\calW = \ketbra{\Phi^-} + \ketbra{\Psi^+}
%% \]
be the projection onto $\calW$.

Fix a constant ${\gamma \in (0,1)}$.

If ${ \tr \Pi_\calW^{\tensor 2} \tau < 1 - \gamma}$,
Step~4 results in rejection with probability greater than $\gamma$.

On the other hand, if ${\tr \Pi_\calW^{\tensor 2} \tau \geq 1 - \gamma}$,
we claim that
the state~$\sigma$ has trace distance at most ${\sqrt{\gamma}}$
to the state~${\sigma' = (\ketbra{0})^{\tensor 2} \tensor \tau'}$,
where
\[
\tau' = \sum_j \mu'_j {\xi'_j}^{\tensor 2}
\]
with
\[
\mu'_j = \frac{1}{\tr \Pi_\calW^{\tensor 2} \tau} (\tr \Pi_\calW \xi_j)^2 \mu_j,
\quad
\xi'_j = \frac{1}{\tr \Pi_\calW \xi_j} \Pi_\calW \xi_j \Pi_\calW,
\]
for each $j$,
and the reduced state of $\tau'$ in ${(\sfS_1, \sfS_2)}$
has trace distance at most $\gamma$ to ${(I/2)^{\tensor 2}}$.
Note that ${\mu'_j \in [0,1]}$ and ${\xi'_j \in \Density(\calW)}$ for each $j$,
and ${\sum_j \mu'_j = 1}$.

Let $\calS$ be the $2^4$-dimensional Hilbert space~${\Complex(\Sigma^4)}$
associated with the quantum register~${(\sfS_1, \sfS'_1, \sfS_2, \sfS'_2)}$
and let $\calT$ be another $2^4$-dimensional Hilbert space~${\Complex(\Sigma^4)}$.
Consider any purification~${\ket{\psi} \in \calS \tensor \calT}$
of ${\tau \in \Density(\calS)}$,
and define an eight-qubit pure state~${\ket{\psi'} \in \calS \tensor \calT}$ by
\[
\ket{\psi'}
=
\frac{1}{\bignorm{(\Pi_\calW^{\tensor 2} \tensor I_\calT) \ket{\psi}}} (\Pi_\calW^{\tensor 2} \tensor I_\calT) \ket{\psi}.
\]
Then, $\ket{\psi'}$ is a purification of $\tau'$, since
\[
\begin{split}
\hspace{1cm} & \hspace{-1cm}
\tr_\calT \ketbra{\psi'}
=
\frac{1}{\bignorm{(\Pi_\calW^{\tensor 2} \tensor I_\calT) \ket{\psi}}^2}
\tr_\calT (\Pi_\calW^{\tensor 2} \tensor I_\calT) \ketbra{\psi} (\Pi_\calW^{\tensor 2} \tensor I_\calT)
\\
&
=
\frac{1}{\tr \Pi_\calW^{\tensor 2} \tr_\calT \ketbra{\psi}}
\Pi_\calW^{\tensor 2} (\tr_\calT \ketbra{\psi}) \Pi_\calW^{\tensor 2}
=
\frac{1}{\tr \Pi_\calW^{\tensor 2} \tau} \Pi_\calW^{\tensor 2} \tau \Pi_\calW^{\tensor 2}
=
\tau',
\end{split}
\]
where the last equality follows from the fact that
\[
\tau'
=
\sum_j
  \frac{1}{\tr \Pi_\calW^{\tensor 2} \tau} (\tr \Pi_\calW \xi_j)^2 \mu_j
  \Bigl( \frac{1}{\tr \Pi_\calW \xi_j} \Pi_\calW \xi_j \Pi_\calW \Bigr)^{\tensor 2}
=
\frac{1}{\tr \Pi_\calW^{\tensor 2} \tau}
\sum_j \mu_j (\Pi_\calW \xi_j \Pi_\calW)^{\tensor 2}
=
\frac{1}{\tr \Pi_\calW^{\tensor 2} \tau}
\Pi_\calW^{\tensor 2} \tau \Pi_\calW^{\tensor 2}.
\]
Therefore, by using the fact that
${
  D \bigl(\ketbra{\psi}, \ketbra{\psi'} \bigr)
  =
  \sqrt{1 - \abs{\braket{\psi}{\psi'}}^2}
}$
holds for any pure states~$\ket{\psi}$ and $\ket{\psi'}$
(which is ensured by calculating eigenvalues of the Hermitian matrix~${\ketbra{\psi} - \ketbra{\psi'}}$),
%% from Lemma~\ref{Lemma: trace distance and fidelity},
\[
\begin{split}
\hspace{1cm} & \hspace{-1cm}
D(\sigma, \sigma')
=
D(\tau, \tau')
\\
&
\leq
D \bigl(\ketbra{\psi}, \ketbra{\psi'} \bigr)
=
\sqrt{1 - \abs{\braket{\psi}{\psi'}}^2}
=
\sqrt{1 - \bignorm{(\Pi_\calW^{\tensor 2} \tensor I_\calT) \ket{\psi}}^2}
=
\sqrt{1 - \tr \Pi_\calW^{\tensor 2} \tau}
\leq
\sqrt{\gamma}.
\end{split}
\]
As the reduced state of $\tau$ in ${(\sfS_1, \sfS_2)}$
has trace distance at most $\delta$ to ${(I/2)^{\tensor 2}}$,
it follows that
the reduced state of $\tau'$ in ${(\sfS_1, \sfS_2)}$
has trace distance at most ${\delta + \sqrt{\gamma}}$ to ${(I/2)^{\tensor 2}}$.
Now from Proposition~\ref{Proposition: soundness with sigma_2},
the protocol should result in rejection with probability at least
${
  \min \bigl\{
         2 (\delta + \sqrt{\gamma}),
         \frac{1}{16}  - 10 (\delta + \sqrt{\gamma})^{\frac{1}{4}}
       \bigr\}
}$
if the state in ${(\sfR_1, \sfR_2, \sfS_1, \sfS'_1, \sfS_2, \sfS'_2)}$ were $\sigma'$
when entering Step~4.
Hence, Lemma~\ref{Lemma: trace distance and probability} implies that
the protocol should result in rejection with probability at least
${
  \min \bigl\{
         2 (\delta + \sqrt{\gamma}) - \sqrt{\gamma},
         \frac{1}{16} - \sqrt{\gamma} - 10 (\delta + \sqrt{\gamma})^{\frac{1}{4}}
       \bigr\}
}$
if the state in ${(\sfR_1, \sfR_2, \sfS_1, \sfS'_1, \sfS_2, \sfS'_2)}$ were $\sigma$
when entering Step~4.

Overall, the protocol should result in rejection with probability at least
\[
\min \Bigl\{
       \gamma,
       2 (\delta + \sqrt{\gamma}) - \sqrt{\gamma},
       \frac{1}{16} - \sqrt{\gamma} - 10 (\delta + \sqrt{\gamma})^{\frac{1}{4}}
     \Bigr\}
\]
if the state in ${(\sfR_1, \sfR_2, \sfS_1, \sfS'_1, \sfS_2, \sfS'_2)}$ were $\sigma$
when entering Step~4.
Taking ${\gamma = 2 \delta}$, this is at least
\[
\min \Bigl\{
       2 \delta,
       2 \delta + (2 \delta)^{\frac{1}{2}},
       \frac{1}{16} - (2 \delta)^{\frac{1}{2}} - 10 \bigl[ \delta + (2 \delta)^{\frac{1}{2}} \bigr]^{\frac{1}{4}}
     \Bigr\}
\geq
\min \Bigl\{
       2 \delta,
       \frac{1}{16} - (2 \delta)^{\frac{1}{2}} - 10 (3 \delta^{\frac{1}{2}})^{\frac{1}{4}}
     \Bigr\}
\geq
\min \Bigl\{
       2 \delta,
       \frac{1}{16} - 15 \delta^{\frac{1}{8}}
     \Bigr\},
\]
which completes the proof.
\end{proof}

% ---------------------------------------------------------------------------
%   QIP(m) is in QIP_1(m+1)
% ---------------------------------------------------------------------------

\section{$\boldsymbol{\QIP(m) \subseteq \QIP_1(m+1)}$}
%% \section{Achieving Perfect Completeness in QIPs with One Additional Message}
\label{Section: QIP(m) is in QIP_1(m+1)}

Now we show that any $m$-message QIP system with two-sided bounded error
can be converted into an ${(m+1)}$-message QIP system with one-sided error of perfect completeness,
for any ${m \geq 2}$.

\begin{theorem}
For any polynomially bounded function~$\function{m}{\Nonnegative}{\Natural}$
and polynomial-time computable functions~$\function{c,s}{\Nonnegative}{[0,1]}$
satisfying ${m \geq 2}$ and ${c - s \geq 1/p}$ for some polynomially bounded function~$\function{p}{\Nonnegative}{\Natural}$,
\[
\QIP(m, c, s) \subseteq \QIP \Bigl( m+1, 1, 1 - \frac{(c-s)^2}{16} \Bigr).
\]
\label{Theorem: QIP(m) is in QIP_1(m+1)}
\end{theorem}

%% For an odd~${m \geq 3}$,
If $m$ is an odd-valued function whose values are at least three,
we can show a stronger statement that any $m$-message QIP system with two-sided bounded error
can be converted into another $m$-message QIP system with one-sided error of perfect completeness.

\begin{theorem}
For any polynomially bounded odd-valued function~$\function{m}{\Nonnegative}{2\Natural + 1}$
and polynomial-time computable functions~$\function{c,s}{\Nonnegative}{[0,1]}$
satisfying ${m \geq 3}$ and ${c - s \geq 1/p}$ for some polynomially bounded function~$\function{p}{\Nonnegative}{\Natural}$,
\[
\QIP(m, c, s) \subseteq \QIP \Bigl(m, 1, 1 - \frac{(c-s)^2}{16} \Bigr).
\]
\label{Theorem: QIP(m) is in QIP_1(m) for odd m}
\end{theorem}

\begin{remark}
In fact, in Theorems~\ref{Theorem: QIP(m) is in QIP_1(m+1)}~and~\ref{Theorem: QIP(m) is in QIP_1(m) for odd m},
it is sufficient for the claims that the functions~$c$ and $s$ satisfy
${c - s \geq 2^{-p}}$ for some polynomially bounded function~$\function{p}{\Nonnegative}{\Natural}$.
\end{remark}

With the perfect parallel repetition theorem for general quantum interactive proofs~\cite{Gut09PhD},
the following corollaries immediately follow.

\begin{corollary}
For any polynomially bounded functions~$\function{m,p}{\Nonnegative}{\Natural}$
and polynomial-time computable functions~$\function{c,s}{\Nonnegative}{[0,1]}$
satisfying ${m \geq 2}$ and ${c - s \geq 1/q}$ for some polynomially bounded function~$\function{q}{\Nonnegative}{\Natural}$,
\[
\QIP(m, c, s) \subseteq \QIP(m+1, 1, 2^{-p}).
\]
\label{Corollary: QIP(m) is in QIP_1(m+1) with amplification}
\end{corollary}

\begin{corollary}
For any polynomially bounded odd-valued function~$\function{m}{\Nonnegative}{2\Natural + 1}$,
polynomially bounded function~$\function{p}{\Nonnegative}{\Natural}$,
and polynomial-time computable functions~$\function{c,s}{\Nonnegative}{[0,1]}$
satisfying ${m \geq 3}$ and ${c - s \geq 1/q}$ for some polynomially bounded function~$\function{q}{\Nonnegative}{\Natural}$,
\[
\QIP(m, c, s) \subseteq \QIP(m, 1, 2^{-p}).
\]
\label{Corollary: QIP(m) is in QIP_1(m) for odd m with amplification}
\end{corollary}

% ---------------------------------------------------------------------------
%   Modified Reflection Procedure
% ---------------------------------------------------------------------------

\subsection{Modified Reflection Procedure}
\label{Subsection: Modified Reflection Procedure}

The \textsc{Reflection Procedure} in Section~\ref{Section: Reflection Procedure}
involves one application of $U$ and one application of $\conjugate{U}$.
Here we modify the procedure so that it involves one application of $\conjugate{U}$ only (and no application of $U$ is required).

To do this, one expects to receive a state just after Step~1 of the \textsc{Reflection Procedure},
and performs two tests,
called \textsc{Reflection Test} and \textsc{Invertibility Test}, respectively,
with equal probability without revealing which test the prover is undergoing.
In the \textsc{Reflection Test},
we simply perform Steps~2--4 of the \textsc{Reflection Procedure}
to finish the simulation of it,
whereas in the \textsc{Invertibility Test},
we apply $\conjugate{U}$ without performing the phase-flip
to check that the state received was a legal state that can appear
just after Step~1 of the \textsc{Reflection Procedure}.
The idea of making use of the \textsc{Invertibility Test}
has originally appeared in Ref.~\cite{KemKobMatVid09CC} when achieving perfect completeness in quantum multi-prover interactive proofs.
From another viewpoint, the modification here may be considered as applying
the ``halving technique'' in Ref.~\cite{KemKobMatVid09CC}
to the \textsc{Reflection Procedure},
the technique originally used to reduce the number of turns by (almost) half in quantum multi-prover interactive proofs.
We will take this view when analyzing the soundness of this procedure in Proposition~\ref{Proposition: soundness of Modified Reflection Procedure} below.
The procedure is summarized in Figure~\ref{Figure: Modified Reflection Procedure}.

\begin{figure}[t!]
\begin{algorithm*}{\textsc{Modified Reflection Procedure}}
\begin{step}
\item
  Receive a quantum register~$\sfQ$.
  Flip a fair coin,
  and proceed to the \textsc{Reflection Test} in Step~2 if it results in ``Heads'',
  and proceed to the \textsc{Invertibility Test} in Step~3 if it results in ``Tails''.
\item
  (\textsc{Reflection Test})\\
  Perform the following:
  \begin{step}
  \item
    Perform a phase-flip (i.e., multiply $-1$ in phase)
    if the state in $\sfQ$ belongs to the subspace corresponding to the projection~$\Pi_0$.
  \item
    Apply $\conjugate{U}$ to $\sfQ$.
  \item
    Reject if the state in $\sfQ$ belongs to the subspace corresponding to the projection~$\Delta_0$,
    and accept otherwise.
%%     Measure the state in $\sfQ$ with respect to the projective measurement~$\{ \Delta_0, \Delta_1 \}$.
%%     Reject if the state is projected onto the space corresponding to $\Delta_0$,
%%     and accept otherwise.
  \end{step}
\item
  (\textsc{Invertibility Test})\\
  Perform the following:
  \begin{step}
  \item
    Apply $\conjugate{U}$ to $\sfQ$.
  \item
    Accept if the state in $\sfQ$ belongs to the subspace corresponding to $\Delta_0$,
    and reject otherwise.
%%     Measure the state in $\sfQ$ with respect to the projective measurement~$\{ \Delta_0, \Delta_1 \}$.
%%     Reject if the state is projected onto the space corresponding to $\Delta_1$,
%%     and accept otherwise.
  \end{step}
\end{step} 
\end{algorithm*}
\caption{The \textsc{Modified Reflection Procedure}.}
\label{Figure: Modified Reflection Procedure}
\end{figure}

\begin{proposition}
Suppose that the Hermitian operator~${M = \Delta_0 \conjugate{U} \Pi_0 U \Delta_0}$
has an eigenvalue $1/2$.
Then there exists a quantum state given in Step~1 of the \textsc{Modified Reflection Procedure}
such that the procedure results in acceptance with certainty.
\label{Proposition: completeness of Modified Reflection Procedure}
\end{proposition}

\begin{proof}
The proof is almost straightforward.
Let $\ket{\psi^\ast}$ be an eigenvector of $M$ corresponding to its eigenvalue~$1/2$,
and consider the case where the state~${U \ket{\psi^\ast}}$ is received in $\sfQ$ in Step~1.

If the \textsc{Reflection Test} is performed,
this essentially simulates the original \textsc{Reflection Procedure}
with its received state being $\ket{\psi^\ast}$.
As in the case of Proposition~\ref{Proposition: completeness of Reflection Procedure},
the procedure results in acceptance with certainty in this case.

On the other hand, if the \textsc{Invertibility Test} is performed,
this produces the state~${\conjugate{U} U \ket{\psi^\ast} = \ket{\psi^\ast}}$
when entering Step~3.2.
As $\ket{\psi^\ast}$ is an eigenvector of $M$ with its corresponding eigenvalue~$1/2$, it holds that
\[
\Delta_0 \ket{\psi^\ast} = 2 \Delta_0 M \ket{\psi^\ast} = 2 M \ket{\psi^\ast} = \ket{\psi^\ast},
\]
and thus, Step~3.2 results in acceptance with certainty.

Hence, given the state~${U \ket{\psi^\ast}}$ in Step~1,
the procedure results in acceptance with certainty,
and the claim follows.
\end{proof}

\begin{proposition}
For any ${\varepsilon \in (0, \frac{1}{2}]}$,
% For any ${\varepsilon > 0}$,
suppose that none of the eigenvalues of the Hermitian operator~${M = \Delta_0 \conjugate{U} \Pi_0 U \Delta_0}$
is in the interval~${\bigl( \frac{1}{2} - \varepsilon, \frac{1}{2} + \varepsilon \bigr)}$.
Then, for any quantum state given in Step~1 of the \textsc{Modified Reflection Procedure},
the procedure results in rejection with probability at least $\varepsilon^2$.
\label{Proposition: soundness of Modified Reflection Procedure}
\end{proposition}

\begin{proof}
The proof is similar to the proofs of Lemmas~4.1~and~5.1 in Ref.~\cite{KemKobMatVid09CC}.
Let $\ket{\psi}$ be any state received in $\sfQ$ in Step~1.
Denote the unitary transformation~${\conjugate{U} (- \Pi_0 + \Pi_1)}$ by $V$,
and let
\[
\ket{\alpha}
=
\frac{\Delta_1 V \ket{\psi}}{\norm{\Delta_1 V \ket{\psi}}},
%% \frac{\Delta_1 \conjugate{U} (-\Pi_0 + \Pi_1) \ket{\psi}}{\norm{\Delta_1 \conjugate{U} (-\Pi_0 + \Pi_1) \ket{\psi}}},
\quad
\ket{\beta}
=
\frac{\Delta_0 \conjugate{U} \ket{\psi}}{\norm{\Delta_0 \conjugate{U} \ket{\psi}}}.
\]
Then
\[
\norm{\Delta_1 V \ket{\psi}}
=
\frac{1}{\norm{\Delta_1 V \ket{\psi}}}
\bigabs{\bra{\psi} \conjugate{V} \Delta_1 V \ket{\psi}}
=
F \bigl( \ketbra{\alpha}, V \ketbra{\psi} \conjugate{V} \bigr)
=
F \bigl( \conjugate{V} \ketbra{\alpha} V, \ketbra{\psi} \bigr),
\]
and thus, the probability~$p_1$ of acceptance when the \textsc{Reflection Test} is performed is given by
\[
p_1 = F \bigl( \conjugate{V} \ketbra{\alpha} V, \ketbra{\psi} \bigr)^2.
\]
Similarly, the probability~$p_2$ of acceptance when the \textsc{Invertibility Test} is performed is given by
\[
p_2 = F \bigl( U \ketbra{\beta} \conjugate{U}, \ketbra{\psi} \bigr)^2.
\]
Hence, the probability~$p_\acc$ of acceptance when the received state in Step~1 was $\ket{\psi}$ is given by
\[
p_\acc
=
\frac{1}{2} (p_1 + p_2)
=
\frac{1}{2}
\Bigl(
  F \bigl( \conjugate{V} \ketbra{\alpha} V, \ketbra{\psi} \bigr)^2
  +
  F \bigl( U \ketbra{\beta} \conjugate{U}, \ketbra{\psi} \bigr)^2
\Bigr).
\]
It follows from Lemma~\ref{Lemma: F(a,b)^2 + F(b,c)^2 < 1 + F(a,c)} that
\[
p_\acc
\leq
\frac{1}{2} \Bigl( 1 + F \bigl( \conjugate{V} \ketbra{\alpha} V, U \ketbra{\beta} \conjugate{U} \bigr) \Bigr)
=
\frac{1}{2} \Bigl( 1 + F \bigl( \ketbra{\alpha}, V U \ketbra{\beta} \conjugate{U} \conjugate{V} \bigr) \Bigr).
\]

Now notice that $\ket{\beta}$ is a state in $\calX_0$,
and thus,
\[
\norm{\Delta_1 VU \ket{\beta}}^2 \leq 1 - 4 \varepsilon^2,
\]
since
${
\norm{\Delta_0 VU \ket{\beta}}^2 \geq 4 \varepsilon^2
}$
from the analysis on the \textsc{Reflection Procedure} in the proof of Proposition~\ref{Proposition: soundness of Reflection Procedure}.
Hence, using ${\Delta_1 \ket{\alpha} = \ket{\alpha}}$,
\[
F \bigl( \ketbra{\alpha}, V U \ketbra{\beta} \conjugate{U} \conjugate{V} \bigr)
=
\bigabs{\bra{\alpha} V U \ket{\beta}}
=
\bigabs{\bra{\alpha} \Delta_1 V U \ket{\beta}}
\leq
\bignorm{\Delta_1 V U \ket{\beta}}
\leq
\sqrt{1 - 4 \varepsilon^2},
\]
and thus,
\[
p_\acc
\leq
\frac{1}{2} + \frac{\sqrt{1 - 4 \varepsilon^2}}{2}
\leq
\frac{1}{2} + \frac{1 - 2 \varepsilon^2}{2}
=
1 - \varepsilon^2.
\] 
Therefore, the procedure results in rejection with probability at least $\varepsilon^2$, as claimed.
\end{proof}

% ---------------------------------------------------------------------------
%   Perfectly Rewindable QIPs
% ---------------------------------------------------------------------------

\subsection{Perfectly Rewindable QIPs}
\label{Subsection: Perfectly Rewindable QIPs}

Here we introduce the notion of \emph{perfectly rewindable} QIP systems.
The concept of perfectly rewindable systems was originally introduced
for quantum multi-prover interactive proofs in Ref.~\cite{KemKobMatVid09CC},
and the notion here is the single-prover version of it as a special case.

\begin{definition}
Given a polynomially bounded function~$\function{m}{\Nonnegative}{\Natural}$
and a function~$\function{s}{\Nonnegative}{[0,1]}$ satisfying ${s < \frac{1}{2}}$,
a promise problem~${A = \{ A_\yes, A_\no \}}$ has a perfectly rewindable $m$-message quantum interactive proof system with soundness~$s$
iff there exists an $m$-message polynomial-time quantum verifier~$V$
such that, for every input~$x$:
\begin{description}
\item[\textnormal{(Perfect Rewindability)}]
  if ${x \in A_\yes}$,
  there exists an $m$-message quantum prover~$P$
  such that the maximum probability that $V$ accepts $x$ when communicating with $P$ is exactly $1/2$,
  where the maximum is taken over all possible initial states~$\rho_x$ of $P$,
\item[\textnormal{(Soundness)}]
  if ${x \in A_\no}$,
  for any $m$-message quantum prover~$P'$
  and any initial state~$\rho'_x$ of $P'$ prepared,
  $V$ accepts $x$
  with probability at most ${s(\abs{x})}$.
\end{description}
\label{Definition: perfectly rewindable systems}
\end{definition}
Note that in the perfect rewindability property
we first fix the transformations of the prover,
and then maximize over all legal initial states,
which hence have a fixed dimension.
We first show how to modify any general QIP system
to a perfectly rewindable one without changing the number of messages.

\begin{lemma}
Let $\function{m}{\Nonnegative}{\Natural}$ be a polynomially bounded function
and let $\function{c,s}{\Nonnegative}{[0,1]}$ be polynomial-time computable functions
satisfying ${c - s \geq 1/p}$ for some polynomially bounded function~$\function{p}{\Nonnegative}{\Natural}$.
Then, any promise problem~${A = (A_\yes, A_\no)}$ in ${\QIP(m, c, s)}$
has a perfectly rewindable $m$-message quantum interactive proof system
with soundness~${\frac{1}{2} - \frac{c-s}{4}}$.
\label{Lemma: making QIP perfectly rewindable}
\end{lemma}

\begin{proof}
Let ${A = (A_\yes, A_\no)}$ be a problem in ${\QIP(m, c, s)}$
and let $V$ be the corresponding $m$-message quantum verifier.
We first modify $V$ to obtain another $m$-message quantum verifier $V'$
that witnesses the inclusion~${A \in \QIP \big( m, \frac{1}{2} + \frac{c-s}{4}, \frac{1}{2} - \frac{c-s}{4} \bigr)}$.
This can be done via a standard technique as follows.
Fix an input~$x$.
The new verifier~$V'$ behaves in a manner exactly same as $V$,
except for the acceptance condition.
If ${c(\abs{x}) + s(\abs{x}) \geq 1}$,
$V'$ accepts with probability~$\frac{1}{c(\abs{x}) + s(\abs{x})}$
when the final state in the system would make $V$ accept
(and reject otherwise).
Thus, $V'$ accepts ${x \in A_\yes}$ with probability at least
% \[
%   \frac{c(\abs{x})}{c(\abs{x}) + s(\abs{x})}
%   =
%   \frac{1}{2} \Bigl( 1 + \frac{c(\abs{x}) - s(\abs{x})}{c(\abs{x}) + s(\abs{x})} \Bigr),
% \]
${
  \frac{c(\abs{x})}{c(\abs{x}) + s(\abs{x})}
  =
  \frac{1}{2} \bigl( 1 + \frac{c(\abs{x}) - s(\abs{x})}{c(\abs{x}) + s(\abs{x})} \bigr)
}$,
while accepts ${x \in A_\no}$ with probability at most
% \[
%   \frac{s(\abs{x})}{c(\abs{x}) + s(\abs{x})}
%   =
%   \frac{1}{2} \Bigl( 1 - \frac{c(\abs{x}) - s(\abs{x})}{c(\abs{x}) + s(\abs{x})} \Bigr).
% \]
${
  \frac{s(\abs{x})}{c(\abs{x}) + s(\abs{x})}
  =
  \frac{1}{2} \bigl( 1 - \frac{c(\abs{x}) - s(\abs{x})}{c(\abs{x}) + s(\abs{x})} \bigr)
}$.
Similarly, if ${c(\abs{x}) + s(\abs{x}) < 1}$,
letting ${\varepsilon(\abs{x}) = 1 - c(\abs{x})}$ and ${\delta(\abs{x}) = 1 - s(\abs{x})}$,
$V'$ rejects with probability~${\frac{1}{\varepsilon(\abs{x}) + \delta(\abs{x})} = \frac{1}{2 - c(\abs{x}) - s(\abs{x})}}$
when the final state in the system would make $V$ reject
(and accept otherwise).
Thus, $V'$ rejects ${x \in A_\yes}$ with probability at most
% \[
% \frac{\varepsilon(\abs{x})}{\varepsilon(\abs{x}) + \delta(\abs{x})}
% =
% \frac{1}{2} \Bigl( 1 - \frac{\delta(\abs{x}) - \varepsilon(\abs{x})}{\varepsilon(\abs{x}) + \delta(\abs{x})} \Bigr)
% =
% \frac{1}{2} \Bigl( 1 - \frac{c(\abs{x}) - s(\abs{x})}{2 - c(\abs{x}) - s(\abs{x})} \Bigr),
% \]
${
  \frac{\varepsilon(\abs{x})}{\varepsilon(\abs{x}) + \delta(\abs{x})}
  =
  \frac{1}{2} \bigl( 1 - \frac{\delta(\abs{x}) - \varepsilon(\abs{x})}{\varepsilon(\abs{x}) + \delta(\abs{x})} \bigr)
  =
  \frac{1}{2} \bigl( 1 - \frac{c(\abs{x}) - s(\abs{x})}{2 - c(\abs{x}) - s(\abs{x})} \bigr)
}$,
while $V'$ rejects ${x \in A_\no}$ with probability at least
% \[
% \frac{\delta(\abs{x})}{\varepsilon(\abs{x}) + \delta(\abs{x})}
% =
% \frac{1}{2} \Bigl( 1 + \frac{\delta(\abs{x}) - \varepsilon(\abs{x})}{\varepsilon(\abs{x}) + \delta(\abs{x})} \Bigr)
% =
% \frac{1}{2} \Bigl( 1 + \frac{c(\abs{x}) - s(\abs{x})}{2 - c(\abs{x}) - s(\abs{x})} \Bigr).
% \]
${
  \frac{\delta(\abs{x})}{\varepsilon(\abs{x}) + \delta(\abs{x})}
  =
  \frac{1}{2} \bigl( 1 + \frac{\delta(\abs{x}) - \varepsilon(\abs{x})}{\varepsilon(\abs{x}) + \delta(\abs{x})} \bigr)
  =
  \frac{1}{2} \bigl( 1 + \frac{c(\abs{x}) - s(\abs{x})}{2 - c(\abs{x}) - s(\abs{x})} \bigr)
}$.
Taking it into account that,
with a given finite-size gate set available for the verifier,
it may not be possible to accept
with probability exactly
$\frac{1}{c(\abs{x}) + s(\abs{x})}$
in the case~${c(\abs{x}) + s(\abs{x}) \geq 1}$,
or to reject with probability exactly
${\frac{1}{\varepsilon(\abs{x}) + \delta(\abs{x})} = \frac{1}{2 - c(\abs{x}) - s(\abs{x})}}$
in the case~${c(\abs{x}) + s(\abs{x}) < 1}$,
we actually consider another verifier~$V''$ who approximately performs the transformations of $V'$ with sufficient accuracy,
where the transformations of $V''$ are exactly implementable with the given finite-size gate set available for the verifier.
As both ${c(\abs{x}) + s(\abs{x})}$ and ${2 - c(\abs{x}) - s(\abs{x})}$ are at most ${2 - \frac{1}{p}}$,
the bounds obtained above are sufficient to claim that
the $m$-message system with the verifier~$V''$
has completeness~${\frac{1}{2} + \frac{c-s}{4}}$
and soundness~${\frac{1}{2} - \frac{c-s}{4}}$.

The rest of the proof is essentially the same as the proof of Lemma~3.2 in Ref.~\cite{KemKobMatVid09CC}.
We further modify $V''$ to construct another $m$-message quantum verifier $W$
for a perfectly rewindable proof system for $A$.
The new verifier $W$ prepares a single-qubit register~$\sfB$
in addition to the register~$\sfV$ which corresponds to the space used by $V''$.
The qubit in $\sfB$ is initialized to $\ket{0}$.
$W$ behaves exactly in the same manner as $V''$ does,
except that, in addition to all actions $V''$ would do,
$W$ also sends $\sfB$ to the prover in the last message from the verifier
and receives $\sfB$ from the prover in the last message from the prover.
As for the final decision,
$W$ accepts if and only if the content of $\sfV$ would make $V''$ accept
\emph{and} $\sfB$ contains $1$.
Notice that $W$ accepts only if $V''$ would accept,
and thus, the soundness is obviously at most ${\frac{1}{2} + \frac{c-s}{4}}$.

For perfect rewindability, we slightly modify the protocol of the honest prover in the case ${x \in A_\yes}$.
Given a protocol of the honest prover~$P$ in the system with $V''$
and an initial state~$\ket{\psi_\init}$ in the system with $V''$
that achieves the maximal acceptance probability~$p_{\max}$
when $V''$ communicating with this $P$,
we construct a protocol of the honest prover~$Q$ in the system with $W$ as follows.
$Q$ uses $\ket{\psi_\init}$ as the initial state
and behaves exactly in the same manner as $P$ does,
except that, upon receiving the last message from $W$,
$Q$ applies to the qubit in $\sfB$
the one-qubit unitary transformation~$U$ satisfying
\[
U \colon \ket{0} \mapsto \sqrt{1 - \frac{1}{2 p_{\max}}} \ket{0} + \sqrt{\frac{1}{2 p_{\max}}} \ket{1},
\]
in addition to all what the original $P$ would do. 
From the construction it is obvious that the maximum accepting probability of $W$ when communicating with $Q$
is exactly equal to $\frac{1}{2}$
and that this maximum is achieved when $Q$ uses $\ket{\psi_\init}$ as the initial state.
Finally, as the transformations of $V''$ are exactly implementable with the given finite-size gate set available for the verifier,
so are the transformations of $W$.
\end{proof}

\begin{remark}
In fact, in Lemma~\ref{Lemma: making QIP perfectly rewindable},
it is sufficient for the claim that the functions~$c$ and $s$ satisfy
${c - s \geq 2^{-p}}$ for some polynomially bounded function~$\function{p}{\Nonnegative}{\Natural}$.
\end{remark}

% ---------------------------------------------------------------------------
%   Proof of Theorem QIP(m) is in QIP_1(m+1)
% ---------------------------------------------------------------------------

\subsection{Proofs of Theorems~\ref{Theorem: QIP(m) is in QIP_1(m+1)}~and~\ref{Theorem: QIP(m) is in QIP_1(m) for odd m}}
\label{Subsection: Proofs of Theorems QIP(m) is in QIP_1(m+1) and QIP(m) is in QIP_1(m) for odd m}

Now we are ready to show Theorems~\ref{Theorem: QIP(m) is in QIP_1(m+1)}~and~\ref{Theorem: QIP(m) is in QIP_1(m) for odd m}.
First we prove Theorem~\ref{Theorem: QIP(m) is in QIP_1(m) for odd m},
assuming that $m$ is an odd-valued function and ${m \geq 3}$.
The case of general $m$ is proved in the same manner as this special case,
except that the number of messages increases by one when ${m(\abs{x})}$ is even,
which gives Theorem~\ref{Theorem: QIP(m) is in QIP_1(m+1)}.
%% First we prove Theorem~\ref{Theorem: QIP(m) is in QIP_1(m+1)}.
%% From the proof of Theorem~\ref{Theorem: QIP(m) is in QIP_1(m+1)},
%% we shall see that Theorem~\ref{Theorem: QIP(m) is in QIP_1(m) for odd m} follows as a special case of Theorem~\ref{Theorem: QIP(m) is in QIP_1(m+1)}.

\begin{proof}[Proof of Theorem~\ref{Theorem: QIP(m) is in QIP_1(m) for odd m}]
As $m$ is an odd-valued function and ${m \geq 3}$,
there is a polynomially bounded function~$\function{r}{\Nonnegative}{\Natural}$
such that ${m = 2r + 1}$.
Let ${A = (A_\yes, A_\no)}$ be in ${\QIP(2r+1,c,s)}$.
Then from Lemma~\ref{Lemma: making QIP perfectly rewindable},
$A$ has a perfectly rewindable ${(2r+1)}$-message quantum interactive proof system
with soundness ${\frac{1}{2} - \frac{c-s}{4}}$.
Let $V$ be the verifier of this perfectly rewindable ${(2r+1)}$-message quantum interactive proof system.
We construct another ${(2r+1)}$-message quantum verifier~$W$
of a new quantum interactive proof system for $A$.

Fix an input~$x$.
Let $\sfV$ be the quantum register consisting of private qubits used by the original verifier~$V$,
and let $\sfM$ be the quantum register consisting of qubits used for communications in the original proof system.
Let $V_{x,j}$ be the $j$th transformation of $V$,
for each ${j \in \{1, \ldots, r(\abs{x})+1 \}}$,
acting over ${(\sfV, \sfM)}$.
The new verifier~$W$ uses the same registers~$\sfV$~and~$\sfM$ as the original verifier~$V$.
$W$ first receives the two registers~$\sfV$~and~$\sfM$,
expecting that the state in ${(\sfV, \sfM)}$ forms what $V$ would have
after the last message from a prover had been received in the original proof system.
$W$ then performs one of the two tests,
called \textsc{Reflection Test} and \textsc{Invertibility Test},
chosen uniformly at random.
In the \textsc{Reflection Test},
$W$ first performs a phase-flip if the state in ${(\sfV, \sfM)}$ would cause $V$
to accept when the last transformation~$V_{x, r(\abs{x})+1}$ of $V$ was performed,
and then moves to a backward simulation of the original system.
$W$ accepts when the backward simulation \emph{does not} produce a legal initial state of the original system.
In the \textsc{Invertibility Test},
$W$ just immediately moves to a backward simulation of the original system.
This time, $W$ accepts when the backward simulation \emph{does} produce a legal initial state of the original system.
The exact protocol is described in Figure~\ref{Figure: Verifier's protocol for achieving perfect completeness for odd m}.
Notice that the number of messages in this system is indeed
${1 + 1 + 2(r(\abs{x}) - 1) + 1 = 2r(\abs{x}) + 1 = m(\abs{x})}$.

\begin{figure}[t!]
\begin{protocol*}{Verifier's Protocol for Achieving Perfect Completeness (Odd-Number-Message Case)}
\begin{step}
\item
  Receive quantum registers~$\sfV$~and~$\sfM$.
\item
  Choose ${b \in \Binary}$ uniformly at random.
  If ${b=0}$, move to the \textsc{Reflection Test}
  described in Step~\ref{Reflection Test},
  while if ${b=1}$, move to the \textsc{Invertibility Test}
  described in Step~\ref{Invertibility Test}.
\item
  (\textsc{Reflection Test})
  \begin{step}
  \item
    Apply $V_{x, r(\abs{x})+1}$ to the state in ${(\sfV, \sfM)}$.
    Perform a phase-flip (i.e., multiply $-1$ in phase)
    if the content of ${(\sfV, \sfM)}$ corresponds to an accepting state of the original system.
    Apply $\conjugate{V_{x, r(\abs{x})+1}}$ to the state in ${(\sfV, \sfM)}$,
    and send $\sfM$ to the prover.
  \item
    For ${j = r(\abs{x})}$ down to $2$, do the following:\\
    Receive $\sfM$ from the prover.
    Apply $\conjugate{V_{x, j}}$ to the state in ${(\sfV, \sfM)}$,
    and send $\sfM$ to the prover.
  \item
    Receive $\sfM$ from the prover.
    Apply $\conjugate{V_{x, 1}}$ to the state in ${(\sfV, \sfM)}$.
    Reject if all the qubits in $\sfV$ are in state~$\ket{0}$,
    and accept otherwise.
  \end{step}
  \label{Reflection Test}
\item
  (\textsc{Invertibility Test})
  \begin{step}
  \item
    Send $\sfM$ to the prover.
  \item
    For ${j = r(\abs{x})}$ down to $2$, do the following:\\
    Receive $\sfM$ from the prover.
    Apply $\conjugate{V_{x, j}}$ to the state in ${(\sfV, \sfM)}$,
    and send $\sfM$ to the prover.
  \item
    Receive $\sfM$ from the prover.
    Apply $\conjugate{V_{x, 1}}$ to the state in ${(\sfV, \sfM)}$.
    Accept if all the qubits in $\sfV$ are in state~$\ket{0}$,
    and reject otherwise.
  \end{step}
  \label{Invertibility Test}
\end{step}
\end{protocol*}
\caption{Verifier's protocol for achieving perfect completeness with ${m = 2r + 1}$.}
%% \caption{Verifier's protocol for achieving perfect completeness with ${m(\abs{x})}$ being odd}
\label{Figure: Verifier's protocol for achieving perfect completeness for odd m}
\end{figure}

For the completeness, suppose that $x$ is in $A_\yes$.

As the original system was perfectly rewindable,
there exists a ${(2r+1)}$-message quantum prover~$P$ in the original system
such that the maximum probability that $V$ accepts $x$ when communicating with this $P$ is exactly $1/2$,
where the maximum is taken over all possible initial states of $P$.
Let $\sfP$ be the quantum register consisting of the private qubits of this $P$,
and let $P_{x, j}$ be the $j$th transformation of $P$,
for each ${j \in \{1, \ldots, r(\abs{x})+1 \}}$,
acting over ${(\sfM, \sfP)}$.
Let $\ket{\psi_x^\ast}$ be an optimal initial state in ${(\sfM, \sfP)}$
with which $P$ achieves the accepting probability~$1/2$
(note that $P$ possesses the message register~$\sfM$ at the beginning of the protocol, and that there always exists an optimal initial state that is pure).

Denote the Hilbert spaces associated with $\sfV$, $\sfM$, and $\sfP$ by $\calV$, $\calM$, and $\calP$, respectively.
Since the first action is done by $P$ in this original proof system,
one can assume without loss of generality that ${P_{x,1} = I_{\calM \tensor \calP}}$
(i.e., the first transformation of $P$ may be regarded as a part of preparing the initial state).
%% (i.e., applying a nontrivial $P_{x,1}$ to an initial state~$\ket{\psi_x}$
%% is equivalent to just preparing ${P_{x,1} \ket{\psi_x}}$.
Taking this into account,
define the unitary transformation~$Q_x$ acting over ${(\sfV, \sfM, \sfP)}$ by
\[
Q_x
=
\bigl( V_{x, r(\abs{x})+1} \tensor I_\calP \bigr)
\bigl( I_\calV \tensor P_{x, r(\abs{x})+1} \bigr)
\cdots
\bigl( V_{x, 2} \tensor I_\calP \bigr)
\bigl( I_\calV \tensor P_{x, 2} \bigr)
\bigl( V_{x, 1} \tensor I_\calP \bigr),
%% (V_{x, 2} \tensor I_\calP)
%% (I_\calV \tensor P_{x, 2})
%% (V_{x, 1} \tensor I_\calP).
\]
and further define the Hermitian matrix~$M_x$ by
\[
M_x = \Pi_\init \conjugate{Q_x} \Pi_\acc Q_x \Pi_\init,
\]
where $\Pi_\init$ is the projection onto the subspace spanned by states
in which all the qubits in $\sfV$ are in state~$\ket{0}$,
and $\Pi_\acc$ is that onto the subspace spanned by accepting states of the original system.
Then the quantum state~$\ket{\phi_x^\ast}$ in ${(\sfV, \sfM, \sfP)}$ defined as
${\ket{\phi_x^\ast} = \ket{0}_\sfV \tensor \ket{\psi_x^\ast}_{(\sfM, \sfP)}}$
is the eigenvector of $M_x$ with its corresponding eigenvalue~$1/2$,
since
\[
\max_{\ket{\phi} \in \calV \tensor \calM \tensor \calP} \bra{\phi} M_x \ket{\phi}
=
\max_{\ket{\psi} \in \calM \tensor \calP} \norm{\Pi_\acc Q_x (\ket{0} \tensor \ket{\psi})}^2
=
\norm{\Pi_\acc Q_x (\ket{0} \tensor \ket{\psi_x^\ast})}^2
=
\bra{\phi_x^\ast} M_x \ket{\phi_x^\ast}
=
\frac{1}{2}.
\]

Now, with a ${(2r+1)}$-message quantum prover~$R$ in the constructed system
who prepares the state~${\ket{\xi_x^\ast} = \bigl( \conjugate{V_{x, r(\abs{x})+1}} \tensor I_\calP \bigr) Q_x \ket{\phi_x^\ast}}$
%% \[
%% \ket{\xi_x^\ast}
%% =
%% \bigl( \conjugate{V_{x, r(\abs{x})+1}} \tensor I_\calP \bigr)
%% Q_x
%% \ket{\phi_x^\ast}
%% =
%% \bigl( I_\calV \tensor P_{x, r(\abs{x})+1} \bigr)
%% \bigl( \conjugate{V_{x, r(\abs{x})}} \tensor I_\calP \bigr)
%% \bigl( I_\calV \tensor P_{x, r(\abs{x})} \bigr)
%% \cdots
%% \bigl( V_{x, 2} \tensor I_\calP \bigr)
%% \bigl( I_\calV \tensor P_{x, 2} \bigr)
%% \bigl( V_{x, 1} \tensor I_\calP \bigr)
%% \bigl( \ket{0}_\sfV \tensor \ket{\psi_x^\ast}_{(\sfM, \sfP)} \bigr)
%% \]
in ${(\sfV, \sfM, \sfP)}$ as an initial state
and applies ${R_{x,1} = I_{\calM \tensor \calP}}$ and ${R_{x,j} = P_{x, r(\abs{x})-j+3}}$ for each ${j \in \{2, \ldots, r(\abs{x})+1 \}}$,
the constructed protocol may be viewed as
performing the \textsc{Modified Reflection Procedure}
with its underlying quantum register~${\sfQ = (\sfV, \sfM, \sfP)}$,
%% with its underlying Hilbert space~${\calH = \calV \tensor \calM \tensor \calP}$,
%% quantum register~${\sfQ = (\sfV, \sfM, \sfP)}$,
unitary transformation
\[
U = \bigl( \conjugate{V_{x, r(\abs{x})+1}} \tensor I_\calP \bigr) Q_x,
\]
and projection operators
\begin{align*}
\Delta_0
&
= \Pi_\init,
\\
\Pi_0
&
= \bigl( \conjugate{V_{x, r(\abs{x})+1}} \tensor I_\calP \bigr) \Pi_\acc \bigl( V_{x, r(\abs{x})+1} \tensor I_\calP \bigr).
\end{align*}
% unitary transformation~${U = \bigl( \conjugate{V_{x, r(\abs{x})+1}} \tensor I_\calP \bigr) Q_x}$,
% and projection operators~${\Delta_0 = \Pi_\init}$
% and ${\Pi_0 = \bigl( \conjugate{V_{x, r(\abs{x})+1}} \tensor I_\calP \bigr) \Pi_\acc \bigl( V_{x, r(\abs{x})+1} \tensor I_\calP \bigr)}$.
As the associated Hermitian operator
\[
M = \Delta_0 \conjugate{U} \Pi_0 U \Delta_0 = \Pi_\init \conjugate{Q_x} \Pi_\acc Q_x \Pi_\init = M_x
\]
%% As the associated Hermitian operator~${M = \Delta_0 \conjugate{U} \Pi_0 U \Delta_0 = \Pi_\init \conjugate{Q_x} \Pi_\acc Q_x \Pi_\init = M_x}$
has an eigenvalue~$1/2$
%% with its corresponding eigenvector~$\ket{\phi_x^\ast}$,
with its corresponding eigenvector~${\ket{\phi_x^\ast} = \ket{0}_\sfV \tensor \ket{\psi_x^\ast}_{(\sfM, \sfP)}}$,
from Proposition~\ref{Proposition: completeness of Modified Reflection Procedure},
the protocol results in acceptance with certainty with this prover~$R$ and the initial state~${\ket{\xi_x^\ast} = \bigl( \conjugate{V_{x, r(\abs{x})+1}} \tensor I_\calP \bigr) Q_x \ket{\phi_x^\ast} = U \ket{\phi_x^\ast}}$,
which shows the perfect completeness.

Now for the soundness, suppose that $x$ is in $A_\no$.

Let $R$ be any ${(2r+1)}$-message quantum prover of the constructed system,
and let $\sfR$ be the quantum register consisting of the private qubits of $R$.
Suppose that $R$ applies the unitary transformation~$R_{x, j}$ to the state in ${(\sfM, \sfR)}$
as the $j$th transformation of $R$,
for each ${j \in \{1, \ldots, r(\abs{x})+1 \}}$.
%% Suppose that $R$ prepares a state~$\ket{\psi_x}$ in ${(\sfV, \sfM, \sfR)}$ as an initial state,
%% and applies the unitary transformation~$R_{x, j}$ to the state in ${(\sfM, \sfR)}$
%% as the $j$th transformation of $R$,
%% for each ${j \in \{1, \ldots, r(\abs{x})+1 \}}$
%% (note that one can assume without loss of generality that the initial state is pure,
%% for the maximum accepting probability with a fixed set of transformations~$\{R_{x,j}\}$
%% is always achievable with a pure initial state).

Define the unitary transformation~$Q_x$ acting over ${(\sfV, \sfM, \sfR)}$ by
\[
Q_x
=
\bigl( V_{x, r(\abs{x})+1} \tensor I_\calR \bigr)
\bigl( I_\calV \tensor \conjugate{R_{x, 2}} \bigr)
\cdots
\bigl( V_{x, 2} \tensor I_\calR \bigr)
\bigl( I_\calV \tensor \conjugate{R_{x, r(\abs{x})+1}} \bigr)
\bigl( V_{x, 1} \tensor I_\calR \bigr),
\]
where $\calR$ is the Hilbert space associated with the register~$\sfR$.
Then the constructed protocol may be viewed as
performing the \textsc{Modified Reflection Procedure}
with its underlying quantum register~${\sfQ = (\sfV, \sfM, \sfR)}$,
%% with its underlying Hilbert space~${\calH = \calV \tensor \calM \tensor \calR}$,
%% quantum register~${\sfQ = (\sfV, \sfM, \sfR)}$,
unitary transformation
\[
U = \bigl( \conjugate{V_{x, r(\abs{x})+1}} \tensor I_\calR \bigr) Q_x,
\]
and projection operators
\begin{align*}
\Delta_0 &= \Pi_\init,\\
\Pi_0 &= \bigl( \conjugate{V_{x, r(\abs{x})+1}} \tensor I_\calR \bigr) \Pi_\acc \bigl( V_{x, r(\abs{x})+1} \tensor I_\calR \bigr).
\end{align*}
% unitary transformation~${U = \bigl( \conjugate{V_{x, r(\abs{x})+1}} \tensor I_\calR \bigr) Q_x}$,
% and projection operators~${\Delta_0 = \Pi_\init}$
% and ${\Pi_0 = \bigl( \conjugate{V_{x, r(\abs{x})+1}} \tensor I_\calR \bigr) \Pi_\acc \bigl( V_{x, r(\abs{x})+1} \tensor I_\calR \bigr)}$.
The associated Hermitian operator of this \textsc{Modified Reflection Procedure} is given by
\[
M_x = \Delta_0 \conjugate{U} \Pi_0 U \Delta_0 = \Pi_\init \conjugate{Q_x} \Pi_\acc Q_x \Pi_\init.
\]
Consider the following ${(2r+1)}$-message quantum prover~$P'$ in the original system:
$P'$ uses $\sfR$ as a register consisting of his/her private qubits,
and applies $I_{\calM \tensor \calR}$ as his/her first transformation,
and $\conjugate{R_{x, r(\abs{x}) - j + 3}}$ as his/her $j$th transformation,
for ${j \in \{2, \ldots, r(\abs{x})+1\}}$.
Then, from the soundness property of the original system,
no matter which state $P'$ initially prepares,
the accepting probability is at most ${\frac{1}{2} - \frac{c(\abs{x}) - s(\abs{x})}{4}}$,
which implies that all the eigenvalues of $M_x$ is at most ${\frac{1}{2} - \frac{c(\abs{x}) - s(\abs{x})}{4}}$.
Hence, from Proposition~\ref{Proposition: soundness of Modified Reflection Procedure},
the constructed protocol results in rejection with probability at least ${\frac{(c(\abs{x}) - s(\abs{x}))^2}{16}}$,
which ensures the soundness~${1 - \frac{(c-s)^2}{16}}$.

Finally, the protocol given in Figure~\ref{Figure: Verifier's protocol for achieving perfect completeness for odd m}
slightly deviates from the standard form of quantum interactive proof systems
in that the length of the first message from a prover is different from
the lengths of other messages,
which may be easily modified into a standard-form system
that has exactly the same number of messages and completeness and soundness parameters.
\end{proof}

Now we prove Theorem~\ref{Theorem: QIP(m) is in QIP_1(m+1)}.
The proof is essentially the same as the proof of Theorem~\ref{Theorem: QIP(m) is in QIP_1(m) for odd m},
and we analyze the case where the number of messages is even.

\begin{proof}[Proof of Theorem~\ref{Theorem: QIP(m) is in QIP_1(m+1)}]
Let ${A = (A_\yes, A_\no)}$ be in ${\QIP(m,c,s)}$.
Then from Lemma~\ref{Lemma: making QIP perfectly rewindable},
$A$ has a perfectly rewindable $m$-message quantum interactive proof system
with soundness ${\frac{1}{2} - \frac{c-s}{4}}$.
Let $V$ be the verifier of this perfectly rewindable $m$-message quantum interactive proof system.
We construct an ${(m+1)}$-message quantum verifier~$W$
of a new quantum interactive proof system for $A$.
The construction is essentially the same as that in the proof of Theorem~\ref{Theorem: QIP(m) is in QIP_1(m) for odd m}.

Fix an input~$x$.
Suppose that ${m(\abs{x}) \geq 2}$ is even,
and write ${m(\abs{x}) = 2 r(\abs{x})}$ for some ${r(\abs{x}) \in \Natural}$ (the proof of Theorem~\ref{Theorem: QIP(m) is in QIP_1(m) for odd m}
already shows the case where ${m(\abs{x})}$ is odd).
The exact protocol is described in Figure~\ref{Figure: Verifier's protocol for achieving perfect completeness},
where the only difference from the protocol in Figure~\ref{Figure: Verifier's protocol for achieving perfect completeness for odd m}
lies in the condition of judging whether the state is initialized or not
-- now a state is a legal initial state
only when all the qubits in both of $\sfV$ and $\sfM$ must be in state~$\ket{0}$.
Notice that the number of messages in this system is indeed
${1 + 1 + 2(r(\abs{x}) - 1) + 1 = 2r(\abs{x}) + 1 = m(\abs{x}) + 1}$.

\begin{figure}[t!]
\begin{protocol*}{Verifier's Protocol for Achieving Perfect Completeness (Even-Number-Message Case)}
\begin{step}
\item
  Receive quantum registers~$\sfV$~and~$\sfM$.
\item
  Choose ${b \in \Binary}$ uniformly at random.
  If ${b=0}$, move to the \textsc{Reflection Test}
  described in Step~\ref{Reflection Test Two},
  while if ${b=1}$, move to the \textsc{Invertibility Test}
  described in Step~\ref{Invertibility Test Two}.
\item
  (\textsc{Reflection Test})
  \begin{step}
  \item
    Apply $V_{x, r(\abs{x})+1}$ to the state in ${(\sfV, \sfM)}$.
    Perform a phase-flip (i.e., multiply $-1$ in phase)
    if the content of ${(\sfV, \sfM)}$ corresponds to an accepting state of the original system.
    Apply $\conjugate{V_{x, r(\abs{x})+1}}$ to the state in ${(\sfV, \sfM)}$,
    and send $\sfM$ to the prover.
  \item
    For ${j = r(\abs{x})}$ down to $2$, do the following:\\
    Receive $\sfM$ from the prover.
    Apply $\conjugate{V_{x, j}}$ to the state in ${(\sfV, \sfM)}$,
    and send $\sfM$ to the prover.
  \item
    Receive $\sfM$ from the prover.
    Apply $\conjugate{V_{x, 1}}$ to the state in ${(\sfV, \sfM)}$.
    Reject if all the qubits in ${(\sfV, \sfM)}$ are in state~$\ket{0}$,
    and accept otherwise.
  \end{step}
  \label{Reflection Test Two}
\item
  (\textsc{Invertibility Test})
  \begin{step}
  \item
    Send $\sfM$ to the prover.
  \item
    For ${j = r(\abs{x})}$ down to $2$, do the following:\\
    Receive $\sfM$ from the prover.
    Apply $\conjugate{V_{x, j}}$ to the state in ${(\sfV, \sfM)}$,
    and send $\sfM$ to the prover.
  \item
    Receive $\sfM$ from the prover.
    Apply $\conjugate{V_{x, 1}}$ to the state in ${(\sfV, \sfM)}$.
    Accept if all the qubits in ${(\sfV, \sfM)}$ are in state~$\ket{0}$,
    and reject otherwise.
  \end{step}
  \label{Invertibility Test Two}
\end{step}
\end{protocol*}
\caption{Verifier's protocol for achieving perfect completeness with ${m(\abs{x}) = 2r(\abs{x})}$.}
\label{Figure: Verifier's protocol for achieving perfect completeness}
\end{figure}

The analysis on this protocol is essentially the same as that in the proof of 
Theorem~\ref{Theorem: QIP(m) is in QIP_1(m) for odd m},
and is omitted.
\end{proof}

% ---------------------------------------------------------------------------
%   QMIP(m) is in QMIP_1(m+1)
% ---------------------------------------------------------------------------

\subsection{Cases with Quantum Multi-Prover Interactive Proofs}
\label{Subsection: QMIP(m) is in QMIP_1(m+1)}

With essentially the same arguments discussed in this section, 
we can show similar properties even for quantum multi-prover interactive proof systems.
The model of quantum multi-prover interactive proofs we use is that in the most general setting
(i.e., both of a verifier and provers use quantum computation and communications,
and provers can share arbitrary entanglement of arbitrarily large size).
Let ${\QMIP(k,m,c,s)}$ be the class of problems having $m$-turn quantum $k$-prover interactive proof systems with completeness~$c$ and soundness~$s$.
See Ref.~\cite{KemKobMatVid09CC} for rigorous definitions of the quantum multi-prover model and resulting complexity classes.
Here we give only the statements of theorems,
as proofs of those theorems are essentially same as Theorems~\ref{Theorem: QIP(m) is in QIP_1(m+1)}~and~\ref{Theorem: QIP(m) is in QIP_1(m) for odd m}.
Note that these theorems give a more communication-efficient way of
achieving perfect completeness in quantum multi-prover interactive proofs
than the original method presented in Ref.~\cite{KemKobMatVid09CC},
where the number of turns increases by a factor of three.

\begin{theorem}
For any polynomially bounded functions~$\function{k,m}{\Nonnegative}{\Natural}$
and polynomial-time computable functions~$\function{c,s}{\Nonnegative}{[0,1]}$
satisfying ${m \geq 2}$ and ${c - s \geq 1/p}$ for some polynomially bounded function~$\function{p}{\Nonnegative}{\Natural}$,
\[
\QMIP(k, m, c, s) \subseteq \QMIP \Bigl( k, m+1, 1, 1 - \frac{(c-s)^2}{16} \Bigr).
\]
\label{Theorem: QMIP(k,m) is in QMIP_1(k,m+1)}
\end{theorem}

\begin{theorem}
For any polynomially bounded function~$\function{k}{\Nonnegative}{\Natural}$,
polynomially bounded odd-valued function~$\function{m}{\Nonnegative}{2\Natural + 1}$,
and polynomial-time computable functions~$\function{c,s}{\Nonnegative}{[0,1]}$
satisfying ${m \geq 3}$ and ${c - s \geq 1/p}$ for some polynomially bounded function~$\function{p}{\Nonnegative}{\Natural}$,
\[
\QMIP(k, m, c, s) \subseteq \QMIP \Bigl(k, m, 1, 1 - \frac{(c-s)^2}{16} \Bigr).
\]
\label{Theorem: QMIP(k,m) is in QMIP_1(k,m) for odd m}
\end{theorem}

\begin{remark}
Similar to the single-prover case, in fact,
it is sufficient for the claims in Theorems~\ref{Theorem: QMIP(k,m) is in QMIP_1(k,m+1)}~and~\ref{Theorem: QMIP(k,m) is in QMIP_1(k,m) for odd m}
that the functions~$c$ and $s$ satisfy
${c - s \geq 2^{-p}}$ for some polynomially bounded function~$\function{p}{\Nonnegative}{\Natural}$.
\end{remark}

% ---------------------------------------------------------------------------
%   Acknowledgements
% ---------------------------------------------------------------------------

\subsection*{Acknowledgements}

The authors are grateful to Richard Cleve for useful discussions,
Attila Pereszl\'enyi for pointing out an error in Lemma~\ref{Lemma: closeness to the mixture of 2-fold products of CJ states}
in the earlier versions of this paper,
%% Attila Pereszl\'enyi for pointing out an error in the earlier version of this paper,
and an anonymous reviewer 
for helpful comments on improving some parameters in the proof of Theorem~\ref{Theorem: QMA is in 2^k-EPR-QMA(1,s)}. 
This work is supported by
the Grant-in-Aid for Scientific Research~(A)~No.~24240001 of the Japan Society for the Promotion of Science
and the Grant-in-Aid for Scientific Research on Innovative Areas~No.~24106009 of
the Ministry of Education, Culture, Sports, Science and Technology in Japan.
HK also acknowledges support from
the Grant-in-Aid for Scientific Research~(B)~No.~21300002 of the Japan Society for the Promotion of Science.
HN also acknowledges support from
the Grant-in-Aids for Scientific Research~(A)~Nos.~21244007~and~23246071 of the Japan Society for the Promotion of Science
and the Grant-in-Aid for Young Scientists~(B)~No.~22700014 of
the Ministry of Education, Culture, Sports, Science and Technology in Japan.

%% HK is partially supported by
%% the Grant-in-Aids for Scientific Research~(A)~No.~24240001 and (B)~No.~21300002 of the Japan Society for the Promotion of Science,
%% and the Grant-in-Aid for Scientific Research on Innovative Areas~No.~24106009 of
%% the Ministry of Education, Culture, Sports, Science and Technology in Japan.
%% FLG is partially supported by...
%% 
%% HN is partially supported by
%% the Grant-in-Aids for Scientific Research~(A)~Nos.~21244007,~23246071,~and~2424001 of the Japan Society for the Promotion of Science
%% and the Grant-in-Aid for Scientific Research on Innovative Areas~No.~24106009
%% and the Grant-in-Aid for Young Scientists~(B)~No.~22700014 of
%% the Ministry of Education, Culture, Sports, Science and Technology in Japan.

%% \clearpage

% ---------------------------------------------------------------------------
%   References
% ---------------------------------------------------------------------------

%% \bibliographystyle{HKhalpha}
%% \bibliography{KobLeGNis12manuscript}

\begin{thebibliography}{KKMV09}

\bibitem[AALV09]{AhaAraLanVaz09STOC}
Dorit Aharonov, Itai Arad, Zeph Landau, and Umesh Vazirani.
\newblock The detectability lemma and quantum gap amplification [extended
  abstract].
\newblock In {\em Proceedings of the 41st Annual ACM Symposium on Theory of
  Computing}, pages 417--426, 2009.

\bibitem[AALV11]{AhaAraLanVaz11FOCS}
Dorit Aharonov, Itai Arad, Zeph Landau, and Umesh Vazirani.
\newblock The {1D} area law and the complexity of quantum states: A combinatorial
  approach.
\newblock In {\em 52nd Annual IEEE Symposium on Foundations of Computer
  Science}, pages 324--333, 2011.

\bibitem[Aar09]{Aar09QIC}
Scott Aaronson.
\newblock On perfect completeness for {QMA}.
\newblock {\em Quantum Information and Computation}, 9(1--2):0081--0089, 2009.

\bibitem[AE11]{AhaEld11FOCS}
Dorit Aharonov and Lior Eldar.
\newblock On the complexity of commuting local {Hamiltonians}, and tight
  conditions for topological order in such systems.
\newblock In {\em 52nd Annual IEEE Symposium on Foundations of Computer
  Science}, pages 334--343, 2011.

\bibitem[Aha03]{Aha03arXiv}
Dorit Aharonov.
\newblock A simple proof that {Toffoli} and {Hadamard} are quantum universal.
\newblock arXiv.org e-Print archive,
  \href{http://arxiv.org/abs/quant-ph/0301040}{arXiv:quant-ph/0301040}, 2003.

\bibitem[AKN98]{AhaKitNis98STOC}
Dorit Aharonov, Alexei Kitaev, and Noam Nisan.
\newblock Quantum circuits with mixed states.
\newblock In {\em Proceedings of the Thirtieth Annual ACM Symposium on Theory
  of Computing}, pages 20--30, 1998.

\bibitem[ALM{\etalchar{+}}98]{AroLunMotSudSze98JACM}
Sanjeev Arora, Carsten Lund, Rajeev Motwani, Madhu Sudan, and Mario Szegedy.
\newblock Proof verification and the hardness of approximation problems.
\newblock {\em Journal of the ACM}, 45(3):501--555, 1998.

\bibitem[AS98]{AroSaf98JACM}
Sanjeev Arora and Shmuel Safra.
\newblock Probabilistic checking of proofs: A new characterization of {$\NP$}.
\newblock {\em Journal of the ACM}, 45(1):70--122, 1998.

\bibitem[Bab85]{Bab85STOC}
L\'aszl\'o Babai.
\newblock Trading group theory for randomness.
\newblock In {\em Proceedings of the Seventeenth Annual ACM Symposium on Theory
  of Computing}, pages 421--429, 1985.

\bibitem[Bra06]{Bra06arXiv}
Sergey Bravyi.
\newblock Efficient algorithm for a quantum analogue of {2-SAT}.
\newblock arXiv.org e-Print archive,
  \href{http://arxiv.org/abs/quant-ph/0602108}{arXiv:quant-ph/0602108}, 2006.

\bibitem[BSW11]{BeiShoWat11ToC}
Salman Beigi, Peter Shor, and John Watrous.
\newblock Quantum interactive proofs with short messages.
\newblock {\em Theory of Computing}, 7:101--117~(Article~7), 2011.

\bibitem[Cho75]{Cho75LAA}
Man-Duen Choi.
\newblock Completely positive linear maps on complex matrices.
\newblock {\em Linear Algebra and its Applications}, 10(3):285--290, 1975.

\bibitem[CKMR07]{ChrKonMitRen07CMP}
Matthias Christandl, Robert K\"onig, Graeme Mitchison, and Renato Renner.
\newblock One-and-a-half quantum {de Finetti} theorems.
\newblock {\em Communications in Mathematical Physics}, 273(2):473--498, 2007.

\bibitem[ESY84]{EveSelYac84ICtrl}
Shimon Even, Alan~L. Selman, and Yacov Yacobi.
\newblock The complexity of promise problems with applications to public-key
  cryptography.
\newblock {\em Information and Control}, 61(2):159--173, 1984.

\bibitem[Gro96]{Gro96STOC}
Lov~K. Grover.
\newblock A fast quantum mechanical algorithm for database search.
\newblock In {\em Proceedings of the Twenty-Eighth Annual ACM Symposium on the
  Theory of Computing}, pages 212--219, 1996.

\bibitem[GSU13]{GhaSikUpa13QIC}
Sevag Gharibian, Jamie Sikora, and Sarvagya Upadhyay.
\newblock {QMA} variants with polynomially many provers.
\newblock {\em Quantum Information and Computation}, 13(1--2):0135--0157, 2013.

\bibitem[Gut09]{Gut09PhD}
Gustav Gutoski.
\newblock {\em Quantum Strategies and Local Operations}.
\newblock PhD thesis, David R. Cheriton School of Computer Science, University
  of Waterloo, 2009.
\newblock \href{http://arxiv.org/abs/1003.0038}{arXiv:1003.0038~[quant-ph]}.

\bibitem[GZ11]{GolZuc11LNCS}
Oded Goldreich and David Zuckerman.
\newblock Another proof that {$\BPP \subseteq \PH$} (and more).
\newblock In Oded Goldreich, editor, {\em Studies in Complexity and
  Cryptography, Miscellanea on the Interplay between Randomness and
  Computation}, volume 6650 of {\em Lecture Notes in Computer Science}, pages
  40--53. Springer-Verlag, 2011.
\newblock Electronic Colloquium on Computational Complexity, Report TR97-045,
  1997.

\bibitem[Jam72]{Jam72RMP}
Andrzej Jamio{\l}kowski.
\newblock Linear transformations which preserve trace and positive
  semidefiniteness of operators.
\newblock {\em Reports on Mathematical Physics}, 3(4):275--278, 1972.

\bibitem[JJUW11]{JaiJiUpaWat11JACM}
Rahul Jain, Zhengfeng Ji, Sarvagya Upadhyay, and John Watrous.
\newblock {${\QIP = \PSPACE}$}.
\newblock {\em Journal of the ACM}, 58(6):Article~30, 2011.

\bibitem[JKNN12]{JorKobNagNis12QIC}
Stephen~P. Jordan, Hirotada Kobayashi, Daniel Nagaj, and Harumichi Nishimura.
\newblock Achieving perfect completeness in classical-witness quantum
  {Merlin-Arthur} proof systems.
\newblock {\em Quantum Information and Computation}, 12(5--6):0461--0471, 2012.

\bibitem[Kit99]{Kit99AQIP}
Alexei~{\relax{Yu}}. Kitaev.
\newblock Quantum {NP}.
\newblock Talk at the 2nd Workshop on Algorithms in Quantum Information
  Processing, January 1999.
%% \newblock Talk at the 2nd Workshop on Algorithms in Quantum Information
%%   Processing, DePaul University, Chicago, January 1999.

\bibitem[KKMV09]{KemKobMatVid09CC}
Julia Kempe, Hirotada Kobayashi, Keiji Matsumoto, and Thomas Vidick.
\newblock Using entanglement in quantum multi-prover interactive proofs.
\newblock {\em Computational Complexity}, 18(2):273--307, 2009.

\bibitem[Kni96]{Kni96TR}
Emanuel Knill.
\newblock Quantum randomness and nondeterminism.
\newblock Technical Report LAUR-96-2186, Los Alamos National Laboratory, 1996.
\newblock arXiv.org e-Print archive,
  \href{http://arxiv.org/abs/quant-ph/9610012}{arXiv:quant-ph/9610012}.

\bibitem[KR05]{KonRen05JMP}
Robert K\"onig and Renato Renner.
\newblock A {de~Finetti} representation for finite symmetric quantum states.
\newblock {\em Journal of Mathematical Physics}, 46(12):122108, 2005.

\bibitem[KSV02]{KitSheVya02Book}
Alexei~{\relax{Yu}}. Kitaev, Alexander~H. Shen, and Mikhail~N. Vyalyi.
\newblock {\em Classical and Quantum Computation}, volume~47 of {\em Graduate
  Studies in Mathematics}.
\newblock American Mathematical Society, 2002.

\bibitem[Kup09]{Kup09arXiv}
Greg Kuperberg.
\newblock How hard is it to approximate the {Jones} polynomial?
\newblock arXiv.org e-print archive,
  \href{http://arxiv.org/abs/0908.0512}{arXiv:0908.0512~[quant-ph]}, 2009.

\bibitem[KW00]{KitWat00STOC}
Alexei Kitaev and John Watrous.
\newblock Parallelization, amplification, and exponential time simulation of
  quantum interactive proof systems.
\newblock In {\em Proceedings of the Thirty-Second Annual ACM Symposium on
  Theory of Computing}, pages 608--617, 2000.

\bibitem[MW05]{MarWat05CC}
Chris Marriott and John Watrous.
\newblock Quantum {Arthur-Merlin} games.
\newblock {\em Computational Complexity}, 14(2):122--152, 2005.

\bibitem[NC00]{NieChu00Book}
Michael~A. Nielsen and Isaac~L. Chuang.
\newblock {\em Quantum Computation and Quantum Information}.
\newblock Cambridge University Press, 2000.

\bibitem[NS03]{NaySho03PRA}
Ashwin Nayak and Peter Shor.
\newblock Bit-commitment-based quantum coin flipping.
\newblock {\em Physical Review A}, 67(1):012304, 2003.

\bibitem[NWZ09]{NagWocZha09QIC}
Daniel Nagaj, Pawel Wocjan, and Yong Zhang.
\newblock Fast amplification of {QMA}.
\newblock {\em Quantum Information and Computation}, 9(11--12):1053--1068,
  2009.

\bibitem[Shi02]{Shi02QIC}
Yaoyun Shi.
\newblock Both {Toffoli} and {Controlled-NOT} need little help to do universal
  quantum computing.
\newblock {\em Quantum Information and Computation}, 3(1):084--092, 2002.

\bibitem[Sho96]{Sho96FOCS}
Peter~W. Shor.
\newblock Fault-tolerant quantum computation.
\newblock In {\em 37th Annual Symposium on Foundations of Computer Science},
  pages 56--65, 1996.

\bibitem[SR02]{SpeRud02PRA}
Robert~W. Spekkens and Terry Rudolph.
\newblock Degrees of concealment and bindingness in quantum bit commitment
  protocols.
\newblock {\em Physical Review A}, 65(1):012310, 2002.

\bibitem[Vya03]{Vya03ECCC}
Mikhail~N. Vyalyi.
\newblock ${\QMA=\PP}$ implies that {$\PP$} contains {$\PH$}.
\newblock Electronic Colloquium on Computational Complexity, Report TR03-021,
  2003.

\bibitem[Wat00]{Wat00FOCS}
John Watrous.
\newblock Succinct quantum proofs for properties of finite groups.
\newblock In {\em 41st Annual Symposium on Foundations of Computer Science},
  pages 537--546, 2000.

\bibitem[Wat03]{Wat03TCS}
John Watrous.
\newblock {$\PSPACE$} has constant-round quantum interactive proof systems.
\newblock {\em Theoretical Computer Science}, 292(3):575--588, 2003.

\bibitem[Wat09a]{Wat09ECSS}
John Watrous.
\newblock Quantum computational complexity.
\newblock In Robert~A. Meyers, editor, {\em Encyclopedia of Complexity and
  Systems Science}, pages 7174--7201. Springer-Verlag, 2009.

\bibitem[Wat09b]{Wat09SIComp}
John Watrous.
\newblock Zero-knowledge against quantum attacks.
\newblock {\em SIAM Journal on Computing}, 39(1):25--58, 2009.

\bibitem[ZF87]{ZacFur87FSTTCS}
Stathis Zachos and Martin F\"urer.
\newblock Probabilistic quantifiers vs. distrustful adversaries.
\newblock In {\em Foundations of Software Technology and Theoretical Computer
  Science, Seventh Conference}, volume 287 of {\em Lecture Notes in Computer
  Science}, pages 443--455, 1987.

\end{thebibliography}

\newcommand{\etalchar}[1]{$^{#1}$}
\providecommand{\noopsort}[1]{}

\end{document}